\def\fullversion{1}

\ifnum\fullversion=0

\def\year{2015}
\documentclass[letterpaper]{article}
\usepackage{aaai}
\usepackage{times}
\usepackage{helvet}
\usepackage{courier}
\usepackage{enumitem}

\else 

\documentclass[11pt]{article}
\usepackage[top=1.in, bottom=1.in, left=1.in, right=1.in]{geometry}

\fi 

\usepackage{amstext}
\usepackage{amsmath}
\usepackage{amsthm}
\usepackage{amssymb}
\usepackage{bm}
\usepackage{wrapfig}
\usepackage{amsfonts, amsmath}
\usepackage[utf8]{inputenc}
\usepackage{graphicx}
\usepackage{bbm, comment}
\usepackage{subfigure}
\usepackage{color}
\usepackage{float}

\usepackage[numbers]{natbib}

\newcommand{\E}{\mathbb{E}}

\newtheorem{theorem}{Theorem}[section]
\newtheorem{lemma}[theorem]{Lemma}

\newtheorem{claim}[theorem]{Claim}
\newtheorem{assumption}[theorem]{Assumption}
\newtheorem{remark}[theorem]{Remark}

\newtheorem{corollary}[theorem]{Corollary}
\newtheorem{example}[theorem]{Example}
\newtheorem{proposition}[theorem]{Proposition}
\newtheorem*{claim*}{Claim}
\newtheorem{observation}[theorem]{Observation}
\newtheorem*{observation*}{Observation}

\theoremstyle{definition}
\newtheorem{definition}[theorem]{Definition}

\newcommand{\tru}[0]{\mathbf{T}}

\usepackage{xcolor}

\newcommand\gs[1]{{}}
\newcommand\yk[1]{{}}

\newcommand{\denselist}{\itemsep 0pt\parsep=1pt\partopsep 0pt}
\newcommand{\bitem}{\begin{itemize}\denselist}
\newcommand{\eitem}{\end{itemize}}
\newcommand{\benum}{\begin{enumerate}\denselist}
\newcommand{\eenum}{\end{enumerate}}
\newcommand{\bdesc}{\begin{description}\denselist}
\newcommand{\edesc}{\end{description}}

\begin{document}

\title{Equilibrium Selection in Information Elicitation without Verification via Information Monotonicity}
\author{Yuqing Kong\\ University of Michigan \and Grant Schoenebeck\\ University of Michigan}
\date{}

%\institute{Computer Science \& Engineering, University of Michigan}

\maketitle

% =============================================================================
% =============================================================================

\begin{abstract}
    Peer-prediction is a mechanism which elicits privately-held, non-variable information from self-interested agents---formally, truth-telling is a strict Bayes Nash equilibrium of the mechanism. The original Peer-prediction mechanism suffers from two main limitations: (1) the mechanism must know the ``common prior'' of agents' signals; (2) additional undesirable and non-truthful equilibria exist which often have a greater expected payoff than the truth-telling equilibrium. A series of results has successfully weakened the known common prior assumption. However, the equilibrium multiplicity issue remains a challenge.% both because equilibrium analysis is complex and because impossibility results imply that truth-telling cannot always have strictly higher payoff than any other equilibrium.

In this paper, we address the above two problems. In the setting where a common prior exists but is not known to the mechanism we show (1) a general negative result applying to a large class of mechanisms showing truth-telling can never pay strictly more in expectation than a particular set of equilibria where agents collude to ``relabel'' the signals and tell the truth after relabeling signals; (2) provide a mechanism that has no information about the common prior but where truth-telling pays as much in expectation as any relabeling equilibrium and pays strictly more than any other symmetric equilibrium; (3) moreover in our mechanism, if the number of agents is sufficiently large, truth-telling pays similarly to any equilibrium close to a ``relabeling'' equilibrium and pays strictly more than any equilibrium that is not close to a relabeling equilibrium.

%(1) identify a particular set of equilibria, which we call permutation equilibria, which exist in every information elicitation mechanism; and show that for any information elicitation mechanism and for any permutation equilibrium, there exists a setting where that permutation equilibrium is paid at least as much as truth-telling; (2) propose an information elicitation mechanism where truth-telling is a Bayesian Nash equilibrium for any finite set of signals and any number of agents greater than 3, and where (a) truth-telling is paid strictly more in expectation than any other symmetric equilibrium that is not a permutation equilibrium, which have an expected payoff equal to the truth-telling equilibrium; (b) the amount by which any asymmetric strategy can be paid more than truth-telling goes to zero as the number of agents grows; (c) any equilibrium that has expected payoff close to truth-telling must be ``close'' to a permutation equilibrium.

\end{abstract}

%\everypar{\looseness=-1}
%\linepenalty=10000 
%\clubpenalty=10000

\section{Introduction}

User feedback requests (e.g Ebay's reputation system and the innumerable survey requests in one's email inbox) are increasingly prominent and important. However, the overwhelming number of requests can lead to low participation rates, which in turn may yield unrepresentative samples. To encourage participation, a system can reward people for answering requests. But this may cause perverse incentives: some people may answer a large of number of questions simply for the reward and without making any attempt to answer accurately. In this case, the reviews the system obtains may be inaccurate and meaningless. Moreover, people may be motivated to lie when they face a potential loss of privacy or can benefit in the future by lying now. 

It is thus important to develop systems that motivate honesty. If we can verify the information people provide in the future (e.g via prediction markets), we can motivate honesty via this future verification. However, sometimes we need to elicit information without verification since the objective truth is hard to access or even does not exist (e.g. a self-report survey for involvement in crime). In our paper, we focus on the situation where the objective truth is not observable.

One important framework for designing incentive systems without verification is \emph{peer prediction}~\cite{MRZ05}. \emph{Peer prediction} uses each person's information to predict other people's information and pays according to how good the prediction is. \emph{Peer prediction} assumes people's information is related and the systems and the people share a common prior.

Despite this clever insight, information elicitation mechanisms, such a peer prediction, are not often used in practice. Two main obstacles that prevent peer prediction from being practical are that (1) the mechanism must know the prior; (2) there are many other equilibria besides truth-telling, which may pay as much or even more than the truth-telling equilibrium.

Recent research~\cite{GaoMCA2014} indicates that individuals in lab experiments do not always truth-tell when faced with peer prediction mechanisms; this may in part be related to the issue of equilibrium multiplicity.

\paragraph{Our Contributions}

%high level result
%We design a new mechanism which we call the \emph{Disagreement Mechanism}; compare the payoffs of different equilibria; and show (via our new impossibility results) that this is nearly the best one can hope for with respect to equilibrium multiplicity.

%particular result
\begin{enumerate}
    \item We identify a particular set of equilibria, which we call \emph{permutation equilibria}, which exist in every (non-trivial) information elicitation mechanism; and show that for any (non-trivial) information elicitation mechanism and for any permutation equilibrium, there exists a setting where that permutation equilibrium is paid at least as much as truth-telling;
    \item We propose an information elicitation mechanism that we call \emph{Disagreement Mechanism} where truth-telling is a Bayesian Nash equilibrium for any finite set of signals and any number of agents greater than 3, and where:
    \begin{enumerate}
        \item[a.] Truth-telling is paid strictly more in expectation than any other symmetric equilibrium that is not a permutation equilibrium, which have an expected payoff equal to the truth-telling equilibrium;
        \item[b.] Any symmetric equilibrium that has expected payoff close to truth-telling must be ``close'' to a permutation equilibrium;
        \item[c.]  If the number of agents is sufficiently large then any asymmetric equilibrium which is ``close" to a permutation equilibrium has expected payoff ``close" to that of the truth-telling equilibrium, and any equilibrium that is not ``close" to a permuation equilibrium pays strictly less then the truth-telling equilibrium.
    \end{enumerate}
 \end{enumerate}

%need to rewrite this part
Permutation equilibria are intuitively unnatural and risky as they require extreme coordination amongst the agents.  Thus our results about symmetric equilibrium are quite strong, despite the impossibility result.
Asymmetric equilibrium require more coordination between agents than symmetric equilibrium.  Additionally, the possible gains from doing so are limited and go to zero as the number of agents increase.  We believe the mechanism we propose can motivate a small number of agents to tell the truth no matter how many signals they receive.

\paragraph{High Level Techniques}
Our \emph{Disagreement Mechanism} pays agents individually (locally) for ``agreement'' and globally for ``disagreement''. When agents collude, they share information about their strategies.  Since they are paid individually for ``agreement'', they will use their information about other people's strategies and ``agree to agree'' which reduces their global payoff which depends on ``disagreement''.

Essentially, when agents collude, our \emph{Disagreement Mechanism} encourages each agent to implicitly admit their collusion by unilaterally increasing their individual payoff for doing so, but the mechanism then simultaneously decreases the total payoffs to all agents.  Only when agents do not choose to collude and lack information about other people's strategies, can they ``agree to disagree.'' In this case, even when they maximize their individual payoff, globally they still have a lot of disagreement.

%However, when all people inexplicitly ``admit'' the collusion, the collusion becomes explicit and can be detected.

%Essentially, our \emph{Disagreement Mechanism} pays agents individually (locally) for ``agreement'' and globally for ``disagreement''. When people collude, they share the information about their strategies. Since they are paid individually for ``agreement'', they will use their information about other people's strategies and ``agree to agree'' which reduces their global payoff which depends on ``disagreement''.

\paragraph{Technical Contributions}
In addition to the above results, our works has several contributions in the techniques employed:

\begin{enumerate}
    \item % In a common prior setting, agents with the same private information cannot agree to disagree. Thus, agents with the same signal should report similarly.
    Our \emph{Disagreement Mechanism} encourages not only agents with the same private information to agree, but also agents with different private information to disagree. We present a novel way to measure the amount of ``information'' by casting the reports as multiple labelled points in a space and measuring the quality of the ``classification''.

    \item To show that the ``classification'' quality always decreases with non-truthful equilibria, we exploit tools from information theory, namely \emph{Information Monotonicity}. Despite their natural and powerful application, to our knowledge, this is the first time such tools have been explicitly employed in the peer prediction literature. %In the last section of our paper, we use \emph{Information Monotonicity} and some of our techniques naturally modify a specific family of mechanisms such that they can have the same property that our mechanism has.
    \item We provide a framework to analyse the structure of the equilibrium of a very complex game.
\end{enumerate}

\subsection{Related Work}
\yk{soda r1 comment:
5. The authors need to cite and discuss the paper:  R. Jurca and B.
Faltings. Mechanisms for Making Crowds Truthful. Journal of Artificial
Intelligence Research (JAIR), 34, 2009, pp. 209-253. Other than Dasgupta
and Ghosh, this is the most significant current work on multiplicity of
equilibria. This is also a nice, experimental view on the challenge: Xi
Alice Gao, Andrew Mao, Yiling Chen, and Ryan P. Adam. Trick or Treat:
Putting Peer Prediction to the Test. In Proc. of the 15th ACM Conference
}
\label{sec:related-work}

\ifnum\fullversion=0

%Miller, Resnick, and Zeckhauser introduces peer prediction~\cite{MRZ05}, a host of results~\cite{prelec2004bayesian,witkowski2012robust,radanovic2013robust,radanovic2014incentives,zhang2014elicitability,riley2014minimum,faltings2014incentives} have  

\emph{Bayesian Truth Serum (BTS)}~\cite{prelec2004bayesian} first successfully weakened the known common prior assumption and provides an important framework for the following without known common prior mechanisms.  BTS requires the agents report--in addition to their reported signal--a forecast (prediction) of the other agents' reported signals, and uses this predictions in lieu of the comment prior.   BTS incentivizes agents to report accurate forecasts by rewarding forecasts with the ability to predict the other agents’ reported signal. 

%In detail, BTS derives a prior distribution using an aggregate of these forecasts, and then rewards agents for giving signal reports that are ``unexpectedly common" with respect to this distribution. Intuitively, an agent will think his private signal is underestimated by other people which means he will think the actual fraction of his own private signal is higher than the average of people's predictions. 

However, BTS has two weakness: (1) BTS requires that the number of agents goes to infinity (or is large enough in a modified version) since the mechanism needs agents to believe it has access to the true distribution of from which agents' signals are drawn. (2) The equilibrium analysis provided in~\cite{prelec2004bayesian} is in the case where the number of agents goes to infinity and only proves that truth-telling has expected payment at least as high as other equilibrium so does not avoid the case that truth-telling pays equivalently to many other equilibrium.

Several mechanisms~\cite{witkowski2012robust,radanovic2013robust,radanovic2014incentives,riley2014minimum} are based on the BTS framework and address the first weakness of BTS. \emph{Robust Bayesian Truth Serum (RBTS)}~\cite{witkowski2012robust} is a mechanism which can only be applied to binary signals and \emph{multi-valued BTS}~\cite{radanovic2013robust} mechanism can be applied to non-binary signals while it requires an additional assumption that an agent will think the probability that other agents receive signal $\sigma$ higher if he himself also receives $\sigma$. Both of these works do not solve the equilibrium multiplicity issue, but do work for a small number of agents.  \emph{Minimal Truth Serum (MTS)}~\cite{riley2014minimum} is a mechanism where agents have the option to report or not report their predictions, and also lacks analysis of non-truthful equilibria. MTS uses a typical zero-sum technique such that all equilibria are paid equally. In contrast, we show that in our \emph{Disagreement Mechanism} any equilibrium that is even close to paying more than the truth-telling equilibrium must be close to a small set of ``permutation" equilibrium. 

Thus, while the above work addresses the first weakness of BTS, it does not address the second and also lacks  analysis of the equilibrium. Our work follows the framework of BTS and addresses the two weaknesses of BTS simultaneously.

There are a series of papers that are in different setting from ours (e.g.\  where agents have several a priori similar tasks)~\cite{faltings2014incentives,dasgupta2013crowdsourced,zhang2014elicitability,cai2014optimum,ghosh2014buying,jurcafaltings06} which are further discussed on the full version of this paper.

\paragraph{Independent Work} We propose a mechanism what we call the \emph{Truthful Mechanism} that is the same with the divergence based BTS mechanism independently proposed in~\cite{radanovic2014incentives}. In this unknown common prior mechanism, agents can naturally report their signal and prediction at the same time and truth-telling is a strict Bayesian Nash equilibrium for a small group of agents and non-binary signals without additional assumption. However,~\cite{radanovic2014incentives} does not analyze the equilibrium of the divergence based BTS or address the equilibrium multiplicity issue. In our paper, we analyze the structure of the equilibrium in \emph{Truthful Mechanism} and propose a modified mechanism what we call \emph{Disagreement Mechanism} which retains the properties of \emph{Truthful Mechanism} and addresses the equilibrium multiplicity issue.

\else
%\yk{I add the first paragraph here}
%After ~\citet{MRZ05} introducing peer prediction, a host of results (see, e.g., ~\cite{jurca2007collusion,jurcafaltings09,jurcafaltings06,goelrp09}) have followed. Like we mentioned in introduction, the known common prior assumption and equilibrium multiplicity issue prevent peer prediction from being practical. 

%For the equilibrium multiplicity issue, recent research~\cite{GaoMCA2014} indicates that individuals in lab experiments do not always truth-tell when faced with peer prediction mechanisms; this may in part be related to the issue of equilibrium multiplicity.~\cite{jurca2007collusion,jurcafaltings09} deal with equilibrium multiplicity issue when the mechanism knows the prior. However, this work only analyses pure strategies and leaves the analysis of mixed strategy as open question.  

There are several papers~\cite{prelec2004bayesian,witkowski2012robust,radanovic2013robust,radanovic2014incentives,zhang2014elicitability,riley2014minimum,faltings2014incentives,witkowski2012peer,witkowski2013learning,witkowski2014robust} that focused on the goal of removing the assumption that the mechanism knows the common prior and successfully weaken this known common prior assumption after peer prediction~\cite{MRZ05} is introduced.

%Miller, Resnick, and Zeckhauser introduces peer prediction~\cite{MRZ05}, a host of results~\cite{prelec2004bayesian,witkowski2012robust,radanovic2013robust,radanovic2014incentives,zhang2014elicitability,riley2014minimum,faltings2014incentives} have  

\emph{Bayesian Truth Serum (BTS)}~\cite{prelec2004bayesian} first successfully weakened the known common prior assumption and provides an important framework for mechanisms without known common prior. BTS requires the agents report---in addition to their reported signal---a forecast (prediction) of the other agents' reported signals, and uses this predictions in lieu of the common prior. BTS incentives agents to report accurate forecasts by rewarding forecasts that have the ability to predict the other agents’ reported signal. 

%In detail, BTS derives a prior distribution using an aggregate of these forecasts, and then rewards agents for giving signal reports that are ``unexpectedly common" with respect to this distribution. Intuitively, an agent will think his private signal is underestimated by other people which means he will think the actual fraction of his own private signal is higher than the average of people's predictions. 

However, BTS has two weakness: (1) BTS requires that the number of agents goes to infinity (or is large enough in a modified version) since the mechanism needs agents to believe it has access to the true distribution of from which agents' signals are drawn. (2) The equilibrium analysis provided in~\cite{prelec2004bayesian} is in the case where the number of agents goes to infinity and only proves that truth-telling has expected payment at least as high as other equilibrium and so does not avoid the case that truth-telling pays equivalently to many other equilibrium.

Several mechanisms~\cite{witkowski2012robust,radanovic2013robust,radanovic2014incentives,riley2014minimum,witkowski2012peer,witkowski2013learning,witkowski2014robust} are based on the BTS framework and address the first weakness of BTS. \emph{Robust Bayesian Truth Serum (RBTS)}~\cite{witkowski2012robust} is a mechanism which can only be applied to binary signals and the \emph{multi-valued RBTS}~\cite{radanovic2013robust} mechanism can be applied to non-binary signals while it requires an additional assumption that every agent believes the probability that other agents receive signal $\sigma$ is higher if he himself also receives $\sigma$. Both of these works do not solve the equilibrium multiplicity issue, but do work for a small number of agents.  \emph{Minimal Truth Serum (MTS)}~\cite{riley2014minimum} is a mechanism where agents have the option to report or not report their predictions, and also lacks analysis of non-truthful equilibria. MTS uses a typical zero-sum technique such that all equilibria are paid equally. In contrast, we show that in our \emph{Disagreement Mechanism} any equilibrium that is even close to paying more than the truth-telling equilibrium must be close to a small set of ``permutation" equilibrium. 

Thus, while the above work addresses the first weakness of BTS, it does not address the second and also lacks  an analysis of the non-truthful equilibrium.~\citet{jurca2007collusion,jurcafaltings09} have analysis of non-truthful pure strategies. However,~\citet{jurca2007collusion,jurcafaltings09} assume the mechanism knows the prior and leave the analysis of mixed strategies as an open question.~\citet{2016arXiv160307319K} have analysis of all equilibria including truthful and non-truthful, pure and mixed strategies while the mechanism still needs to know the prior in \cite{2016arXiv160307319K}'s setting. Our work follows the framework of BTS and addresses the two weaknesses of BTS simultaneously when the mechanism does not know the common prior.

Now we introduce several works that have different settings than our work. The mechanisms in~\cite{dasgupta2013crowdsourced,kamble2015truth} are under a different setting where agents have several a priori similar tasks. The mechanism in~\cite{dasgupta2013crowdsourced} rewards agents based on agreement which is similar to the Peer Prediction setting. That mechanism also uses the presence of multiple tasks to elicit agent strategies with high effort, and thus \citet{dasgupta2013crowdsourced} address the equilibrium multiplicity issue for binary signals in their setting. Our setting is different since the agents only have one task (and thus we do not have to assume relations between tasks) and our results for equilibrium multiplicity issue are robust to non-binary signals.~\citet{kamble2015truth} consider both homogeneous and heterogeneous population setting. However, this work requires large group of a priori similar tasks and their mechanisms contain equilibria that are paid higher than truth-telling.~\citet{zhang2014elicitability} propose a mechanism \emph{Knowledge Free Peer Prediction (KFPP)} that does not know the common prior and has truth-telling as an equilibrium for small group of agents and non-binary signals. However, \emph{Knowledge Free Peer Prediction (KFPP)} is a sequential game so agents cannot naturally report signals and predictions at the same time which is different than our setting (KFPP can be implemented non-sequentially but this requires a very complicated and unrealistic prediction reports).  Finally, \citet{cai2014optimum} have a different setting than our work while it also uses the Peer Prediction insight. The mechanism in \cite{cai2014optimum} collects a set of data $(x_i,y_i)$ to approximate a function $f$ where $y_i$ is a noisy version of $f(x_i)$. That mechanism rewards worker $i$ who provides data $(x_i,y_i)$ by comparing $y_i$ and $\hat{f}_{-i}(x_i)$ where $\hat{f}_{-i}$ is an approximate function based on other workers' data.   

%%%%%%%%Grant Comment Here%%%%%%%
%\textbf{make sure you mention ghosh2014buying jurcafaltings06}

\paragraph{Independent Work} As a building block to our final mechanism, we propose a mechanism that we call the \emph{Truthful Mechanism} which is the same as the ``Divergence-Based Bayesian Truth Serum" independently proposed in~\cite{radanovic2014incentives}. In this unknown common prior mechanism, agents can naturally report their signal and prediction at the same time and truth-telling is a strict Bayesian Nash equilibrium for a small group of agents and non-binary signals without additional assumptions. However,~\citet{radanovic2014incentives} do not analyse non-truthful equilibria. In our paper, we analyse the structure of the equilibrium (including non-truthful equilibria) in the \emph{Truthful Mechanism} and propose a modified mechanism that we call the \emph{Disagreement Mechanism} which retains the same set of equilibria as the \emph{Truthful Mechanism} yet addresses the equilibrium multiplicity issue. 

\fi

\section{Preliminaries, Background, and Notation}\label{section:prelim}

\ifnum\fullversion=1
We will defer the proofs for most claims to Section~\ref{section:proof_claims}.
\else
\fi

%Due to the space limitation, we will defer most proofs to the full version.

\subsection{Prior Definitions and Assumptions}

We consider a setting with $n$ agents and a set of signals $\Sigma$, and define a \textit{setting} as a tuple $(n,\Sigma)$. Each agent $i$ has a private signal $\sigma_i \in \Sigma$ chosen from a joint distribution $Q$ over $\Sigma^n$ called the prior.
Given a prior $Q$, for $\sigma \in \Sigma$, let $q_i(\sigma) = \Pr_Q[\sigma_i = \sigma]$ be the \emph{a priori} probability that agent $i$ receives signal $\sigma$.  Let $q_{j, i}(\sigma'|\sigma) = \Pr_Q[\sigma_j = \sigma'| \sigma_i = \sigma]$ be the probability that agent $j$ receives signal $\sigma$ given that agent $i$ received signal $\sigma'$.

We say that a prior $Q$ over $\Sigma$ is \emph{symmetric} if for all $\sigma$, $\sigma' \in \Sigma$ and for all pairs of agents $i \neq j$ and $i' \neq j'$  we have $q_i(\sigma) = q_{i'}(\sigma)$ and $q_{i, j}(\sigma|\sigma') = q_{i', j'}(\sigma|\sigma')$.  That is, the first two moments of the prior do not depend on the agent identities.

%\begin{definition}
%We say that a prior $Q$ over $\Sigma$ is \emph{symmetric} if for all $\sigma$, $\sigma' \in \Sigma$ there exists values  $q(\sigma)$ and $q(\sigma|\sigma')$ such that for all $i \neq j$, $q(\sigma) = \Pr[\sigma_i = \sigma]$ and  $q(\sigma'|\sigma) = \Pr[\sigma_j = \sigma'| \sigma_i = \sigma]$.
%\end{definition}

\begin{assumption}[Symmetric Prior]
We assume throughout that the agents' signals $\boldsymbol{\sigma}$ are drawn from some joint {\bf symmetric prior} $Q$.
\end{assumption}

Because we will assume that the prior is symmetric, we denote $q_i(\sigma)$ by $q(\sigma)$ and   $q_{i, j}(\sigma|\sigma') $ (where $i \neq j$) by  $q(\sigma|\sigma')$.
We also define $\mathbf{q_{\sigma}}=q(\cdot|\sigma)$.

\begin{assumption}[Non-zero Prior]
We assume that for any $\sigma,\sigma' \in \Sigma$, $q(\sigma)>0,q(\sigma|\sigma')>0$.
\end{assumption}

\begin{assumption}[Informative Prior]
We assume if agents have different private signals, they will have  different expectations for the fraction of at least one signal.  That is for any $\sigma \neq \sigma' $, there exists $\sigma''$ such that $q(\sigma''|\sigma)\neq q(\sigma''|\sigma')$.
\end{assumption}

%\yk{another two assumptions}

%We sometimes will denotes the class of priors that satisfy all three of these assumptions as SIN priors.

%Now we introduce a strengthening of the Informative Prior assumption that is necessary if truth-telling is to be \emph{strictly} ``better'' than any other symmetric equilibrium excluding permutation equilibrium. However, this assumption is not necessary for the result that truth-telling is a strict Bayesian equilibrium and ``better'' than any other symmetric equilibrium.

The following assumption conceptually states that one state is not just a more likely version of another state, and can be thought of as a weaker version of assuming $q(\sigma|\cdot)$ are linearly independent.

\begin{assumption}[Fine-grained Prior]
We assume that for any $\sigma\neq \sigma' \in \Sigma$, there exists $\sigma'',\sigma'''$ such that $$ \frac{q(\sigma|\sigma'')}{q(\sigma'|\sigma'')}\neq \frac{q(\sigma|\sigma''')}{q(\sigma'|\sigma''')} $$
\end{assumption}

If this assumption does not hold, then in some since $\sigma$ and $\sigma'$ are the same signal.  We can create a new prior by replacing $\sigma$ and $\sigma'$ with a new signal $\sigma_0:=\sigma$ or $\sigma'$, and not lose any information, in the sense that we can still recover the original prior.  To see this, we first define $p = \frac{q(\sigma)}{q(\sigma')}$, and note that for all $\sigma''$, $p = \frac{q(\sigma|\sigma'')}{q(\sigma'|\sigma'')}$.  Whenever $\sigma_0$ is drawn in the new prior, we simply replace it by $\sigma$ with probability $p$ and $\sigma'$ with probability $1 - p$.  This produces the same prior for agents that have no information or other their signal's information.

We illustrate this in the below example:

%This assumption can be seen a robust version of informative prior: even agents have some information about other agents' signals, agents with different private signals still have different expectations for the fraction of at least one signal. Let's see the below example:

%%%%%%%%%%%%%%%%cut begin%%%%%%%%%%%%%%%%%

\ifnum\fullversion=1

\begin{example}
$Q=\left(\begin{array}{ccc}
q(s_1|s_1) & q(s_1|s_2) & q(s_1|s_3) \\
q(s_2|s_1) & q(s_2|s_2) & q(s_2|s_3) \\
q(s_3|s_1) & q(s_3|s_2) & q(s_3|s_3)
\end{array}\right)=\left(\begin{array}{ccc}
0.1 & 0.2 & 0.3 \\
0.2 & 0.4 & 0.6 \\
0.7 & 0.4 & 0.1
\end{array} \right)$ is not a fine-grained prior since $$\frac{q(s_1|s_1)}{q(s_2|s_1)}=\frac{q(s_1|s_2)}{q(s_2|s_2)}=\frac{q(s_1|s_3)}{q(s_2|s_3)}$$
\end{example}

Note that in this example, even we combine $s_1$ and $s_2$ to be a single signal $s_0$ which is defined as $s_0:=s_1$ or $s_2$, we do not lose any information: if an agent knows that the fraction of agents who report $s_0$ is $x$, we know his belief for the expectation of the fraction of $s_1$ must be $\frac{x}{3}$ no matter what private signal he receives.

We only  require the fine-grained prior assumption to show that truth-telling is \emph{strictly} ``better'' than any other symmetric equilibrium (excluding permutation equilibrium). In the above example where the prior is not fine-grained, if agents always report $s_1$ when they receive $s_1$ or $s_2$, this does not lose information (is not ``worse'') comparing with the case agents always tell the truth. So we cannot say truth-telling is strictly ``better'' than any other equilibrium when the prior is not fine-grained. However, this assumption is not necessary to show that truth-telling is a strict Bayesian equilibrium of our mechanism, nor to show that the agent welfare of truth-telling is at least as high as other symmetric equilibrium.

\else

\fi

%%%%%%%%%%%%%%%%%%%%%%cut end%%%%%%%%%%%%%%%%%%%

 %that is necessary if truth-telling is to be \emph{strictly} ``better'' than any other symmetric equilibrium excluding permutation equilibrium. However, this assumption is not necessary for the result that truth-telling is a strict Bayesian equilibrium and ``better'' than any other symmetric equilibrium.

\gs{Yuqing should check this}
%and all $Q_n$ have the same $\{\mathbf{q}_{\sigma}|\sigma\in \Sigma\}$ and $\{q_{\sigma}|\sigma\in \Sigma\}$ which is the assumption below:
\begin{assumption}[Ensemble Prior]
Although we talk of a single prior, in fact we have an ensemble $Q = \{Q_n\}_{n\in N,n\geq 3}$ of priors; one for each possible number of agents greater than 3.  We assume that all $Q_n$ are over the same signal set $\Sigma$ have have identical $q(\sigma)$ and $q(\sigma'|\sigma)$.
\end{assumption}

When the number of agents $n$ changes, the joint prior actually changes as well, but the first two moments of the prior are fixed. This allows us to make meaningful statements about $n$ going to infinity.

%We will call the first two moments of $Q$ as $Q_2$.

%Let
%\[q(\sigma'|\sigma) := \Pr[\sigma_j = \sigma'| \sigma_i = \sigma]\]
%(where $j \neq i$) be the fraction of other agents that a user $i$ expects have received %signal $\sigma'$ given that he has signal $\sigma$.

%Here we can see agents can not agree to disagree which means once they have the same private signal, they will have the same opinion for the world.

We sometimes will denote the class of priors that satisfy all five of these assumptions as SNIFE priors.

\subsection{Game Setting and Equilibrium Concepts}

Given a setting $(n,\Sigma)$ with prior $Q$, we consider a game in which each agent $i$ is asked to report his private signal $\sigma_i \in \Sigma$ and his prediction $ \mathbf{p}_i \in \Delta_{\Sigma}$, a distribution over $\Sigma$, where  $\mathbf{p}_i= \mathbf{q}_{\sigma_i}$. For any $\sigma\in \Sigma$, $\mathbf{p_i}(\sigma)$ is agent $i$'s (reported) expectation for the fraction of other agents who has received $\sigma$ given he has received $\sigma_i$. However, agents may not tell the truth. We denote $\Sigma \times \Delta_{\Sigma}$ by $\mathcal{R}$. We define a report profile of agent $i$ as $r_i=(\hat{\sigma}_i,\mathbf{\mathbf{\hat{p}}}_i)\in \mathcal{R} $ where $\hat{\sigma_i}$ is agent $i$'s reported signal and $\mathbf{\hat{p}}_i$ is agent $i$'s reported prediction.

We would like to encourage truth-telling, namely that agent $i$ reports $\hat{\sigma_i} = \sigma_i,\mathbf{\hat{p}}_i=\mathbf{q}_{\sigma_i}$. To this end, agent $i$ will receive some payment $\nu_i(\hat{\sigma}_i,\mathbf{\hat{p}}_i, \hat{\sigma}_{-i},\mathbf{\hat{p}}_{-i})$ from our mechanism.

Now we consider the strategy an agent plays in the game.

\begin{definition}[Strategy]
Given a mechanism $\mathcal{M}$, we define the strategy of $\mathcal{M}$ for setting $(n,\Sigma)$ as a mapping $s$ from $(\sigma,Q)$ (the signal and common prior received) to a probability distribution over $\mathcal{R}$ (the reported signal, prediction pair).
\end{definition}

That is, for each possible signal $\sigma$ and prior $Q$ an receives, he will choose a signal, prediction pair to report from some distribution $s(\sigma,Q)$. We define a strategy profile $\mathbf{s}$ as a profile of all agents' strategies $\{s_1,s_2,...s_n\}$ and we say agents play $\mathbf{s}$ if for any $i$, agent $i$ plays strategy $s_i$. We say a strategy profile is \textbf{symmetric} if each agent plays the same strategy.

%If there exists $\sigma\in \Sigma$ and at least two report profiles $r,r'$ such that both $s_i(\sigma)(r),s_i(\sigma)(r')$ are positive, then the strategy $s_i$ is a mixed strategy.

We define the \textbf{agent welfare} of a strategy profile $\mathbf{s}$ and a mechanism $\mathcal{M}$ for setting $(n,\Sigma)$ with prior $Q$ to be the expectation of the sum of payments to each agent and we write it as $AW_{\mathcal{M}}(n,\Sigma,Q,\mathbf{s})$.

%For prediction strategy, it is a distribution over $\Delta_{\Sigma}$. But later we will see the payoff function is strict convex for agent's prediction, so if we only consider equilibrium, the best prediction response for each agent should be a unique value which means we don't need to consider the randomness of prediction strategy, we can only consider pure prediction strategy.

%For example, for truth-telling strategy $\tru$, $\tru(\sigma,Q)$ is a distribution where with probability 1, agent reports $(\sigma,p_i)$.

A {\em Bayesian Nash equilibrium} consists of a strategy profile $s = (s_1, \ldots, s_n)$ such that no player wihes to change her strategy, given the strategies of the other players and the information contained in the prior and her signal. Formally, \begin{definition}[Bayesian Nash equilibrium]
Given a family of priors $\mathcal{Q}$, a strategy profile $s = (s_1, \ldots, s_n)$ is a Bayesian Nash equilibrium if and only if for any prior $Q\in\mathcal{Q}$, for any $i$, and for any $s'_i$
\begin{align*}
&\E_{(\hat{\sigma}'_i,\mathbf{\hat{p}}'_i)\leftarrow s'_i(\sigma_i,Q), (\hat{\sigma}_{-i},\mathbf{\hat{p}}_{-i})\leftarrow s_{-i}(\sigma_{-i},Q)}[\nu_i(\hat{\sigma}'_i,\mathbf{\hat{p}}'_i, \hat{\sigma}_{-i},\mathbf{\hat{p}}_{-i})]\\
\leq & \E_{(\hat{\sigma}_i,\mathbf{\hat{p}}_i)\leftarrow s_i(\sigma_i,Q), (\hat{\sigma}_{-i},\mathbf{\hat{p}}_{-i})\leftarrow s_{-i}(\sigma_{-i},Q)}[\nu_i(\hat{\sigma}_i,\mathbf{\hat{p}}_i, \hat{\sigma}_{-i},\mathbf{\hat{p}}_{-i})]
\end{align*}

In the case where, for some $i$, the equality holds if and only if $s'_i=s_i$, we say this strategy profile is a \emph{strict Bayesian Nash equilibrium} for prior family $\mathcal{Q}$.
\end{definition}

\begin{remark}[Equilibrium for a Given Prior]
Note that we assume agents have a common prior $Q$, so often for convenience, we will implicitly assume $Q$ is fixed, at which point a strategy is a mapping from $\Sigma$ to a probability distribution over $\mathcal{R}$.  We will call such a strategy profile $\mathbf{s}$ an equilibrium for prior $Q$ if it satisfies the condition of Bayesian Nash equilibrium when $Q$ is fixed.
\end{remark}

Assuming a fixed prior $Q$, for any strategy profile $s=(s_1,s_2,...,s_n)$, we will represent the marginal distribution of an agent $i$'s strategy for her signal report as a matrix $ \theta_i$ where $\theta_i(\hat{\sigma},\sigma)$ is the probability that agent will report signal $\hat{\sigma}$ when his private signal is $\sigma$. Note that $\theta_i$ is a \textbf{transition matrix}, that is the sum of every column is 1. We call $\theta_i$ the signal strategy of agent $i$. We also call $(\theta_1,\theta_2,...,\theta_n)$ the signal strategy of $s$. We define the \textit{average signal strategy} of $s$ as $\bar{\theta}_n=\frac{\sum_i\theta_i}{n}$. The following claim relates this average signal strategy to the distribution of all reported signals:

\begin{claim}\label{claim:average signal strategy}
Assume that the distribution over all agents' private signals is $\omega\in \Delta_{\Sigma}$, the distribution over all agents' reported signals will be $\bar{\theta}_n\omega$.
\end{claim}

Note that the mechanism actually collects agents' reported signals, so in order to estimate the distribution over their private signals, we hope $\bar{\theta}_n$ is (close to) the identity matrix $I$.

%We will show the strategy profile where every one tells the truth is a Bayesian Nash equilibrium in our mechanism. However, it turns out that with no knowledge of the prior, it is impossible to make truth-telling the only Bayesian Nash equilibrium that obtains the highest sum of payments for agents in expectation. Agents may use what we call a permutation strategy profile to obtain the same sum of payments with truth-telling.
\subsection{Special Strategy Profiles}\label{strategyprofile}

In this section, we will introduce three special types of strategy profiles that we call \textit{truth-telling}, \textit{best prediction strategy profiles}, and \textit{permutation strategy profiles}.

\begin{definition}[Truth-telling]
We define a strategy profile as truth-telling if for all $i$, and for all $Q$, $s(\sigma_i,Q) = (\sigma_i, \mathbf{q}_{\sigma_i})$ with probability 1. We write the truth-telling strategy profile as $\tru$.
\end{definition}

 For every agent $i$, let $\hat{\sigma}$ be a randomly chosen agent's reported signal, when other agents tell the truth, the distribution of $\hat{\sigma}$ is $\mathbf{q}_{\sigma_i}$. However, if agents play strategy $\mathbf{s}$, for agent $i$, the distribution of $\hat{\sigma}$ depends on not only his prior $Q$ but also the strategy $\mathbf{s}$. We define the distribution of $\hat{\sigma}$ for agent $i$ as $\mathbf{q}^{\mathbf{s}}_{\sigma_i}$.

 \begin{claim}\label{claim:best prediction}
 $$\mathbf{q}^{\mathbf{s}}_{\sigma_i}=\theta_{-i}\mathbf{q}_{\sigma_i}$$ where $(\theta_1,\theta_2,...,\theta_n)$ is $\mathbf{s}$'s signal strategy and $\theta_{-i}=\frac{\sum_{j\neq i}\theta_j}{n-1}$.
\end{claim}

When agents play strategy $\mathbf{s}$, to best predict other agents' reported signal, agent $i$ should be report $\mathbf{q}^{\mathbf{s}}_{\sigma_i}$ rather than $\mathbf{q}_{\sigma_i}$. This motivates our definition for \textit{best prediction strategy profile} which is a strategy profile where every agent $i$ gives his ``best prediction'' $\mathbf{q}^{\mathbf{s}}_{\sigma_i}$.

\begin{definition}[Best Prediction Strategy Profile]\label{bpsp}
We say a strategy profile $\mathbf{s}$ is a best prediction strategy profile if for every agent $i$, he reports $\mathbf{q}^{\mathbf{s}}_{\sigma_i}$. We call a best strategy prediction strategy profile $\mathbf{s}$ a \textit{symmetric best strategy prediction strategy profile} if $\theta_i=\theta$ for every $i$.
\end{definition}

Now we begin to introduce the definition of a permutation strategy profile. Intuitively, if agents ``collude'' to relabel the signals and then tell the truth with relabeled signals, they actually play what we will call permutation strategy profile.

Given a permutation $\pi:\Sigma\mapsto\Sigma$ (which is actually a relabeling of signals), by abusing notation a little bit, we define $\pi:\mathcal{Q}\mapsto\mathcal{Q}$ as a mapping from a prior $Q$ to a \textit{permuted prior} $\pi(Q)$ where for any $\sigma_1,\sigma_2,...,\sigma_n\in \Sigma$,
$$Pr_{\pi( Q)}(\sigma_1,\sigma_2,...,\sigma_n)=Pr_Q(\pi^{-1}(\sigma_1),\pi^{-1}(\sigma_2),...,\pi^{-1}(\sigma_n))$$ where $\sigma_i$ is the private signal of agent $i$. Notice that it follows that:
$$Pr_{\pi( Q)}(\pi(\sigma_1),\pi(\sigma_2),...,\pi(\sigma_n))=Pr_Q(\sigma_1,\sigma_2,...,\sigma_n).$$  Intuitively, $\pi(Q)$ is the same with $Q$ when the signals are relabeled according to $\pi$.

For any strategy $s$, we define $\pi(s)$ as the strategy such that  $\pi(s)(\sigma,Q)=s(\pi(\sigma),\pi (Q))$.

\begin{definition}[Permuted Strategy Profile]
For any strategy profile $\mathbf{s}$, we define $\pi(\mathbf{s})$ as a strategy profile with $\pi(\mathbf{s})=(\pi(s_1),\pi(s_2),...,\pi(s_n))$.
\end{definition}

Note that $\pi^{-1}\pi Q=Q$ which implies $\pi^{-1}\pi(\mathbf{s})=\mathbf{s}$.

\begin{definition}[Permutation Strategy Profile]
We define a strategy profile $\mathbf{s}$ as a permutation strategy profile if there exists a permutation $\pi: \Sigma \rightarrow \Sigma$ such that $\mathbf{s}=\pi(\tru)$.
\end{definition}

Note that if agents play $\pi(\tru)$, then the signal strategy of each agent is $\pi$, and so the distribution of report profiles is $\bar{\theta}_n\omega=\pi \omega $.
\ifnum\fullversion=1

There exists a natural bijection between permutation strategy profiles and $|\Sigma|\times|\Sigma|$ permutation matrices. If the permutation strategy profile is constructed by permutation $\pi$, the only non-zero entries of the corresponding permutation matrix $\theta_{\pi}$ are $\theta_{\pi}(\pi(\sigma),\sigma)=1$ for all $\sigma\in \Sigma$. For a transition matrix $\theta$, if $\theta$ is not a permutation matrix, we would like to give a definition for when a transition matrix $\theta$ is what we call \emph{$\tau$-close} to a permutation given any $\tau>0$. This definition is motivated by the below claim and will be described after it.

\begin{claim}\label{claim:permutation_matrix}
For any transition matrix $\theta_{m\times m}$ where the sum of every column is 1, $\theta$ is a permutation matrix iff for any row of $\theta$, there  at most one non-zero entry.
\end{claim}

Now we give a definition for \emph{$\tau$-close}.

\begin{definition}[$\tau$-close]
We say a signal strategy $\theta$ is \emph{$\tau$-close} to a permutation if for any row of $\theta$, there is at most one entry that is greater than $\tau$.
\end{definition}

%We will use the concept of \emph{$\tau$-close} in our definition for \emph{approximate-quasi-focal} in the below Section.

\else
\begin{definition}[$\tau$-close]
We say a signal strategy $\theta$ is \emph{$\tau$-close} to a permutation if for any row of $\theta$, there is at most one entry that is greater than $\tau$.
\end{definition}

Thus a permutation stragety is 0-close to a permutation. For any stategy profile $s$, if the average signal strategy of $s$ is \emph{$\tau$-close} to a permutation matrix, we say $s$ is \emph{$\tau$-close} to a permutation profile as well.
\gs{I changed this in the short version; I also added a dash.}
\fi

\subsection{Mechanism Design Goals}

We want to design ``good'' mechanisms that motivate the agents to reveal their private information truthfully. Now we will give several definitions for ``good'' mechanisms:

We say that a mechanism is \textbf{\emph{truthful}} if truth-telling is a strict Bayesian Nash equilibrium whenever $Q$ is a symmetric and informative distribution.

We say that a mechanism has truth-telling as a \textbf{\emph{focal}} equilibrium if, for any SNIFE prior $Q$, the mechanism is truthful and the agent welfare is strictly higher in the truth-telling equilibrium than in any other any other Bayesian Nash equilibrium.  Recall that the agent welfare is the expected sum of payments for all agents.

However, it turns out that making truth-telling \emph{focal} is too much to ask. We will show that with an unknown common prior, for any mechanism, and any permutation strategy profile, there exists a prior such that that permutation strategy profile is an equilibrium and has agent welfare at least as much as truth-telling. One natural question is whether there is any mechanism such that all permutation strategy profiles have agent welfare equal to each other, but strictly higher than any other Bayesian Nash equilibrium? This question motivates a weaker version of \emph{focal}: \emph{quasi-focal}.

Formally, we say that a mechanism has truth-telling as a \textbf{\emph{quasi-focal}} equilibrium if, for any SNIFE prior $Q$, the mechanism is truthful and agent-welfare is strictly higher in the permutation equilibrium than in any other Bayesian Nash equilibrium where the agents do not play a permutation strategy profile.

It will turn out that the mechanism we purpose cannot make truth-telling \emph{quasi-focal}. However, we can show that it satisfies the following three slight relaxations of the definition:

%\ifnum\fullversion=1
%\begin{description}
%\else
%\bdesc
%\fi
\begin{description}

    \item[\textbf{\emph{Symmetric-quasi-focal}}:]   We say truth-telling is symmetric-quasi-focal in a mechanism if
    \begin{enumerate}
        \item all permutation equilibrium have equal agent welfare; and
        \item any \emph{symmetric} equilibrium that is not a permutation equilibrium has agent-welfare strictly less than truth-telling.
    \end{enumerate}
    \item[\textbf{\emph{($\tau_1,\gamma_1$)-robust-symmetric-quasi-focal}}:] We say truth-telling is \emph{($\tau_1,\gamma_1$)-robust-symmetric-quasi-focal} in a mechanism if any \emph{symmetric} equilibrium that pays within $\gamma_1$ of truth-telling must be $\tau_1(\gamma_1)$ close to a permutation strategy profile.
    \item[\textbf{\emph{($\tau_2,\gamma_2$)-robust-approximate-quasi-focal}}:] We say truth-telling is \emph{($\tau_2,\gamma_2$)-robust-approximate-quasi-focal} in a mechanism if \begin{enumerate}
        \item all permutation equilibrium pay the same;
        \item no equilibrium has agent welfare greater than $\gamma_2(n)$ more than that of truth-telling where $n$ is the number of agents; and
        \item any profile that pays within $\gamma_2(n)$ of truth-telling must be $\tau_2(n)$ close to a permutation strategy profile.
    \end{enumerate}

\end{description}

%%%%%%%%%%%%%%%%%%cut begin%%%%%%%%%%%%%%%%%%%%

\ifnum\fullversion=1

\subsection{F-divergence}\label{psr}

Now we introduce $f$-divergence, a key tool we will use in our mechanism design.
$f$-divergence(\cite{amari2007methods}) is used to measuring the ``difference'' between distributions. One important property of $f$-divergence is information monotonicity: For any two distributions, if we post-process each distribution in the same way, the two distributions will become ``closer'' because of the information loses.
%We too will use them as an essential ingredient to design a mechanism where truth-telling strategy is a strict Bayesian-Nash equilibrium.

 % We will introduce $f$-divergence---a measure for the difference between two probability distributions and its properties (\cite{amari2010information}).
\textbf{$F$-divergence}~\cite{amari2010information} $D_f:\Delta_{\Sigma}\times \Delta_{\Sigma}\rightarrow \mathbb{R}$ is a non-symmetric measure of difference between a distribution $\mathbf{p}\in \Delta_{\Sigma} $ and a distribution $\mathbf{q}\in \Delta_{\Sigma} $ %. $f$-divergence of $\mathbf{p}$ and $\mathbf{q}$
and is defined to be $$D_f(\mathbf{p},\mathbf{q})=\sum_{\sigma\in \Sigma}
\mathbf{p}(\sigma)f\left( \frac{\mathbf{p}(\sigma)}{\mathbf{q}(\sigma)}\right)$$
where $f(\cdot)$ is a convex function.
Now we introduce the properties of $f$-divergence:

\ifnum\fullversion=1
\begin{enumerate}
\else
\benum
\fi
\item \textbf{Non-negative}: For any $\mathbf{p},\mathbf{q} $, $D_f(\mathbf{p},\mathbf{q})\geq 0$ and $D_f(\mathbf{p},\mathbf{q})=0$ if and only if $\mathbf{p}=\mathbf{q}$.
\item \textbf{Convexity}: Both $D_f(\cdot,\mathbf{q})$ and $D_f(\mathbf{p},\cdot)$ are convex functions for any $ \mathbf{p},\mathbf{q} $.
\item \textbf{Information Monotonicity}: For any $\mathbf{p},\mathbf{q} $, and transition matrix $\theta\in \mathbb{R}^{|\Sigma|\times |\Sigma|}$ where $\theta(\sigma,\sigma')$ is the probability that we map $\sigma'$ to $\sigma$, we have $D_f(\mathbf{p},\mathbf{q})\geq D_f(\theta \mathbf{p},\theta \mathbf{q})$. When $\theta$ is a permutation, $D_f(\mathbf{p},\mathbf{q})= D_f(\theta \mathbf{p},\theta \mathbf{q})$.
\ifnum\fullversion=1
\end{enumerate}
\else
\eenum
\fi
%%%%%%

Now we introduce the proof in (\cite{amari2010information}) for information monotonicity and give the conditions for  the inequality of information monotonicity to be strict.

\begin{lemma}[Information Monotonicity (\cite{amari2010information})]\label{lem:im}
For any strictly convex function $f$,  $f$-divergence $D_f(\mathbf{p},\mathbf{q})$ %=\sum_{\sigma} \mathbf{p}(\sigma) f(\frac{\mathbf{q}(\sigma)}{\mathbf{p}(\sigma)})$
%where $f$ is a strictly convex function, $D_f(\cdot,\cdot)$
satisfies information monotonicity so that for any transition matrix $\theta \in \mathbbm{R}^{\Sigma \times \Sigma}$, $D_f(\mathbf{p},\mathbf{q})\geq D_f(\theta \mathbf{p},\theta \mathbf{q})$.

Moreover, the inequality is strict if and only if there exists $\sigma, \sigma',\sigma''$ such that $\theta(\sigma,\sigma') \mathbf{p}(\sigma')>0$, $\theta(\sigma,\sigma'') \mathbf{p}(\sigma'')>0$ and  $\frac{\mathbf{p}(\sigma'')}{\mathbf{p}(\sigma')}\neq \frac{\mathbf{q}(\sigma'')}{\mathbf{q}(\sigma')}$.
\end{lemma}

We give an example where the strictness condition is not satisfied in appendix.

\gs{the following example should be moved to the appendix}

%We first consider the case where $\sigma \in \{1, 2\}$.  Then, because  $\theta(\sigma,\sigma') \mathbf{p}(\sigma')>0$ and $\theta(\sigma,\sigma'') \mathbf{p}(\sigma'')>0$, it must be that $\sigma', \sigma'' \in \{1, 2\}$.  However, no matter how we choose $\sigma', \sigma'' \in \{1, 2\}$, we cannot get that  $\frac{\mathbf{p}(\sigma'')}{\mathbf{p}(\sigma')}\neq \frac{\mathbf{q}(\sigma'')}{\mathbf{q}(\sigma')}$.

%Note that if $\sigma = 3$, the to have that $\theta(\sigma,\sigma') \mathbf{p}(\sigma')>0$ and $\theta(\sigma,\sigma'') \mathbf{p}(\sigma'')>0$ we need that $\sigma', \sigma'' = 3$.  However, then  $\frac{\mathbf{p}(\sigma'')}{\mathbf{p}(\sigma')} = $\frac{\mathbf{p}(3)}{\mathbf{p}(3)}= \frac{\mathbf{q}(\sigma'')}{\mathbf{q}(\sigma')}$.

%In this case, we can only pick $\sigma'=1$ and $\sigma''=2$ to satisfy the condition: $\theta(\sigma,\sigma') \mathbf{p}(\sigma')>0$, $\theta(\sigma,\sigma'') \mathbf{p}(\sigma'')>0$. However, $\frac{\mathbf{p}(1)}{\mathbf{p}(2)}=\frac{\mathbf{q}(1)}{\mathbf{q}(2)}$. 

%If the strictness condition is not satisfied, we can see $\theta \mathbf{p}$ and $\theta \mathbf{q}$ are $\mathbf{p}$ and $\mathbf{q}$'s sufficient statistic which means the transition $\theta$ does not lose any information, thus, the equality holds.
%%%%%%%%%cut begin%%%%%%%%%%%%%%%%%%%%%%%%

\ifnum\fullversion=1

\begin{proof}
The proof follows from algebraic manipulation and one application of convexity.

\begin{align}
D_f(\theta \mathbf{p},\theta \mathbf{q})=&\sum_{\sigma} (\theta \mathbf{p})(\sigma) f\left(\frac{(\theta \mathbf{q})(\sigma)}{(\theta\mathbf{p})(\sigma)}\right)\\
=& \sum_{\sigma} \theta(\sigma,\cdot) \mathbf{p} f\left(\frac{\theta(\sigma,\cdot) \mathbf{q}}{\theta(\sigma,\cdot)\mathbf{p}}\right)\\
=& \sum_{\sigma} \theta(\sigma,\cdot) \mathbf{p} f\left(\frac{1}{\theta(\sigma,\cdot)\mathbf{p}}\sum_{\sigma'} \theta(\sigma,\sigma')\mathbf{p}(\sigma') \frac{\mathbf{q}(\sigma')}{\mathbf{p}(\sigma')}\right)\\ \label{eq_im}
\leq & \sum_{\sigma} \theta(\sigma,\cdot) \mathbf{p} \frac{1}{\theta(\sigma,\cdot)\mathbf{p}}\sum_{\sigma'} \theta(\sigma,\sigma')\mathbf{p}(\sigma') f\left( \frac{\mathbf{q}(\sigma')}{\mathbf{p}(\sigma')}\right)\\
= & \sum_{\sigma} \mathbf{p}(\sigma) f\left(\frac{\mathbf{q}(\sigma)}{\mathbf{p}(\sigma)}\right)= D_f(\mathbf{p},\mathbf{q})
\end{align}

The second equality holds since $(\theta\mathbf{p})(\sigma)$ is dot product of the $\sigma^{th}$ row of $\theta$ and $\mathbf{p}$.

The third equality holds since $\sum_{\sigma'} \theta(\sigma,\sigma')\mathbf{p}(\sigma') \frac{\mathbf{q}(\sigma')}{\mathbf{p}(\sigma')}=\theta(\sigma,\cdot)\mathbf{q}$.

The fourth inequality follows from the convexity of $f(\cdot)$.

The last equality holds since $\sum_{\sigma}\theta(\sigma,\sigma')=1$.

We now examine under what conditions the inequality in Equation~\ref{eq_im} is strict.
Note that for any strictly convex function $g$, if $\forall u , \lambda_u>0$, $g(\sum_u \lambda_u x_u)=\sum_u \lambda_u g(x_u)$ if and only if there exists $x$ such that $\forall u , x_u=x$. By this property, the inequality is strict if and only if there exists $\sigma, \sigma',\sigma''$ such that $\frac{\mathbf{p}(\sigma'')}{\mathbf{p}(\sigma')}\neq \frac{\mathbf{q}(\sigma'')}{\mathbf{q}(\sigma')}$ and $\theta(\sigma,\sigma') \mathbf{p}(\sigma')>0$, $\theta(\sigma,\sigma'') \mathbf{p}(\sigma'')>0$.

\end{proof}

\else

\fi

%%%%%%%%%%%cut end%%%%%%%%%%%%%%%%%%%%%%%%%%%%%%

\begin{corollary}\label{im_cor}
Given SNIFE prior $Q$, for any $\theta$ that is not a permutation, there exists two private signals $\sigma_1\neq \sigma_2$ such that $D_f(\theta \mathbf{q}_{\sigma_1},\theta \mathbf{q}_{\sigma_2})<D_f(\mathbf{q}_{\sigma_1},\mathbf{q}_{\sigma_2})$
\end{corollary}

%%%%%%%%%%%%%%%%cut begin%%%%%%%%%%%%%%%%%%%%%

\ifnum\fullversion=1
\begin{proof}
First notice that when $\theta$ is not a permutation, based on Claim~\ref{claim:permutation_matrix}, there exists a row of $\theta$ such that the row has at least two positive entries, in other words, there exists $\sigma,\sigma',\sigma''$ such that $\theta(\sigma,\sigma'),\theta(\sigma,\sigma'')>0$. Based on the non-zero and fine-grained assumptions of $Q$, there exists $\sigma_1\neq \sigma_2$ such that\\ $\theta(\sigma,\sigma')\mathbf{p}(\sigma'),\theta(\sigma,\sigma'')\mathbf{p}(\sigma'')>0$ and $\frac{\mathbf{p}(\sigma')}{\mathbf{p}(\sigma'')}\neq \frac{\mathbf{q}(\sigma')}{\mathbf{q}(\sigma'')} $ where $\mathbf{p}=\mathbf{q}_{\sigma_1},\mathbf{q}=\mathbf{q}_{\sigma_2}$. When\\ $\theta(\sigma,\sigma')\mathbf{p}(\sigma'),\theta(\sigma,\sigma'')\mathbf{p}(\sigma'')>0$, we have $\theta(\sigma,\cdot)\mathbf{p}>0$. By Lemma~\ref{lem:im}, we have $D_f(\theta \mathbf{q}_{\sigma_1},\theta \mathbf{q}_{\sigma_2})<D_f(\mathbf{q}_{\sigma_1},\mathbf{q}_{\sigma_2})$
\end{proof}

\else

\fi

%%%%%%%%%%%%%%%%%%%%%%%cut end%%%%%%%%%%%%%%

%\begin{proof}We only need to prove that for any subset $S$ of $[m]$, there exists $\mathbf{p},\mathbf{q}$ such that $\mathbf{p}_S$ is linear independent with $\mathbf{q}_S$. It's true since $Q_{m\times m}$ has full rank.\end{proof}

%\gs{define KL and put any theorems that we use here}
Now we will introduce a special $f$-divergence called Hellinger-divergence and then we will give several properties of Hellinger-divergence (\cite{amari2010information}).

\paragraph{Hellinger-divergence}
If we pick the convex function $f(\cdot)$ to be $(\sqrt{x}-1)^2$, we will obtain Hellinger-divergerce
$$D^*(\mathbf{p},\mathbf{q})=\sum_{\sigma} (\sqrt{\mathbf{p}(\sigma)}-\sqrt{\mathbf{q}(\sigma)})^2$$
 Since Hellinger-divergence is an $f$-divergence, it also has convexity and information monotonicity. Besides these two properties, Hellinger-divergence has another two properties which will be used in the future.
\ifnum\fullversion=1
\begin{enumerate}
\else
\benum
\fi
\item \textbf{Bounded divergence}: $0\leq D^*(\mathbf{p},\mathbf{q})\leq 1$
\item \textbf{Square root triangle inequality}: $\sqrt{D^*(\cdot,\cdot)}$ is a metric.
\ifnum\fullversion=1
\end{enumerate}

Note that Hellinger-divergence is bounded which is different than KL-divergence. Since we use Hellinger-divergence in our disagreement mechanism, we can always guarantee bounded payment.  
\else
\eenum
\fi

\else

%%%%%%%%%%%%%%%%%divergence short version%%%%%%%%%%%%%

%%%%%%%%%%%AAAI change begin here%%%%%%%%%%%%%%%%%%%%%%%%

\ifnum\fullversion=1

\subsection{$F$-divergence and Proper Scoring Rules}

Now we introduce $f$-divergence and strictly proper scoring rules, which are two of the main tools we will use in our mechanism design. Starting with~\cite{MRZ05}, proper scoring rules have become a common ingredient in mechanisms for unverifiable information elicitation (e.g.~\cite{prelec2004bayesian,witkowski2012robust}). $F$-divergence (\cite{amari2007methods}) is always used in measuring the ``difference'' between distributions. One important property of $f$-divergence family is information monotonicity: For any two distributions, if we use the same way to post-process each distribution, the two distributions will become ``closer'' because of potential information loses.

\else

\subsection{Hellinger-divergence and Proper Scoring Rules}

Now we introduce Hellinger-divergence and strictly proper scoring rules, which are two of the main tools we will use in our mechanism design. Starting with~\cite{MRZ05}, proper scoring rules have become a common ingredient in mechanisms for unverifiable information elicitation (e.g.~\cite{prelec2004bayesian,witkowski2012robust}). Hellinger-divergence is a type of $f$-divergence (\cite{amari2007methods}). $F$-divergence is always used in measuring the ``difference'' between distributions. One important property of $f$-divergence is information monotonicity: For any two distributions, if we use the same way to post-process each distribution, the two distributions will become ``closer'' because of potential information loses. The reason we pick Hellinger-divergence rather than other $f$-divergence is that we need \textbf{square root triangle inequality} of Hellinger-divergence (which we will describe later).

\fi
%%%%%%%%%%%%%%%AAAI change end here%%%%%%%%%%%%%%%%%%%%%%%%%%

%%%%%AAAI Cut begin here%%%%%%%%%%%%%%%%%%%%%

\ifnum\fullversion=1

\paragraph{$F$-divergence} % We will introduce $f$-divergence---a measure for the difference between two probability distributions and its properties (\cite{amari2010information}).
$F$-divergence~\cite{amari2010information} $D_f:\Delta_{\Sigma}\times \Delta_{\Sigma}\rightarrow \mathbb{R}$ is a non-symmetric measure of difference between distribution $\mathbf{p}\in \Delta_{\Sigma} $ and distribution $\mathbf{q}\in \Delta_{\Sigma} $ %. $f$-divergence of $\mathbf{p}$ and $\mathbf{q}$
and is defined to be $$D_f(\mathbf{p},\mathbf{q})=\sum_{\sigma\in \Sigma}
\mathbf{p}(\sigma)f( \frac{\mathbf{p}(\sigma)}{\mathbf{q}(\sigma)})$$
where $f(\cdot)$ is a convex function.

\else

\fi

%%%%%AAAI Cut end here%%%%%%%%%%%%%%%%%%%%%
\paragraph{Hellinger-divergence}
\ifnum\fullversion=1
If we pick convex function $f(\cdot)$ as $(\sqrt{x}-1)^2$, we will obtain Hellinger-divergerce
\else
$D^*:\Delta_{\Sigma}\times \Delta_{\Sigma}\rightarrow \mathbb{R}$ is a non-symmetric measure of difference between distribution $\mathbf{p}\in \Delta_{\Sigma} $ and distribution $\mathbf{q}\in \Delta_{\Sigma} $ and is defined to be

\fi
$$D^*(\mathbf{p},\mathbf{q})=\sum_{\sigma} (\sqrt{\mathbf{p}(\sigma)}-\sqrt{\mathbf{q}(\sigma)})^2$$

\ifnum\fullversion=1
Thus Hellinger-divergence is a type of $f$-divergence.
\else
\fi

We highlight two important properties of Hellinger-divergence: one is \emph{Information Monotonicity} which other $f$-divergences also have; another is \emph{square root triangle inequality}.

(1) \textbf{Information Monotonicity}: For any $\mathbf{p},\mathbf{q} $, and transition matrix $\theta\in \mathbb{R}^{|\Sigma|\times |\Sigma|}$ where $\theta(\sigma,\sigma')$ is the probability that we map $\sigma'$ to $\sigma$, we have $D^*(\mathbf{p},\mathbf{q})\geq D^*(\theta \mathbf{p},\theta \mathbf{q})$. When $\theta$ is a permutation, $D^*(\mathbf{p},\mathbf{q})= D^*(\theta \mathbf{p},\theta \mathbf{q})$.

(2) \textbf{Square root triangle inequality}: $|\sqrt{D^*(\mathbf{p},\mathbf{q})}-\sqrt{D^*(\mathbf{p},\mathbf{q'})}|<\sqrt{D^*(\mathbf{q'},\mathbf{q})}$ for any $\mathbf{p},\mathbf{q},\mathbf{q'} $

\fi

%%%%%%%%%%%%%%%%cut end%%%%%%%%%%%%%%%%%%%%%%%%%%

\subsection{Proper Scoring Rules}
%\paragraph{Proper Scoring Rules}

Now we introduce strictly proper scoring rules, another key tool we will use in our mechanism design.  Starting with~\cite{MRZ05}, proper scoring rules have become a common ingredient in mechanisms for elicit unverifiable information elicitation (e.g.~\cite{prelec2004bayesian,witkowski2012robust}).

%Starting with~\cite{MRZ05}, proper scoring rules have become a common ingredient in mechanisms for elicit unverifiable information elicitation (e.g.~\cite{prelec2004bayesian,witkowski2012robust}). \gs{more in full version}
%We too will use them as an essential ingredient to design a mechanism where truth-telling strategy is a strict Bayesian-Nash equilibrium.

A scoring rule $PS:  \Sigma \times \Delta_{\Sigma} \rightarrow \mathbb{R}$ takes in a signal $\sigma \in \Sigma$  and a distribution over signals $\delta_{\Sigma} \in \Delta_{\Sigma}$ and outputs a real number.  A scoring rule is \emph{proper} if, whenever the first input is drawn from a distribution $\delta_{\Sigma}$, then the expectation of $PS$ is maximized by  $\delta_{\Sigma}$. A scoring rule is called \emph{strictly proper} if this maximum is unique. We will assume throughout that the scoring rules we use are strictly proper. By slightly abusing notation, we can extend a scoring rule to be $PS:  \Delta_{\Sigma} \times \Delta_{\Sigma} \rightarrow \mathbb{R}$  by simply taking $PS(\delta_{\Sigma}, \delta'_{\Sigma}) = \E_{\sigma \leftarrow \delta_{\Sigma}}(\sigma,  \delta'_{\Sigma})$.  We note that this means that any proper scoring rule is linear in the first term.

%In the case of scoring rules over binary signals, a distribution can be represented by a number in the unit interval, denoting the probability placed on the signal $1$.  In the binary signal setting, then, we extend proper scoring rules to be defined on $[0, 1] \times [0, 1]$,
%and in fact sometimes further require that they be well-defined on a larger, convex domain $\mathbb{R} \times \mathbb{R}$ that contains $[0,1] \times [0,1]$.
%and in fact further require that they be well-defined on $\mathbb{R} \times \mathbb{R}$.
%\jz{Changed ''and in fact further require that they be well-defined on $\mathbb{R} \times \mathbb{R}$''}
%One can see that in this representation the scoring rule must be affine in the first entry~\cite{mccarthy1956}.

%We will use the Brier Scoring Rule for predicting binary events.\jz{Not true anymore. We want to keep the part on the Brier scoring rule somewhere to show that we have at least one example of a scoring rule that works, and because it is an example of scoring rule that can achieve any slope a/b}

%%%%%%%%%%%%%%cut begin%%%%%%%%%%%%%%%%%%%%%

\ifnum\fullversion=1

\begin{example}[Example of Proper Scoring Rule]
Fix an outcome space $\Sigma$ for a signal $\sigma$.  Let $\mathbf{q} \in \Delta_{\Sigma}$ be a reported distribution.
The Logarithmic Scoring Rule maps a signal and reported distribution to a payoff as follows:
$$L(\sigma,\mathbf{q})=\log (\mathbf{q}(\sigma)).$$

Let the signal $\sigma$ be drawn from some random process with distribution $\mathbf{p} \in \Delta_\Sigma$.

Then the expected payoff of the Logarithmic Scoring Rule
$$ \E_{\sigma \leftarrow \mathbf{p}}[L(\sigma,\mathbf{q})]=\sum_{\sigma}\mathbf{q}(\sigma)\log \mathbf{q}(\sigma)=L(\mathbf{p},\mathbf{q})$$

According to \cite{gneiting2007strictly}, this value will be maximized if and only if $\mathbf{q}=\mathbf{p}$.

\end{example}

\else

\fi

%%%%%%%%%%%%cut end%%%%%%%%%%%%%%%%%%%%%%%%

\section{Impossibility of Truth-telling Being Focal}\label{section:impossible}

In this section, we show the impossibility of truth-telling being focal. In fact, we show something stronger, that no strategy profile can always be the equilibrium and have agent welfare that is strictly greater than any other equilibrium (including truth-telling). Moreover, our impossibility result applies to a very general setting of mechanisms. The case that we consider, where each agent reports a signal and prediction pair, is a special case of this general setting. 

We first define the class of mechanisms to which our impossibility result will apply. 

\begin{definition}[Mechanism]
We define a mechanism $\mathcal{M}$ for a setting $(n,\Sigma)$ as a tuple $\mathcal{M}:=(\mathcal{R},M)$ where $\mathcal{R}$ is a set of all possible report profiles the mechanism allows, and $M:\mathcal{R}^n\mapsto \mathbb{R}^n$ is a mapping from all agents' report profiles to each agent's reward.
\end{definition}

The intuitive explanation for this impossibility result is that the agents can collude to relabel the signals and the mechanism has no way to defend against this relabeling without knowing some information about agents' common prior.  

The proposition stated below implies that, in particular, no strategy profile can always be an equilibrium that has the agent welfare that is strictly greater than any other equilibrium.  

\begin{proposition}\label{permutation} 
Let $\mathcal{M}$ be a mechanism that does not know the common prior, for any strategy profile $s$, and any permutation $\pi$:\\
(1) $s$ is a strict Bayesian Nash equilibrium of $\mathcal{M}$ for any symmetric, informative prior iff $\pi(s)$ is a strict Bayesian Nash equilibrium of $\mathcal{M}$ for any symmetric, informative prior.\\
(2) There exists a prior $Q$ such that $AW_{\mathcal{M}}(n,\Sigma,Q,\mathbf{s})\leq AW_{\mathcal{M}}(n,\Sigma,Q,\pi(\mathbf{s})) $. 
\end{proposition}

The key idea to prove this theorem is what we refer to as \textbf{Indistinguishable Scenarios}:

\begin{definition}[Scenario]
We define a scenario for the setting $(n,\Sigma)$ as a tuple $(Q,\mathbf{s})$ where $Q$ is a prior, and $\mathbf{s}$ is a strategy profile. 
\end{definition}

Given mechanism $\mathcal{M}$, for any scenario $A=(Q_A,\mathbf{s}_A)$, we write $AW_{\mathcal{M}}(n,\Sigma,A)$ as the agent welfare when agents play $\mathbf{s}_A$ and have common prior $Q_A$.

For two scenarios $A=(Q_A,\mathbf{s}_A)$, $B=(Q_B,\mathbf{s}_B)$ for setting $(n,\Sigma)$, let $\sigma_A:=(\sigma_{1_A},\sigma_{2_A},...,
\sigma_{n_A})$ be agents' private signals drawn from $Q_A$, $\sigma_B:=(\sigma_{1_B},\sigma_{2_B},...,
\sigma_{n_B})$ be agents' private signals drawn from $Q_B$. 

\begin{definition}[Indistinguishable Scenarios]\label{indistinguishable}
We say two scenarios $A,B$ are indistinguishable $A\approx B$ if there is a coupling of the random variables $\sigma_A$ and $\sigma_B$ such that $\forall i$, $s_A(\sigma_{i_A},Q_A)=s_B(\sigma_{i_B},Q_B)$ and agent $i_A$ has the same belief about the world as agent $i_B$, in other words, $Pr(\hat{\sigma}|\sigma_{i_A},Q_A,s_A)=Pr(\hat{\sigma}|\sigma_{i_B},Q_B,s_B)$ $\forall \hat{\sigma}\in\Sigma$. 
\end{definition}

Now we will prove two properties of indistinguishable scenarios which are the main tools in the proof for our impossibility result.

\begin{observation}
If $(Q_A,\mathbf{s}_A)\approx (Q_B,\mathbf{s}_B)$, then (i) for any mechanism $\mathcal{M}$,  $\mathbf{s}_A$ is a (strict) equilibrium for prior $Q_A$ iff $\mathbf{s}_B$ is a (strict) equilibrium for prior $Q_B$. (ii) $AW_{\mathcal{M}}(n,\Sigma,A)=AW_{\mathcal{M}}(n,\Sigma,B)$
\end{observation}

At a high level, (1) is true since any reported profile distribution that agent $i_A$ can deviate to, agent $i_B$ can deviate to the same reported profile distribution as well and obtain the same expected payment as agent $i_A$.

Formally, we will prove the $\Rightarrow$ direction in (1) by contradiction. The proof of the other direction will be similar. Consider the coupling for $\sigma_A,\sigma_B$ mentioned in the definition of indistinguishable scenarios. For the sake of contradiction, assume there exists $i$ and $\sigma_{i_B}$ such that $\hat{\sigma}'\neq s_{i_B}(\sigma_{i_B},Q_B)$ is a best response for agent $i_B$. Since agent $i_A$ has the same belief about the world as agent $i_B$ and $s_{i_A}(\sigma_{i_A},Q_A)=s_{i_B}(\sigma_{i_B},Q_B)$, $\hat{\sigma}'\neq s_{i_A}(\sigma_{i_A},Q_A)$ is a best response to agent $i_A$ as well, which is a contradiction to the fact that $\mathbf{s}_A$ is a strictly equilibrium for prior $Q_A$. 

To gain intuition about (2), consider the coupling again.  For any $i$, agent $i_A$ reports the same thing and has the same belief for the world as agent $i_B$, which implies the expected payoff of agent $i_A$ is the same as agent $i_B$. (2) follows.

Now we are ready to prove our impossibility result:
\begin{proof}[of Proposition~\ref{permutation}]
We prove part (1) and part (2) separately. 

\paragraph{Proof of Part (1)} 

Let $A:=(Q,\mathbf{s}),B:=(\pi^{-1}(Q),\pi(\mathbf{s}))$. We will show that for any strategy profile $\mathbf{s}$ and any prior $Q$, $A\approx B$. Based on our above observations, part (1) immediately follows from that fact. 

To prove $(Q,\mathbf{s})\approx (\pi^{-1}Q,\pi(\mathbf{s}))$, we can couple $(\sigma_1,\sigma_2,...,\sigma_n)$ with $(\pi^{-1}(\sigma_1),\pi^{-1}(\sigma_2),..,\pi^{-1}(\sigma_n))$ where $(\sigma_1,\sigma_2,...,\sigma_n)$ is drawn from $Q$. It is a legal coupling since $$ Pr_{\pi^{-1}(Q)}(\pi^{-1}(\sigma_1),\pi^{-1}(\sigma_2),..,\pi^{-1}(\sigma_n))=Pr_{Q}(\sigma_1,\sigma_2,...,\sigma_n)$$ according to the definition of $\pi^{-1}(Q)$. 

Now we show this coupling satisfies the condition in Definition~\ref{indistinguishable}. First note that\\ $\pi(s_i)(\pi^{-1}(\sigma_i),\pi^{-1}(Q))=s_i(\sigma_i,Q)$. Now we begin to calculate $Pr(\hat{\sigma}|\sigma_{i_B},Q_B,s_B)$
\begin{align}
\label{im0}
Pr(\hat{\sigma}|\sigma_{i_B},Q_B,s_B)=&Pr(\hat{\sigma}|\pi^{-1}\sigma_i,\pi^{-1} (Q), \pi(s))\\ \label{im1}
=&\sum_j \Pr(j)\sum_{\sigma'}Pr_{\pi^{-1} (Q)}(\sigma'|\pi^{-1}\sigma_i)Pr({\pi(s_j)(\sigma',\pi^{-1} (Q))}=\hat{\sigma})\\ \label{im2}
=&\sum_j \Pr(j)\sum_{\sigma'}Pr_{\pi^{-1} (Q)}(\sigma'|\pi^{-1}\sigma_i)Pr(s_j(\pi (\sigma'),\pi \pi^{-1} (Q))=\hat{\sigma})\\ \label{im3}
= & \sum_j \Pr(j)\sum_{\sigma'}Pr_{Q}(\pi (\sigma')|\sigma_i)Pr( s_j(\pi (\sigma'), Q)=\hat{\sigma})\\ \label{im4}
= &\sum_j \Pr(j) \sum_{\sigma''}Pr_{Q}(\sigma''|\sigma_i)Pr(s_j(\sigma'', Q)=\hat{\sigma})\\ \label{im5}
=& Pr(\hat{\sigma}|\sigma_i,Q, s)=Pr(\hat{\sigma}|\sigma_{i_A},Q_A,s_A)
\end{align}

From (\ref{im0}) to (\ref{im1}): To calculate a randomly chosen agent's reported signal, we should sum over all possible agents $j$ and calculate the probability conditioning on agent $j$ being picked. To calculate the probability that agent $j$ has reported $\hat{\sigma}$, we should sum over all possible private signals agent $j$ has received and calculate the probability agent $j$ reported $\hat{\sigma}$ conditioning on he received private signal $\sigma'$, which is determined by agent $j$'s strategy. 

By abusing notation a little bit, we can write ${\pi(s_j)(\sigma',\pi^{-1} Q)}$ as a random variable (it is actually a distribution) with $Pr({\pi(s_j)(\sigma',\pi^{-1} Q)}=\hat{\sigma})={\pi(s_j)(\sigma',\pi^{-1} Q)}(\hat{\sigma})$. According to above explanation, (\ref{im1}) follows. 

(\ref{im2}) follows from the definition of permuted strategy (See Section~\ref{strategyprofile}). 

(\ref{im3}) follows from the definition of permuted prior (See Section~\ref{strategyprofile}). 

By replacing $\pi(\sigma')$ by $\sigma''$, (\ref{im4}) follows.

We finished the proof $A\approx B$, as previously argued, result (1) follows.  

\vspace{5pt}

\paragraph{Proof for Part (2)}

We will prove the second part by contradiction:

Fix permutation strategy profile $\pi$. First notice that there exists an positive integer $O_d$ such that $\pi^{O_d}=I$ where $I$ is the identity and agents play $I$ means they tell the truth (we can pick $O_d$ as the order of $\pi$ in the permutation group).

Given any strategy profile $s$, for the sake of contradiction, we assume that there exists a mechanism $\mathcal{M}$ with unknown common prior such that $AM_{\mathcal{M}}(n,\Sigma,Q,\mathbf{s})>AM_{\mathcal{M}}(n,\Sigma,Q,\pi(\mathbf{s}))$ for any prior $Q$. For positive integer $k\in\{0,1,...,O_d\}$, we construct three scenarios:

$$A_k:=(\pi^k(Q),s),\ A_{k+1}:= (\pi^{k+1}(Q),s),\ B_{k}:=(\pi^k(Q),\ \pi(s))$$

and show for any $k$, 

(I)$AM_{\mathcal{M}}(n,\Sigma,A_{k})>AM_{\mathcal{M}}(n,\Sigma,B_{k})$, 

(II) $AM_{\mathcal{M}}(n,\Sigma,A_{k+1})=AM_{\mathcal{M}}(n,\Sigma,B_{k})$. 

Combining (I), (II) and the fact $A_0=A_{O_d}$, we have $$AM_{\mathcal{M}}(n,\Sigma,A_{0})>AM_{\mathcal{M}}(n,\Sigma,A_{1})>...AM_{\mathcal{M}}(n,\Sigma,A_{O_d})=AM_{\mathcal{M}}(n,\Sigma,A_{0})$$ which is a contradiction. 

Now it is only left to show (I) and (II). Based on our assumption $$AM_{\mathcal{M}}(n,\Sigma,Q,\mathbf{s})>AM_{\mathcal{M}}(n,\Sigma,Q,\pi(\mathbf{s}))$$ for any prior $Q$, we have (I). By the same proof we have in part (1), we have $A_{k+1}\approx B_{k}$, which implies (II) according to our above observations.

\end{proof}

%Consider a special strategy $s$ that for any permutation $\pi$, $\pi(s)=s$, then it is possible that $s$ can be focal (for example, when there is only one strategy $s$, $s$ is focal.). However, for most cases (including our truth-telling strategy), this is not true. 

\begin{corollary}
Let $\mathcal{M}$ be a mechanism that does not know the common prior, given truth-telling strategy $\tru$, if there exists a permutation $\pi$ such that $\pi(\tru)\neq \tru$, $\tru$ cannot have agent welfare that is always strictly highest among all equilibria. 
\end{corollary}

%%%%%%%%%%%%%%%%%%%%%%%cut end%%%%%%%%%%%%%%%
The requirement that $\pi(\tru)\neq\tru$ ensures, that a truthful input to the mechanism depends on the private signal of an agent. An example when this would not hold would be if there is only possible input.  
\begin{proof}
According to Proposition~\ref{permutation}, $\pi(\tru)$ is an equilibrium as well and there exists a prior $Q$ such that the agent welfare of $\pi(\tru)$ is greater or equal to $\tru$, so $\tru$ cannot have agent welfare that is always strictly highest among all equilibria.
\end{proof}

\section{The Disagreement Mechanism}\label{section:DM}
%\subsection{Intuition for Mechanisms}

In this section, we will introduce two mechanisms: the \textit{Truthful Mechanism} and the \textit{Disagreement Mechanism}. Both of these mechanisms are based on the Bayesian Truth serum (BTS) framework, which means agents are paid based on a ``prediction score'' and an ``information score''. We will introduce the ``prediction score'' used in the two mechanisms first, it will become a strong tool to motivate people to tell the truth. 

\textbf{Prediction Score via Proper Scoring Rules}  Agents will receive a prediction score based on how well their prediction predicts a randomly chosen agent's reported signal. Say an agent $i$ reports prediction $\hat{\mathbf{p}}_i$ then a random agent, call him agent $j$, is picked for him, agent $i$ will receive a prediction score $PS(\hat{\sigma}_j,\hat{\mathbf{p}}_i)$ when $PS$ is a proper scoring rule (see Section~\ref{psr}) of proper scoring rules family. Note that any proper scoring rule works. $PS(\hat{\sigma}_j,\hat{\mathbf{p}}_i)$ is maximized if and only if agent $i$'s reported prediction $\hat{\mathbf{p}}_i$ is his expected likelihood for $\hat{\sigma}_j$. Agent $i$ cannot pretend to have a different expected likelihood without reducing his expectation for his prediction score. Note that actually when agents play a \textit{best prediction strategy profile} (See definition~\ref{bpsp}), each agent maximizes his individual prediction score given the signal strategy $(\theta_1,\theta_2,...,\theta_n)$.  

\vspace{5pt}

We cannot only pay agents based on prediction score, that will give no incentive to agents to report their private signals. Similar with BTS, we need ``information score'' which is motivated by a concept we call \emph{Inconsistency}:

\subsection{Inconsistency: Agree to Disagree}
The common prior assumption tells us agents cannot agree to disagree. That is, if two agents receive the same private information, they must have the same ``belief'' about the world. In our setting, if agents tell the truth (or use a permutation strategy profile), whenever two agents report the same signal, they will report the same prediction as well. We use the concept of \textit{inconsistency} to represent how much agents agree to disagree which we would like to discourage. We define \textit{Inconsistency} as the ``difference'' between two random agents' predictions when they report the same signal. Permutation strategy profiles will have the lowest inconsistency score $0$. We will give a more formal definition later. \looseness=-1

The \textit{Inconsistency} concept motivates the \textit{Truthful Mechanism}. We use a strictly proper scoring rule $PS(\cdot,\cdot)$ to define the ``difference'' between two agents' predictions in the \textit{Truthful Mechanism} and give each agent $i$ an \textbf{Information Score} by randomly picking an agent $j$ and punishing agent $i$ the ``difference'' between agent $i$'s and agent $j$'s predictions if they report the same signal. 

%There is a problem if we only pay each agent his information score: even though everyone else tells the truth, this agent, who receives private signal $\sigma$, can still claim his signal is $\sigma'\neq \sigma$ and his prediction is $\mathbf{q}_{\sigma'}$ without any punishment. To deal with this problem, we pay each agent $i$ the information score plus a prediction score: $PS(\hat{\sigma}_j,\hat{\mathbf{p}}_i)$. Note that $PS(\hat{\sigma}_j,\hat{\mathbf{p}}_i)$ is maximized if and only if agent $i$'s reported prediction $\hat{\mathbf{p}}_i$ is his expected likelihood for $\hat{\sigma}_j$. Even though agent $i$ can pretend to receive other private signal, he cannot pretend to have a different expected likelihood since other likelihood will reduce his prediction score from his prospective. 
\yk{soda r1 comment: Mechanism M: r is not defined, and the way the prediction score is
defined, it is not clear whether i's payment depends on i's signal or
prediction report. }

\paragraph{Truthful Mechanism $\mathcal{M}$:}

Let $\alpha,\beta > 0$ be parameters and let $PS$ be a strictly proper scoring rule, then we define the \textit{truthful mechanism} $\mathcal{M}(\alpha,\beta,PS)$\footnote{This mechanism is essentially the same as the Divergence-Based Bayesian Truth Serum mechanism %independently proposed and 
previously published in ~\cite{radanovic2014incentives}.} as follows:

\ifnum\fullversion=1
\begin{enumerate}
\else
\benum
\fi
\item Each agent $i$ reports a signal and a prediction $r_i=(\hat{\sigma}_i,\mathbf{\hat{p}}_i)$

\item For each agent $i$ and agent $j$, we define a prediction score that depends on agent $i$'s prediction and agent $j$'s report signal $$score_P(r_i,r_j) = PS(\hat{\sigma}_j,\mathbf{\hat{p}}_i),$$ and an  information score  $$ score_I(r_i,r_j)=\left\{
\begin{aligned}
0 &  & \hat{\sigma}_i \neq \hat{\sigma}_j \\
-(PS(\mathbf{\hat{p}}_j,\mathbf{\hat{p}}_j)-PS(\mathbf{\hat{p}}_j,\mathbf{\hat{p}}_i)) &  & \hat{\sigma}_i= \hat{\sigma}_j 
\end{aligned}
\right.
$$ 

\item Each agent $i$ is matched with a random agent $j$. The payment for agent $i$ is $$payment_{\mathcal{M(\alpha, \beta, PS)}}(i,\mathbf{r})=\alpha score_P(r_i,r_j) + \beta score_I(r_i,r_j). $$
\ifnum\fullversion=1
\end{enumerate}
\else
\eenum
\fi

\begin{theorem}\label{theorem:truthful}\footnote{An equivalent theorem was proved in \citet{radanovic2014incentives}, and we include it here for completeness.}

For  any $\alpha, \beta >0$ and any strictly proper scoring rule $PS$, $\mathcal{M}(\alpha, \beta, PS)$ has truth-telling as a strict Bayesian-Nash equilibrium whenever the prior $Q$ is informative and symmetric. 
\end{theorem}

%%%%%%%%%%%%%%%%%cut begin%%%%%%%%%%%%%%%%%%%
%\begin{comment}
\ifnum\fullversion=1

\begin{proof}
We must show that for every agent, if other agents tell the truth, then this agent can (strictly) maximize his expected payoff if and only if he chooses to tell the truth.  

Assume that all agents other than $i$ are telling the truth. The probability that agent $i$ is matched with agent $j$ is $Pr(j)=\frac{1}{n-1}$. The expected payoff for agent $i$ is:  
\begin{align}
&\E[payment_{\mathcal{M}(\alpha,\beta,PS(\cdot,\cdot))}(i,\mathbf{r})|\sigma_i]\\ \label{e01}
&= \sum_{j\neq i}(Pr(j)\E[\alpha score_P(r_i,r_j) + \beta score_I(r_i,r_j)|\sigma_i])\\ \label{e02}
&=\sum_{j\neq i} \frac{1}{n-1} [\alpha PS(\E(\hat{\sigma}_{j}|\sigma_i),\mathbf{\hat{p}}_i)+\beta (-Pr(\hat{\sigma}_j=\hat{\sigma}_i|\sigma_i) \E[(PS(\mathbf{\hat{p}}_j,\mathbf{\hat{p}}_j)-PS(\mathbf{\hat{p}}_j,\mathbf{\hat{p}}_i))|\sigma_i,\hat{\sigma}_j=\hat{\sigma}_i]]\\ \label{e03}
&=\sum_{j\neq i} \frac{1}{n-1} [\alpha PS(\E(\sigma_{j}|\sigma_i),\mathbf{\hat{p}}_i)+\beta (Pr(\sigma_j=\hat{\sigma}_i|\sigma_i) (PS(\mathbf{q}_{\hat{\sigma}_i},\mathbf{\hat{p}}_i)-PS(\mathbf{q}_{\hat{\sigma}_i},\mathbf{q}_{\hat{\sigma}_i})))]\\ \label{e04}
&=\alpha PS(\E(\sum_{j\neq i} \frac{1}{n-1} \sigma_{j}|\sigma_i),\mathbf{\hat{p}}_i)+\sum_{j\neq i} \frac{\beta}{n-1} [ (Pr(\sigma_j=\hat{\sigma}_i|\sigma_i) (PS(\mathbf{q}_{\hat{\sigma}_i},\mathbf{\hat{p}}_i)-PS(\mathbf{q}_{\hat{\sigma}_i},\mathbf{q}_{\hat{\sigma}_i})))]\\ \label{e05}
&= \alpha PS(\theta_{-i}\mathbf{q}_{\sigma_i},\mathbf{\hat{p}}_i)+\sum_{j\neq i} \frac{\beta}{n-1} [ (Pr(\sigma_j=\hat{\sigma}_i|\sigma_i) (PS(\mathbf{q}_{\hat{\sigma}_i},\mathbf{\hat{p}}_i)-PS(\mathbf{q}_{\hat{\sigma}_i},\mathbf{q}_{\hat{\sigma}_i})))]\\ \label{e06}
&= \alpha PS( \mathbf{q}_{\sigma_i},\mathbf{\hat{p}}_i)+\sum_{j\neq i} \frac{\beta}{n-1} [ (Pr(\sigma_j=\hat{\sigma}_i|\sigma_i) (PS(\mathbf{q}_{\hat{\sigma}_i},\mathbf{\hat{p}}_i)-PS(\mathbf{q}_{\hat{\sigma}_i},\mathbf{q}_{\hat{\sigma}_i})))]
%\leq& \alpha PS(\mathbf{p}_i,\mathbf{p}_i)
\end{align}

From (\ref{e01}) to (\ref{e02}): When $\hat{\sigma}_i\neq 
\hat{\sigma}_j$, the information score is 0, so we only need to consider the case  $\hat{\sigma}_i= \hat{\sigma}_j$.\\

From (\ref{e02}) to (\ref{e03}): All agents other than $i$ tell the truth, so $\hat{\sigma}_j=\sigma_j$ and $$\hat{\mathbf{p}}_j=\mathbf{q}_{\sigma_j}=\mathbf{q}_{\hat{\sigma}_j}=\mathbf{q}_{\hat{\sigma}_i}.$$\\

From (\ref{e03}) to (\ref{e04}): The proper scoring rule is linear for the first entry.\\

From (\ref{e04}) to (\ref{e05}): Based on Claim~\ref{claim:best prediction}, $E(\sum_{j\neq i} \frac{1}{n-1} \sigma_{j}|\sigma_i)=\theta_{-i}\mathbf{q}_{\sigma_i}$.\\

From (\ref{e05}) to (\ref{e06}): Note that for any $j\neq i$, agent $j$ tells the truth so $\theta_{-i}=I$. 

First, if agent $i$ plays truthfully, then $\hat{\sigma}_i=\sigma_i,\mathbf{\hat{p}}_i=\mathbf{q}_{\sigma_i}$, and we will have $ \E(payment(i,\mathcal{M})|\sigma_i) = \alpha PS(\mathbf{q}_{\sigma_i},\mathbf{q}_{\sigma_i}) $ because $PS(\mathbf{q}_{\hat{\sigma}_i},\mathbf{\hat{p}}_i)-PS(\mathbf{q}_{\hat{\sigma}_i},\mathbf{q}_{\hat{\sigma}_i})= 0$.

%Equation (6) is upper bounded by $\alpha PS(\mathbf{p}_i,\mathbf{p}_i)$ because $PS(\mathbf{p}_i,\hat{\mathbf{p}}_i)\leq PS(\mathbf{p}_i,\mathbf{p}_i) $ and because the PS divergence is always non-negative. 

%We show that $E[payment(i,\mathcal{M})|\sigma_i] = \alpha PS(\mathbf{p}_i,\mathbf{p}_i)$ if and only if agent $i$ plays truthfully.  

%%From (6) to (7): Based on the property of proper scoring, we know $PS(\mathbf{p}_i,\hat{\mathbf{p}}_i)\leq PS(\mathbf{p}_i,\mathbf{p}_i) $, we also know that PS divergence is always non-negative. \\

%First, if $\hat{\sigma}_i=\sigma_i,\mathbf{\hat{p}}_i=\mathbf{p}_i$, we will have $ \E(payment(i,\mathcal{M})|\sigma_i) = \alpha PS(\mathbf{p}_i,\mathbf{p}_i) $ because $PS(\mathbf{q}_{\sigma_i},\mathbf{q}_{\sigma_i}) = 0$.

Now show that to receive a payment this high, agent $i$ must play truthfully.  Assume that \\
$\E(payment(i,\mathcal{M})|\sigma_i) \geq \alpha PS(\mathbf{q}_{\sigma_i},\mathbf{q}_{\sigma_i})$. First, the second term of Equation (\ref{e06}) is non-positive based on the property of proper scoring rule. Then we must have that  $PS(\mathbf{q}_{\sigma_i},\mathbf{\hat{p}}_i) \geq  PS(\mathbf{q}_{\sigma_i},\mathbf{q}_{\sigma_i})$, but because $PS$ is a strictly proper scoring rule, this happens only if $\mathbf{\hat{p}}_i  = \mathbf{q}_{\sigma_i}$. But this implies that the second term of Equations (\ref{e06}) equals 0, and this requires that $PS(\mathbf{q}_{\hat{\sigma}_i},\mathbf{\hat{p}}_i)=PS(\mathbf{q}_{\hat{\sigma}_i},\mathbf{q}_{\hat{\sigma}_i})$.

However, by the properties of strictly proper scoring rules, this means
$\mathbf{q}_{\hat{\sigma}_i}=\hat{\mathbf{p}}_i$. However, we already showed that $\hat{\mathbf{p}}_i = \mathbf{p}_i= \mathbf{q}_{\sigma_i}$. Putting this together we see that $\mathbf{q}_{\hat{\sigma}_i}= \mathbf{q}_{\sigma_i}$. Based on the informative prior assumption, this implies that $\hat{\sigma}_i=\sigma_i$.

%we will have $\hat{\sigma}_i=\sigma_i,\mathbf{\hat{p}}_i=\mathbf{p}_i$ which means agent $i$ must tell the truth to get the maximum. 

%Based on the property of strict proper scoring rule, we know that $ PS(\mathbf{p}_i,\mathbf{\hat{p}}_i) $ is maximized if and only if $\mathbf{\hat{p}}_i=\mathbf{p}_i$. When $\mathbf{\hat{p}}_i=\mathbf{p}_i$, since we know $Pr(\sigma_j=\hat{\sigma}_i|\sigma_i)>0$, we need $PS(\mathbf{q}_{\hat{\sigma}_i},\mathbf{p}_i)=0$ which means $\mathbf{q}_{\hat{\sigma}_i}=\mathbf{p}_i=\mathbf{q}_{\sigma_i}$. Then based on informative prior assumption, we have $\hat{\sigma}_i=\sigma_i$.

So we proved that for any agent $i$, when other agents tell the truth, agent $i$ can obtain the best expected payoff if and only if he tells the truth which means truth-telling is a strict Bayesian-Nash equilibrium in the truthful mechanism.  
\end{proof}

\else

\fi 

%\end{comment}

%%%%%%%%%%%%cut end%%%%%%%%%%%%%%%%%%%%%%%%%%

\subsection{Diversity from Disagreement}
The previous section shows that we can encourage truth-telling (or permutation strategy profiles) by punishing \textit{Inconsistency}. The problem is that there may be many other equilibria with inconsistency score $0$. 
%We call those equilibria \textit{consistent} equilibria. 
We consider an extreme example: All agents coordinate together and report the same signal $\sigma$. For every agent, he will predict the fraction of $\sigma$ is 100$\%$ to maximize his prediction score. While such a strategy profile is consistent, the unitary predictions implies their report profiles have little information. Motivated by this extreme example, we define a concept called \textit{Diversity}. 

Recall that \textit{Inconsistency} is the expected ``difference'' (we will describe this ``difference'' later) between two random agents' predictions when they report the same signal, we define \textit{Diversity} as the expected \emph{Hellinger-divergence} between two random agents' predictions when they report different signals. One of the reasons we use \emph{Hellinger-divergence} is that it is $f$-divergence and we can use information monotonicity here. To give a flavor of our future proof, we will give a simple observation here which motivates our \textit{Disagreement Mechanism}. 

\begin{observation}
If the number of agents is infinite and agents play a best prediction strategy, for every two agents $i,j$, their predictions becomes ``closer'' compared with agents play a permutation strategy profile. 
\end{observation}

\begin{proof}
Based on Claim~\ref{claim:best prediction}, agent $i$ will report $\theta_{-i} \mathbf{q}_{\sigma_i}$ given $\sigma_i$ is his private signal and recall that $\theta_{-i}=\frac{\sum_{j\neq i}\theta_i}{n-1}$ where $(\theta_1,\theta_2,....,\theta_n)$ is the signal strategy. Since $\lim_{n\rightarrow\infty}\theta_{-i}=\lim_{n\rightarrow\infty}\theta_{-j}$, we can use $\theta$ to represent the common limit.

Via \textbf{information monotonicity} we have $D_f(\theta \mathbf{q}_{\sigma_i},\theta \mathbf{q}_{\sigma_j})\leq D_f(\mathbf{q}_{\sigma_i},\mathbf{q}_{\sigma_j})%=D_f(\pi\mathbf{q}_{\sigma_i},\pi \mathbf{q}_{\sigma_j})
$ which implies their reported predictions becomes ``closer''. 
\end{proof}

Note that if the number of agents is not infinite, we can still have the same result if agents play a symmetric best prediction strategy profile since we only need $\theta_{-i}=\theta_{-j}$. 

Ideally, we would like to show every equilibrium is a best prediction strategy profile, since it will imply any permutation strategy profile is more diverse than any other equilibrium. However, we believe it is not the case, instead, we introduce a new concept \textit{Classification} that combines \textit{Diversity} and \textit{Inconsistency} and show that the permutation strategy profiles almost have the highest classification scores. 

%%%%%%%AAAI Cut Begin%%%%%%%%%%%%%%%%%%

\ifnum\fullversion=1

%any permutation strategy profile is more diverse than any other equilibrium in \textit{Truthful Mechanism}. However, we do not believe it is always true. Instead, we can show any permutation strategy profile is more diverse than any \emph{best prediction} strategy profile via information monotonicity. Consider two agents who report different signals. If they use a permutation strategy profile $\pi$ then their predictions will be $\pi\mathbf{q}_{\sigma},\pi\mathbf{q}_{\sigma'}$ given their private signals are $\sigma\neq \sigma'$. If they use a symmetric best prediction strategy, then their reported predictions will be $\theta \mathbf{q}_{\sigma},\theta \mathbf{q}_{\sigma'}$. In the first case, the Hellinger divergence between the two agents' reported predictions is $D^*(\pi\mathbf{q}_{\sigma},\pi\mathbf{q}_{\sigma'})=D^*(\mathbf{q}_{\sigma},\mathbf{q}_{\sigma'})$ while in the second case, the Hellinger divergence between the two agents' reported predictions is $D^*(\theta \mathbf{q}_{\sigma},\theta \mathbf{q}_{\sigma'})\leq D^*(\mathbf{q}_{\sigma},\mathbf{q}_{\sigma'})=D^*(\pi\mathbf{q}_{\sigma},\pi\mathbf{q}_{\sigma'})$. The inequality follows from the information monotonicity of Hellinger divergence. Thus, the two agents' predictions in the second case is ``closer'' than those in the first case. So a permutation strategy profile is more diverse than any other symmetric best prediction strategy. 

\else

\fi 
%%%AAAI cut end%%%%%%%%%%%%%%%%%

\subsection{Classification: Diversity $-$ Inconsistency }

The set of all agents' report profiles is a classification. In general, a classification is a map $\mathcal{C}:U\rightarrow\Sigma$ where $U$ are multiple points and $\Sigma$ is a set of labels. All points that are mapped to one label $\sigma$ construct a cluster $C_{\sigma}$. In our setting, the agents who report the same signal construct a cluster which is a collection of their reported profiles. 

If all agents tell the truth, then the predictions that are labeled with the same label are in the same position in $\Delta_{\Sigma}$ (common prior assumption) while the predictions that are labeled with different labels are in different positions (because of our informative prior assumption). Intuitively, it is an ideal classification since typically we hope in a good classification, the points with the same label (in the same cluster) are close while the points with different labels are far from each other. 

%Now we introduce a typical way to measure classification and then apply it to $\mathcal{C}_r$ while yields $\textit{Diversity}-\textit{Inconsistency}$ score. 

In general, a typical way to measure a classification $\mathcal{C}$ is 

$$Q(\mathcal{C})=\sum_{u,v\in U,\mathcal{C}(u)\neq \mathcal{C}(v)}D_1(u,v)-\sum_{u,v\in U,\mathcal{C}(u)= \mathcal{C}(v)}D_2(u,v)$$

where $D_1$ and $D_2$ are functions that  measure the ``difference'' between points $u,v\in U$. 

Note that for $Q(\mathcal{C}_r)$, $\sum_{u,v\in U,\mathcal{C}_r(u)\neq \mathcal{C}_r(v)}D_1(u,v)$ intuitively captures the concept of \textit{Diversity} while $\sum_{u,v\in U,\mathcal{C}_r(u)\neq \mathcal{C}_r(v)}D_2(u,v)$ intuitively captures \textit{Inconsistency}. So $Q(\mathcal{C}_r)$ captures the concept of $\textit{Diversity}-\textit{Inconsistency}$.

%Our technical definition for $\textit{Diversity}-\textit{Inconsistency}$ will be slightly different to the above formula since we will define ``difference'' mentioned above differently in \textit{Diversity} and \textit{Inconsistency}.

\ifnum\fullversion=1

We will also show that if we use $\textit{Diversity}-\textit{Inconsistency}$ as a new score, permutation strategies will have strictly higher score than any other symmetric equilibrium and if the number of agents are sufficient large, the equilibria with a higher score or even a score ``close" to truth-telling must be ``close'' to a permutation strategy profile. We will use this property to design a mechanism that has truth-telling as both a robust-symmetric-quasi-focal and robust-approximate-quasi-focal equilibrium. We will design a mechanism that satisfies the below two conditions which implies that this mechanism has truth-telling as a weakly-quasi-focal equilibrium, and
(1) has the same equilibria with our previous \emph{Truthful Mechanism}; and
(2) in which the average agent-welfare is $\textit{Diversity}-\textit{Inconsistency}$.

\else

\fi

%%%%%AAAI cut end%%%%%%%%%%%%%%%%%%%%%%%%%
\ifnum\fullversion=1
We will see that the below \emph{Disagreement Mechanism} we propose satisfies the above two conditions. To design this mechanism, we
\else
We will show all permutation strategy profiles gain equal $\textit{Diversity}-\textit{Inconsistency}$ score which is almost higher than that of any other equilibrium. We use this fact to design our \textit{Disagreement Mechanism}. We will
\fi 
(a) first use a typical trick\gs{Can we site Peer Prediction here} to create a zero-sum game which has the same equilibria with the previous \emph{Truthful Mechanism}; (b) pay each agent an extra score that only depends on other agents which will not change the structure of the equilibria. We want this extra score to represent $\textit{Diversity}-\textit{Inconsistency}$.

\yk{I do not understand the last part of the comment. soda r1: Mechanism M+: r is not defined, it is not clear when calling payment M
for agent i A whether its peer will only come from A or could come from
A or B (and similarly for paymentM(j B,r)), in step 2 it is not clear
whether j and k can only come from A or can come from any agents. Also: it took me a while to realize that step 1 is the 'zero sum game'
trick and score C is where the I-D payment is computed. You could be
more explicit about this in the text before the mechanism. Furthermore,
it is not clear why the simple modification to the first mechanism that
just introduces a payment term in score I when agents disagree on
signals wouldn't work. I would find this an extremely helpful discussion.}

\paragraph{Disagreement Mechanism $\mathcal{M+}(\alpha,\beta,PS(\cdot,\cdot))$}
$\mathbf{r}=\{r_1,r_2,...,r_n\}$ is all agents' report profiles where for any $r$, $r_i=(\hat{\sigma}_i,\hat{\mathbf{p}}_i)$. 

\begin{enumerate}
\item \textit{Zero-sum Trick}: Divide the agents into two non-empty groups-group A and group B. Each group of agents plays the game (mechanism) $\mathcal{M}$ that is restricted in their own group. For group A, each agent $i_A$ receives a 

\begin{align*}
score_{\mathcal{M}}(i_A,\mathbf{r}) = & payment_{\mathcal{M}(\alpha,\beta,PS(\cdot,\cdot))}(i_A,\mathbf{r}_A)\\
 - & \frac{1}{|A|}\sum_{j_B\in B} payment_{\mathcal{M}(\alpha,\beta,PS(\cdot,\cdot))}(j_B,\mathbf{r}_B)
\end{align*}
Where $payment_{\mathcal{M}(\alpha,\beta,PS(\cdot,\cdot))}(i_A,\mathbf{r}_A)$ is agent $i_A$'s payment when he is paid by mechanism $\mathcal{M}(\alpha,\beta,PS(\cdot,\cdot))$ given group $A$'s report profiles $\mathbf{r}_A$ and he can only be paired with a random peer from group $A$ (we have similar explanation for $payment_{\mathcal{M}(\alpha,\beta,PS(\cdot,\cdot))}(j_B,\mathbf{r}_B)$). For agents in group B, we use the analogous way to score them. 

\item \textit{Additional Classification Reward}: Each agent $i$ is matched with two random agents $j,k$ chosen from all agents (including group $A$ and group $B$), the payment for agent $i$ is 
$$payment_{\mathcal{M+}(\alpha,\beta,PS(\cdot,\cdot)}(i,\mathbf{r})=score_{\mathcal{M}}(i,\mathbf{r})+score_C(r_j,r_k) $$
where $$score_C(r_j,r_k)=\left\{
\begin{aligned}
D^*(\mathbf{\hat{p}}_j,\mathbf{\hat{p}}_k) &  & \hat{\sigma}_j \neq \hat{\sigma}_k \\
-\sqrt{D^*(\mathbf{\hat{p}}_j,\mathbf{\hat{p}}_k)} &  & \hat{\sigma}_j = \hat{\sigma}_k 
\end{aligned}\right.
$$  recall that $D^*$ denotes the Hellinger Divergence.
\end{enumerate}

Note that our \emph{Disagreement Mechanism} cannot make truth-telling \emph{quasi-focal} since if there are $m$ (the number of signals) agents and for any $i$, agent $i$ always reports the $i^{th}$ signal and predicts that the $j^{th}$ signal occurs with probability $\frac{1}{m-1}$ if $j\neq i$ while $i^{th}$ signal occurs with probability 0. In this situation, each agent plays his best response and $\textit{Inconsistency}=0$, yet $\textit{Diversity}$ is much higher than it is in the truth-telling equilibrium. Instead, we are going to prove our \emph{Disagreement Mechanism} makes truth-telling both a robustly-symmetric-quasi-focal and robustly-approximately-quasi-focal equilibrium.

\yk{I do not quite understand this comment as well. soda r1: 6. It would be worth a brief discussion of the challenges with these
kinds of mechanisms that penalize in the case of disagreement. For
example, they can tend to be vulnerable to threats by some agents, and
may suffer low payments out of equilibrium.  (Related: given that the
authors emphasize relaxing common knowledge as one the two key
problems of classical PP, the following two papers are relevant:
Witkowski, J. and Parkes, D. C. (2012b). Peer Prediction Without a
Common Prior. In Proceedings of the 13th ACM Conference on Electronic
Commerce (EC'12). Witkowski, J. and Parkes, D. C. (2013).
Learning the Prior in Minimal Peer Prediction. In Proceedings of the 3rd
Workshop on Social Computing and User Generated Content (SC'13).)}

\begin{theorem}[Main Theorem]\label{main_thm}
Given any SNIFE prior, if the number of agents $n\geq 3$, and if $\frac{\alpha}{\beta}<\frac{1}{4m}$ where $m$ is the number of signals and $n$ is the number of agents, then

\begin{enumerate}
    \item $\mathcal{M+}(\alpha,\beta,PS(\cdot,\cdot))$ is truthful; 
    \item $\mathcal{M+}(\alpha,\beta,PS(\cdot,\cdot))$ has truth-telling as a \emph{symmetric-quasi-focal} equilibrium; 
    \item $\mathcal{M+}(\alpha,\beta,PS(\cdot,\cdot))$ has truth-telling as a \emph{($\tau_1,\gamma_1$)-robust-symmetric-quasi-focal} equilibrium; and
    \item $\mathcal{M+}(\alpha,\beta,PS(\cdot,\cdot))$ has truth-telling as a \emph{($\tau_2,\gamma_2$)-robust-approximate-quasi-focal} equilibrium
\end{enumerate}

where $\tau_1(\gamma_1)=O(\sqrt[3]{\gamma_1})$, $\gamma_2(n)=O(\frac{m}{\sqrt{n}})$ and $\tau_2(n)=O(\sqrt[6]{\frac{m^2}{n})}$ (the constants we omit only depend on the first two moments of prior $Q$)\footnote{Actually $\tau_1(\gamma_1)=\frac{1}{c_1}\sqrt[3]{\frac{\gamma_1}{c_2,c_3,c_4}}$, $\gamma_2(n)=\frac{4\sqrt{2}m}{\sqrt{n}}$ and $\tau_2^6(n)=\frac{128*m^2}{n c_1^6 (c_2 c_3 c_4)^2}$, $c_1 = \min_{s,t\in \Sigma} q(s|t)$, $c_2 = \min_{s,t\in \Sigma}Pr(s,t)$, $c_3 = \min_{u,v} \max_{s,t}||\frac{q(u|s)}{q(u|t)}-\frac{q(v|s)}{q(v|t)}||^2$, $c_4 = \min_{s,t,u} f''(\frac{q(u|s)}{q(u|t)})$ where $f(x)=(\sqrt{x}-1)^2$.}. 
 
\end{theorem}

%\begin{comment}

\paragraph{\textbf{Theorem~\ref{main_thm} Part 1: $\mathcal{M+}(\alpha,\beta,PS(\cdot,\cdot))$ is truthful.}}

\begin{claim}\label{claim:sameeq}
The \emph{Disagreement Mechanism} has the same equilibria as the \emph{Truthful Mechanism}. 
\end{claim}

Note that we have already show $\mathcal{M}(\alpha,\beta,PS(\cdot,\cdot))$ has truth-telling as a strict equilibrium for any SNIFE prior in Theorem~\ref{theorem:truthful}. Since $\mathcal{M+}(\alpha,\beta,PS(\cdot,\cdot))$ does not change the equilibrium structure of $\mathcal{M}(\alpha,\beta,PS(\cdot,\cdot))$ according to Claim~\ref{claim:sameeq}, we have $\mathcal{M+}(\alpha,\beta,PS(\cdot,\cdot))$ has truth-telling as a strict equilibrium for any SNIFE prior as well. 

We will prove the rest of the theorem in the next section. 

%\end{comment}

\ifnum\fullversion=1

\else
%\paragraph{Proof Idea} 

\fi

%One observation is that when agents use permutation strategy profile, for any two agents, if they report the same signal, they will report the same prediction based on our symmetric prior assumption. However, when they play a non-permutation strategy profile, it is possible that two agents report the same signal while report different predictions since they may actually have different private signals and different ``opinions'' for the world. So if we define \textit{Inconsistency} as the divergence between two random agents' predictions when they report the same signal, we can see permutation will have the lowest inconsistency score $0$. The problem is that there may be many other equilibria with inconsistency score $0$. We call those equilibria \textit{consistent} equilibria. We consider an extreme example: All agents report the same predictions. It is clear that it is consistent. However, the information their report profiles have is little since every one's predictions are the same. Motivated by this extreme example, we define a concept called \textit{Diversity}. 

%However, we will prove that a permutation strategy profile is either more consistent or more diverse than any other equilibrium. 
%We will give a formal and technical definition for $\textit{Diversity}$ and $\textit{Inconsistency}$ in the below section and show another intuition for it. Once we finish it, we will describe a mechanism that satisfies the above two conditions. 

\ifnum\fullversion=1

\section{Showing Properties of Equilibrium}\label{section:MTWQ}

In this section, we are going to show Theorem~\ref{main_thm} part 2, 3, 4. 

We first give technical definitions for \textit{Diversity} and \textit{Inconsistency} and then prove that the average agent-welfare in the \emph{Disagreement Mechanism} is $\textit{Diversity} - \textit{Inconsistency}$.  

We first introduce a short hand which will simplify the formula for \textit{Diversity} and \textit{Inconsistency}. 

$$\int_{\hat{j},\hat{k}} Pr(\hat{j},\hat{k})\triangleq \int_{\hat{\sigma}_j,\hat{\mathbf{p}}_j,\hat{\sigma}_k,\hat{\mathbf{p}}_k}Pr_{(\hat{\sigma}_j,\hat{\mathbf{p}}_j)\leftarrow s_j(\sigma_j)}(\hat{\sigma}_j,\hat{\mathbf{p}}_j) Pr_{(\hat{\sigma}_k,\hat{\mathbf{p}}_k)\leftarrow s_k(\sigma_k)}(\hat{\sigma}_k,\hat{\mathbf{p}}_k)$$
where $s_j$ is the strategy of agent $j$ and $s_j(\sigma_j)$  is a distribution over agent $j$'s report profile $(\hat{\sigma}_j,\hat{\mathbf{p}}_j)$ given agent $j$ receives private signal $\sigma_j$ and uses strategy $s_j$, and similarly  for agent $k$.  This defines the natural measure on the reports of agents $j$ and $k$ given  that they play strategies $s_j$ and $s_k$ and a fixed prior $Q$ (which is implicit), and allows us to succinctly describe probabilities of events in this space.

Recall that $\textit{Diversity}$ is the expected Hellinger divergence $D^*$ between two random agents when they report different signals, so

\begin{align*}
\textit{Diversity}= \sum_{\substack{j\\k\neq j}}\sum_{\substack{\sigma_j,\sigma_k}} Pr(j,k) Pr(\sigma_j,\sigma_k) \int_{\hat{j},\hat{k}} Pr(\hat{j},\hat{k})\delta(\hat{\sigma}_j\neq \hat{\sigma}_k)D^*(\hat{\mathbf{p}}_j,\hat{\mathbf{p}}_k)
\end{align*}

where $Pr(j,k)$ is the probability agents $j,k$ are picked, and $Pr(\sigma_j,\sigma_k)$ is the probability that agent $j$ receives private signal $\sigma_j$ and agent $k$ receives private signal $\sigma_k$.

Similarly, we can write down the technical definition for \textit{Inconsistency}. But here we do not use Hellinger divergence as the ``difference'' function in $\sum_{u,v\in U,\mathcal{C}_r(u)= \mathcal{C}_r(v)}D(u,v)$, we use square root of the Hellinger divergence which is the Hellinger distance as the ``difference'' function. The reason is we want to use the convexity of the Hellinger divergence and the triangle inequality of the Hellinger distance. We will describe the details in the future. For now we give a technical definition for \textit{Inconsistency}: 
\begin{align*}
\textit{Inconsistency}=- \sum_{\substack{j\\k\neq j}}\sum_{\substack{\sigma_j,\sigma_k}} Pr(j,k) Pr(\sigma_j,\sigma_k) \int_{\hat{j},\hat{k}} Pr(\hat{j},\hat{k})\delta(\hat{\sigma}_j = \hat{\sigma}_k)\sqrt{D^*(\hat{\mathbf{p}}_j,\hat{\mathbf{p}}_k)}
\end{align*}

%Actually the difference function we use is $=\left\{
%\begin{aligned} D^*(\mathbf{\hat{p}}_j,\mathbf{\hat{p}}_k) &  & \hat{\sigma}_j \neq \hat{\sigma}_k \\ \sqrt{D^*(\mathbf{\hat{p}}_j,\mathbf{\hat{p}}_k)} &  & \hat{\sigma}_j = \hat{\sigma}_k  \end{aligned}\right.$

%Once we pay each agent an extra score, we hope the structure of the equilibria does not change. However, if we pay each agent the average $\textit{Diversity}-\textit{Inconsistency}$ score, the structure of the equilibria will change since $\textit{Diversity}-\textit{Inconsistency}$ depends on all agents' report profiles. To solve this problem, instead of paying each agent the average $\textit{Diversity}-\textit{Inconsistency}$ score, for each agent $i$, we randomly pick two agents $j,k$ (first pick one agent over uniform distribution, and then pick another over uniform distribution ) for agent $i$. We reward agent $i$ $D^*(\hat{\mathbf{p}}_j,\hat{\mathbf{p}}_k)$ if $\hat{\sigma}_j\neq \hat{\sigma}_k$ and we punish agent $i$ $\sqrt{D^*(\hat{\mathbf{p}}_j,\hat{\mathbf{p}}_k)}$ if $\hat{\sigma}_j= \hat{\sigma}_k$. Formally, we define a classification score for agent $i$:

Now we define the $\textit{ClassificationScore}$ as the expected average extra score $score_C$:

\begin{align*}
\textit{ClassificationScore}=\sum_{\substack{i\\j\neq i}}\sum_{\substack{k\neq i,j\\\sigma_i,\sigma_j,\sigma_k}}  Pr(i) Pr(\sigma_i)Pr(j,k)Pr(\sigma_j,\sigma_k|\sigma_i) \int_{\hat{j},\hat{k}} Pr(\hat{j},\hat{k}) score_C(r_j,r_k)
\end{align*}

%Note that permutation will not change the average classification score (it just relabels the clusters). Now we describe a mechanism which satisfies the two conditions we discussed in the beginning of Section~\ref{section:AQFM}. 

\begin{claim}\label{claim:cdi}
$\textit{ClassificationScore}=\textit{Diversity}-\textit{Inconsistency}$
\end{claim}

\begin{claim}\label{claim:per}
Any permutation strategy profile has the same \textit{ClassificationScore}, \textit{Diversity}, and \textit{Inconsistency} as truth-telling. 
\end{claim}

\begin{claim}\label{claim:welfare}
The average agent-welfare in our \emph{Disagreement Mechanism} is the $\textit{ClassificationScore}$  
\end{claim}

\subsection{Proof Outline for Main Theorem}
%\paragraph{Proof Outline for Main Theorem}
First note that the average agent-welfare is the \textit{ClassificationScore}. We want to show that 
\begin{enumerate}
\item If the number of agents is greater than 3, then any \emph{symmetric} equilibrium that is not permutation equilibrium must have \textit{ClassificationScore} strictly less than truth-telling; and any \emph{symmetric} equilibrium that has \textit{ClassificationScore} close to truth-telling must be close to a permutation equilibrium.

\item If the number of agents is sufficient large, then no equilibrium can have a \textit{ClassificationScore} that is much greater than truth-telling; any equilibrium that has \textit{ClassificationScore} close to truth-telling must be close to a permutation equilibrium.  
\end{enumerate}

To prove our main theorem, we first introduce the concept of \textit{TotalDivergence} and then we use this value as a bridge. Recall that we defined\\
$\textit{Diversity}= \sum_{j,k\neq j,\sigma_j,\sigma_k} Pr(j,k)Pr(\sigma_j,\sigma_k) \int_{\hat{j},\hat{k}} Pr(\hat{j},\hat{k})\delta(\hat{\sigma}_j\neq \hat{\sigma}_k)D^*(\mathbf{\hat{p}}_j,\mathbf{\hat{p}}_k)$, now we define a similar concept 
\begin{align*}
\textit{TotalDivergence}= \sum_{j,k,\sigma_j,\sigma_k} Pr(j,k)Pr(\sigma_j,\sigma_k) \int_{\hat{j},\hat{k}} Pr(\hat{j},\hat{k}) D^*(\mathbf{\hat{p}}_j,\mathbf{\hat{p}}_k)
\end{align*}

First note that total divergence is robust to summing over $j,k$ or $j\neq k$ since when $j=k$, $D^*(\mathbf{\hat{p}}_j,\mathbf{\hat{p}}_k)=0$. 

We can see $\textit{TotalDivergence}\geq \textit{Diversity}$ since \textit{TotalDivergence} also includes the divergence between the agents who report the same signals. We show that the equality holds if and only if $\textit{Inconsistency}=0$: 

\begin{claim}\label{claim:inconsistency:zero}
For any strategy profile $s$, $\textit{Diversity}(s)=\textit{TotalDivergence}(s)$ $\Leftrightarrow$ $\textit{Inconsistency}(s)=0$
\end{claim}

\begin{corollary}\label{coro:truth}
\textit{ClassificationScore}(truth-telling)=\textit{Diversity}(truth-telling)=\textit{TotalDivergence}(truth-telling)
\end{corollary}
\begin{proof}
At the truth-telling equilibrium, $\hat{\sigma}_i=\sigma_i, \hat{\mathbf{p}}_i=\mathbf{q}_{\sigma_i}$ for any $i$, so the inconsistency score of truth-telling is 0 since $\hat{\sigma}_j=\hat{\sigma}_k\Rightarrow {\sigma}_j={\sigma}_k\Rightarrow \hat{\mathbf{p}}_j=\hat{\mathbf{p}}_k\Rightarrow D^*(\hat{\mathbf{p}}_j,\hat{\mathbf{p}}_k)=0$ which implies this corollary. 
\end{proof}

Now we begin to state our proof outline: For any equilibrium $s$, we define two modified strategies for $s$:

\begin{enumerate}
\item We define $s_{BP}$ what we call a \textit{best prediction strategy} of $s$ as a strategy where each agent uses the same signal strategy which he uses in $s$ but plays his \textit{best prediction} which maximizes the prediction score. In this case, based on Claim~\ref{claim:best prediction}, for any $i$, agent $i$ plays $\theta_{-i}\mathbf{q}_{\sigma_i}$. We illustrate the report profiles of $s_{BP}$ in the second picture from the left in Figure~\ref{fig:proof outline for main thm}. 
\item We define $symmetrized\ s_{BP}$ as a strategy where each agent plays $\bar{\theta}_n \mathbf{q}_{\sigma}$ given $\sigma$ is his private signal where $\bar{\theta}_n$ is the average signal strategy of $s_{BP}$ (also of $s$). We show the report profiles of $symmetrized\ s_{BP}$ in the third picture of Figure~\ref{fig:proof outline for main thm}  
\end{enumerate}

Our proof can be divided in four parts which are illustrated in Figure~\ref{fig:proof outline for main thm}: 

\begin{enumerate}
\item [(1)] \textbf{$\textit{ClassificationScore}(s)\leq\textit{TotalDivergence}(s_{BP})$.}[ Lemma~\ref{lem:main_lem} ]. This is our main lemma and we defer the proof of main lemma to Section~\ref{section:proof_main_lemma}. Once we show it, we can directly prove that the modified mechanism has truth-telling as symmetric quasi-focal equilibrium based on information monotonicity which will be described in our proof for main theorem. 

\item [(2)] \textbf{$\textit{TotalDivergence}(s_{BP})\approx \textit{TotalDivergence}(symmetrized\ s_{BP})$ when the number of agents is sufficient large.} [Lemma~\ref{lem:N_epsilon}] Intuitively, when $n$ is large enough, $\theta_{-j}$ will be close to $\bar{\theta}_n$. We will use this observation to prove this part. 

\item [(3)] \textbf{$\textit{ClassificationScore}(s^*)\geq \textit{ClassificationScore}(truthtelling)\\
\Rightarrow \textit{TotalDivergence}(truthtelling)\approx \textit{TotalDivergence}(symmetrized\ s^*_{BP})$ when the number of agents is sufficient large.}[Corollary~\ref{coro:ineq_cor}] Here $s^*_{BP}$ is the best prediction strategy of $s^*$. This part will also imply no equilibrium can have \textit{ClassificationScore} that is much greater than truth-telling.

\item [(4)]  \textbf{$\textit{TotalDivergence}(truthtelling)\approx \textit{TotalDivergence}(symmetrized\ s^*_{BP})\Rightarrow \mathbf{\bar{\theta}_n\approx \pi}$} where $\mathbf{\bar{\theta}_n}$ is the average signal strategy of $s^*$. [Lemma~\ref{lem:general case lem}]

We will show $\textit{TotalDivergence}(truthtelling)$ is \begin{align*}
&\sum_{j,k,\sigma_j,\sigma_k}Pr(j,k)Pr(\sigma_j,\sigma_k)D^*(\mathbf{q}_{\sigma_j}, \mathbf{q}_{\sigma_k})\\&=\sum_{j,k,\sigma_j,\sigma_k}Pr(j,k)Pr(\sigma_j,\sigma_k)D^*(\theta_{\pi} \mathbf{q}_{\sigma_j},\theta_{\pi} \mathbf{q}_{\sigma_k})
\end{align*}
where $Pr(j,k)Pr(\sigma_j,\sigma_k)$ is the probability that agent $j,k$ are picked and agent $j$ receives private signal $\sigma_j$; agent $k$ receives private signal $\sigma_k$. 

We will also show $\textit{TotalDivergence}(symmetrized\ s^*_{BP})$ is $\sum_{j,k,\sigma_j,\sigma_k}Pr(j,k)Pr(\sigma_j,\sigma_k)D^*(\bar{\theta}_n \mathbf{q}_{\sigma_j}, \bar{\theta}_n \mathbf{q}_{\sigma_k})$. 

Once we prove (1) and (2), we will show that (3) implies  $\mathbf{\bar{\theta}_n\approx \pi}$. This, informally, means that if an equilibrium pays more than truth-telling, it must be close to a permutation equilibrium, and this pays about the same as truth-telling.

\end{enumerate} 
\begin{figure}
\centering
\includegraphics[scale=0.5]{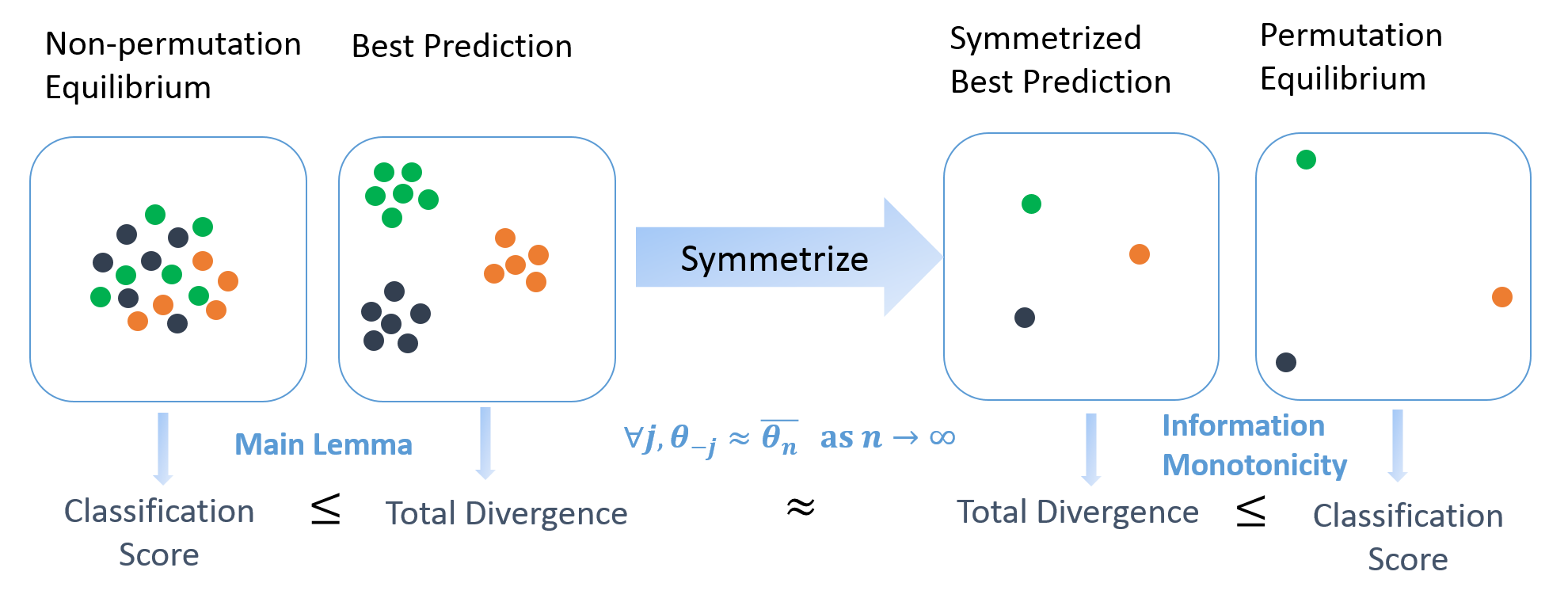}
\caption{Proof Outline for Main Theorem}
\label{fig:proof outline for main thm}
\end{figure}

\subsection{Proof for Main Theorem}\label{section:proof_main_theorem}

\paragraph{(1)$\textit{ClassificationScore}(s)<\textit{TotalDivergence}(s_{BP})$}

\begin{lemma} [Main Lemma]\label{lem:main_lem}
For any equilibrium $s$, if $s_{BP}$ is a \textit{best prediction strategy} of $s$, we have $$\textit{ClassificationScore}(s)\leq \textit{TotalDivergence}(s_{BP})$$
If the equality holds, then we have $\textit{Inconsistency}(s)=0$ and $s=s_{BP}$. 
\end{lemma}

We defer the proof of our main lemma to Section~\ref{section:proof_main_lemma}.

\paragraph{\textbf{Theorem~\ref{main_thm} Part 2: $\mathcal{M+}(\alpha,\beta,PS(\cdot,\cdot))$ has truth-telling as a symmetric-quasi-focal equilibrium. 
 }} We use our main lemma directly to prove: \textbf{any \emph{symmetric} non-permutation equilibrium's agent welfare (\textit{ClassificationScore}) must be strictly less than truth-telling}

Notice that if all agents play a symmetric signal strategy $\theta$, then for any $j,k$, $\theta_{-j}=\theta_{-k}=\theta$. For any symmetric non-permutation equilibrium $s$, it is possible that the signal strategy of $s$ is not a permutation or it is a permutation $\theta_{\pi}$ but agents do not report $\pi \mathbf{q}_{\sigma}$ given $\sigma$ is their private signal. So we consider two cases:

\textbf{(a)} We first consider the case that the signal strategy $\theta$ of $s$ is a permutation matrix $\theta_{\pi}$, but agents do not report $\pi \mathbf{q}_{\sigma}$.

\begin{align*}
\textit{ClassificationScore}(s) < & \textit{TotalDivergence}(s_{BP})\\
=& \sum_{j,k,\sigma_j,\sigma_k} Pr(j,k)Pr(\sigma_j,\sigma_k)D^*(
\theta_{-j}\mathbf{q}_{\sigma_j},\theta_{-k}\mathbf{q}_{\sigma_k})\\
= & \sum_{j,k,\sigma_j,\sigma_k}
Pr(j,k)Pr(\sigma_j,\sigma_k)D^*(
\theta_{\pi}\mathbf{q}_{\sigma_j},\theta_{\pi}\mathbf{q}_{\sigma_k})\\
= & \sum_{j,k,\sigma_j,\sigma_k}
Pr(j,k)Pr(\sigma_j,\sigma_k)D^*(
\mathbf{q}_{\sigma_j},\mathbf{q}_{\sigma_k})\\
=& \textit{TotalDivergence}(truthtelling)\\
=& \textit{ClassificationScore}(truthtelling)
\end{align*}

The first inequality follows from our main lemma. The inequality is strict for the following reason: when the signal strategy $\theta$ of $s$ is a permutation matrix, $s_{BP}$ is a permutation strategy profile since for any $i$, agent $i$'s best prediction is $\theta_{-i}\mathbf{q}_{\sigma}=\theta \mathbf{q}_{\sigma}$. Based on our main lemma if $\textit{ClassificationScore}(s) = \textit{TotalDivergence}(s_{BP})$, we have $s=s_{BP}$ which implies that $s$ is a permutation strategy profile which is a contradiction to the fact $s$ is a non-permutation strategy profile.

The second line follows since at $s_{BP}$, each agent's reported prediction only depends on his private signal.

The last equality follows from Corollary~\ref{coro:truth}.

\textbf{(b)} We consider the case that the signal strategy $\theta$ of $s$ is not a permutation matrix. The above proof still holds except in two places: one is the inequality in the first line may not be strict; another is the equality in the fourth line should be a strict inequality:
\begin{align*}
&\sum_{j,k,\sigma_j,\sigma_k}
Pr(j,k)Pr(\sigma_j,\sigma_k)D^*(
\theta_{\pi}\mathbf{q}_{\sigma_j},\theta_{\pi}\mathbf{q}_{\sigma_k})\\
< & \sum_{j,k,\sigma_j,\sigma_k}
Pr(j,k)Pr(\sigma_j,\sigma_k)D^*(
\mathbf{q}_{\sigma_j},\mathbf{q}_{\sigma_k})
\end{align*}
The inequality must be strict since based on Corollary~\ref{im_cor}, we know that if $\theta$ is not a permutation, and $Q$ is fine-grained, then there exists $\sigma_1\neq \sigma_2$ such that $D^*(\theta \mathbf{q}_{\sigma_1},\theta \mathbf{q}_{\sigma_2})<D^*(\mathbf{q}_{\sigma_1},\mathbf{q}_{\sigma_2})$. Also based on non-zero assumption of $Q$, we have $Pr(\sigma_j=\sigma_1,\sigma_k=\sigma_2)>0$. 

So in both of the above two cases, we have $$\textit{ClassificationScore}(s)<\textit{ClassificationScore}(truthtelling)$$ if $s$ is not a permutation equilibrium \\

Things are more complicated if agents play asymmetric equilibria, and we begin to handle this case by showing Lemma~\ref{lem:N_epsilon}, Corollary~\ref{coro:ineq_cor}, and Lemma~\ref{lem:general case lem}.

\paragraph{(2) $\textit{TotalDivergence}(s_{BP})\approx \textit{TotalDivergence}(symmetrized\ s_{BP})$ when the number of agents is sufficient large.}

We symmetrize $s_{BP}$ which means we let each agent $i$ report $\bar{\theta}_n\mathbf{q}_{\sigma_i}$ given $\sigma_i$ is agent $i$'s private signal and $\bar{\theta}_n$ is the average signal strategy of $s_{BP}$ and show that the total divergence will not change much. Intuitively, this is because $\theta_{-i}$ are similar among agents when there are many agents.

\begin{lemma}\label{lem:N_epsilon}
Given any SNIFE prior $Q$, for any $\epsilon>0$, there exists $N_{\epsilon}=\frac{32*m^2}{\epsilon^2}$ such that if $n>N_{\epsilon}$, for any strategy $(\theta_1,\theta_2,...,\theta_n)$, any two agents $j,k$,
\begin{align*}
| D^*(\theta_{-j}\mathbf{q}_{\sigma_j},\theta_{-k}\mathbf{q}_{\sigma_k})
- D^*(\bar{\theta}_{n}\mathbf{q}_{\sigma_j},\bar{\theta}_{n}\mathbf{q}_{\sigma_k}) |<\epsilon
\end{align*}
\end{lemma}

\begin{proof}[Proof of Lemma~\ref{lem:N_epsilon}]
For convenience, let $s=\sigma_j,t=\sigma_k$
\begin{align}
&|D^*(\theta_{-j}\mathbf{q}_{s},\theta_{-k}\mathbf{q}_{t})-D^*(\bar{\theta}_{n}\mathbf{q}_{s},\bar{\theta}_{n}\mathbf{q}_{t})|\\
=& |\sum_u \left(\sqrt{\theta_{-j}(u,\cdot) \mathbf{q}_s}-\sqrt{\theta_{-k}(u,\cdot) \mathbf{q}_t}\right)^2-\left(\sqrt{\bar{\theta}_n(u,\cdot) \mathbf{q}_s}-\sqrt{\bar{\theta}_n(u,\cdot) \mathbf{q}_t}\right)^2|\\ 
=& |\sum_u \left(\sqrt{\theta_{-j}(u,\cdot) \mathbf{q}_s}-\sqrt{\theta_{-k}(u,\cdot) \mathbf{q}_t}-\sqrt{\bar{\theta}_n(u,\cdot) \mathbf{q}_s}+\sqrt{\bar{\theta}_n(u,\cdot) \mathbf{q}_t}\right)* \nonumber \\
&\left(\sqrt{\theta_{-j}(u,\cdot) \mathbf{q}_s}-\sqrt{\theta_{-k}(u,\cdot) \mathbf{q}_t}+\sqrt{\bar{\theta}_n(u,\cdot) \mathbf{q}_s}-\sqrt{\bar{\theta}_n(u,\cdot) \mathbf{q}_t}\right) |\\ \label{e:n_epsilon:2}
\leq & 2 * m * \max_{u} \left(|\sqrt{\theta_{-j}(u,\cdot) \mathbf{q}_s}-\sqrt{\bar{\theta}_n(u,\cdot) \mathbf{q}_s}|+|\sqrt{\theta_{-k}(u,\cdot) \mathbf{q}_t}-\sqrt{\bar{\theta}_n(u,\cdot) \mathbf{q}_t}|\right) \\
\leq & 4 * m * \max_{u,s,j} |\sqrt{\theta_{-j}(u,\cdot) \mathbf{q}_s}-\sqrt{\bar{\theta}_n(u,\cdot) \mathbf{q}_s}|
\end{align}

The first equality follows from the definition of Helinger-divergence.

The second equality is just formula for the difference of square.

To arrive at (\ref{e:n_epsilon:2}), $\sum_u|\left(\sqrt{\theta_{-j}(u,\cdot) \mathbf{q}_s}-\sqrt{\theta_{-k}(u,\cdot) \mathbf{q}_t}+\sqrt{\bar{\theta}_n(u,\cdot) \mathbf{q}_s}-\sqrt{\bar{\theta}_n(u,\cdot) \mathbf{q}_t}\right)|\leq \sum_u 2=2m$ where the inequality follows from the fact $0<D^*<1$. 

The last equality follows since both $|\sqrt{\theta_{-j}(u,\cdot) \mathbf{q}_s}-\sqrt{\bar{\theta}_n(u,\cdot) \mathbf{q}_s}|$ and $|\sqrt{\theta_{-k}(u,\cdot) \mathbf{q}_t}-\sqrt{\bar{\theta}_n(u,\cdot) \mathbf{q}_t}|$ are less than $\max_{u,s,j} |\sqrt{\theta_{-j}(u,\cdot) \mathbf{q}_s}-\sqrt{\bar{\theta}_n(u,\cdot) \mathbf{q}_s}|$. 

Now we consider two cases for any $u,s,j$:

(1) $ |\sqrt{\theta_{-j}(u,\cdot) \mathbf{q}_s}-\sqrt{\bar{\theta}_n(u,\cdot) \mathbf{q}_s}|\leq \frac{\epsilon}{4*m}$: It is clear the result in this Lemma follows. 

(2) $ |\sqrt{\theta_{-j}(u,\cdot) \mathbf{q}_s}-\sqrt{\bar{\theta}_n(u,\cdot) \mathbf{q}_s}|> \frac{\epsilon}{4*m}$: Notice that $(n-1)\theta_{-j}=n \bar{\theta}_n-\theta_j$, then we can see $$  \theta_{-j} =\bar{\theta}_n+\frac{1}{n}(\theta_{-j}-\theta_j) $$

\begin{align}
& 4*m*|\sqrt{\theta_{-j}(u,\cdot) \mathbf{q}_s}-\sqrt{\bar{\theta}_n(u,\cdot) \mathbf{q}_s}|\\
& =  4*m*|\frac{\theta_{-j}(u,\cdot) \mathbf{q}_s-\bar{\theta}_n(u,\cdot) \mathbf{q}_s}{\sqrt{\theta_{-j}(u,\cdot) \mathbf{q}_s}+\sqrt{\bar{\theta}_n(u,\cdot) \mathbf{q}_s}}|\\
& =  4*m*\frac{\frac{1}{n}|(\theta_{-j}(u,\cdot)-\theta_j(u,\cdot))\mathbf{q}_s|}{\sqrt{\theta_{-j}(u,\cdot) \mathbf{q}_s}+\sqrt{\bar{\theta}_n(u,\cdot) \mathbf{q}_s}}\\
& <  4*m*2*\frac{4*m}{\epsilon} \frac{1}{n}<\epsilon
\end{align}
when $n>N_{\epsilon}=\frac{32*m^2}{\epsilon^2}$

The first equality follows from the formula of the difference of squares. 

The second equality follows from $  \theta_{-j} =\bar{\theta}_n+\frac{1}{n}(\theta_{-j}-\theta_j) $. 

If $ |\sqrt{\theta_{-j}(u,\cdot) \mathbf{q}_s}-\sqrt{\bar{\theta}_n(u,\cdot) \mathbf{q}_s}|> \frac{\epsilon}{4*m}$, we have $ |\sqrt{\theta_{-j}(u,\cdot) \mathbf{q}_s}+\sqrt{\bar{\theta}_n(u,\cdot) \mathbf{q}_s}|> \frac{\epsilon}{4*m}$ as well, the third line follows.

\end{proof}

%If $s^*$ is the equilibrium that obtains the highest classification score, then the classification score of $s^*$ is greater than the classification score of truth-telling. Now combining the results of Lemma~\ref{lem:N_epsilon} and Lemma~\ref{lem:N_epsilon}, we can prove that the classification score of $s^*$ must be close to the classification score of truth-telling which is stated in Corollary~\ref{coro:ineq_cor}. 

\paragraph{(3)  \textbf{$\textit{ClassificationScore}(s^*)\geq \textit{ClassificationScore}(truthtelling)\\
\Rightarrow \textit{TotalDivergence}(truthtelling)\approx \textit{TotalDivergence}(symmetrized\ s^*_{BP})$ when the number of agents is sufficient large.}  }

The below corollary is derived from Lemma~\ref{lem:N_epsilon}. It will imply not only\\ $\textit{TotalDivergence}(truthtelling)\approx \textit{TotalDivergence}(symmetrized\ s^*_{BP})$ but also any equilibrium cannot have agent-welfare (\textit{ClassificationScore}) that is much greater than truth-telling when the number of agents is sufficient large. 

\begin{corollary} \label{coro:ineq_cor}
Given any SNIFE prior $Q$, for any $\epsilon>0$, if $n>N_{\epsilon}=\frac{128*m^2}{\epsilon^2}$, for any equilibrium $s^*$ that has greater \textit{ClassificationScore} than the truth-telling \textit{ClassificationScore} minus $\epsilon/2$:
\begin{align*}
&\textit{Classification}(truthtelling)\\
&< \textit{Classification}(s^*)+\frac{\epsilon}{2}\\
&<\textit{TotalDivergence}(symmetrized\ s^*_{BP})+\epsilon\\
&\leq  \textit{Classification}(truthtelling)+\epsilon
\end{align*}
\end{corollary}

\begin{proof}[Proof for Corollary~\ref{coro:ineq_cor}]

\begin{align*}
\textit{TotalDivergence}(truthtelling) & =\textit{ClassificationScore}(truthtelling) \\ & \leq\textit{ClassificationScore}(s^*)+\frac{\epsilon}{2}\\
& \leq \textit{TotalDivergence}(s^*_{BP})+\frac{\epsilon}{2} \\
& < \textit{TotalDivergence}(symmetrized\ s^*_{BP})+\epsilon \\
& \leq \textit{ClassificationScore}(truthtelling)+\epsilon 
\end{align*}

The first equality follows from Corollary~\ref{coro:truth}. 

The second inequality follows from the condition. 

The third inequality follows from the main lemma. 

The fourth inequality follows from Lemma~\ref{lem:N_epsilon}. 

The last inequality follows from information monotonicity since 
\begin{align*}
&\textit{ClassificationScore}(truthtelling)-\textit{TotalDivergence}(symmetrized\ s^*_{BP})\\
&= \sum_{\substack{j\\k\neq j}}\sum_{\substack{\sigma_j,\sigma_k}} Pr(j,k)Pr(\sigma_j,\sigma_k) D^*( \mathbf{q}_{\sigma_j}, \mathbf{q}_{\sigma_k})-\sum_{j,k,\sigma_j,\sigma_k} Pr(j,k)Pr(\sigma_j,\sigma_k) D^*(\bar{\theta}_n \mathbf{q}_{\sigma_j},\bar{\theta}_n \mathbf{q}_{\sigma_k})\\
&= \sum_{j,k,\sigma_j,\sigma_k} Pr(j,k)Pr(\sigma_j,\sigma_k) (D^*( \mathbf{q}_{\sigma_j}, \mathbf{q}_{\sigma_k})-D^*(\bar{\theta}_n \mathbf{q}_{\sigma_j},\bar{\theta}_n \mathbf{q}_{\sigma_k}))\geq 0
\end{align*}

The second equality follows since if $j=k$, $D^*( \mathbf{q}_{\sigma_j}, \mathbf{q}_{\sigma_k})=0$

\end{proof}

This corollary induces the following result:

\paragraph{No equilibrium can have agent-welfare that is much greater than truth-telling}

Let $\frac{\epsilon}{2}=\gamma_2$, we need $n\geq \frac{128*m^2}{\epsilon^2}$ to obtain $\gamma_2$ tolerance based on Corollary~\ref{coro:ineq_cor}. By manipulations, we obtain our result.

\paragraph{(3) $\Rightarrow$ (4) $\mathbf{\bar{\theta}_n\approx \pi}$}

We already know that if $s^*$ obtains higher \textit{ClassificationScore} than truth-telling, the classification score of truth-telling is close to that of the symmetrized $s^*_{BP}$. Now we will prove that $\bar{\theta}_n$ is close to a permutation where $\bar{\theta}_n$ is the average signal strategy of $s^*$. We prove it by contradiction. We first assume that $\bar{\theta}_n$ is far from a permutation equilibrium, that is, recalling the definition of $\tau$-close, we assume there exists a row of $\bar{\theta}_n$ that has at least two large numbers. Formally, we assume that given any $\tau$, there exists $u',v',w'\in \Sigma$ such that $\bar{\theta}_n(u',v')>\tau,\bar{\theta}_n(u',w')>\tau $. We will prove that it is impossible when $n>N(\tau,Q)$ since when $\bar{\theta}_n$ is far from a permutation, the total divergence of symmetrized $s_{BP}$ is far from the classification score of truth-telling which contradicts Corollary~\ref{coro:ineq_cor}.  

The below lemma tells us if a symmetric strategy $s_{\theta}$, where agents play $\theta$ as their signal strategy and best prediction as their reported prediction, is far from a permutation, then the total divergence of $s_{\theta}$ and total divergence of permutation are also far. Once we proved the below lemma, we can replace $\theta$ by $\bar{\theta}_n$ to finish our proof for Theorem~\ref{main_thm} part 4, that is, truth-telling is robust-approximate-quasi-focal. 

\begin{lemma} \label{lem:general case lem}
Given any fixed $\tau$, for any signal strategy $\theta$, if there exists $u',v',w'\in \Sigma$ such that $\theta(u',v')>\tau,\theta(u',w')>\tau $, then consider the case $s_{\theta}$ that all agents play $\theta$ as their signal strategy and report their best prediction, we have 
$$ \textit{TotalDivergence}(truthtelling)- \textit{TotalDivergence}(s_{\theta})\geq c_2(\tau c_1)^3 c_4 c_3$$
\end{lemma}

\begin{proof}[Proof of Lemma~\ref{lem:general case lem}]
We first write \textit{TotalDivergence} in an explicit form:

\begin{align}
\sum_{j,k,\sigma_j,\sigma_k} Pr(j,k)Pr(\sigma_j,\sigma_k) D^*( \mathbf{q}_{\sigma_j}, \mathbf{q}_{\sigma_k})-\sum_{j,k,\sigma_j,\sigma_k} Pr(j,k)Pr(\sigma_j,\sigma_k) D^*(\theta \mathbf{q}_{\sigma_j},\theta \mathbf{q}_{\sigma_k})
\end{align}

Actually, We will show for any $j,k$,  \begin{align}\label{540}
\sum_{\sigma_j,\sigma_k} Pr(\sigma_j,\sigma_k) D^*( \mathbf{q}_{\sigma_j}, \mathbf{q}_{\sigma_k})-\sum_{\sigma_j,\sigma_k} Pr(\sigma_j,\sigma_k) D^*(\theta \mathbf{q}_{\sigma_j},\theta \mathbf{q}_{\sigma_k})
\end{align}
is greater than $c_2 (\tau c_1)^3 c_4 c_3$, which implies the result. 

We want give a lower bound for (\ref{540}). In order to obtain this lower bound, we are going to transform this value to $\sum_{u}\lambda_u g(x_u)- g(\sum_u \lambda_u x_u)$ where $g(\cdot)$ is a convex function. To obtain a lower bound of $\sum_{u}\lambda_u g(x_u)- g(\sum_u \lambda_u x_u)$, we have an observation:

For any convex function $g(\cdot)$, $g(\sum_u \lambda_u x_u)$ and $\sum_{u}\lambda_u g(x_u)$ are ``very different'' if there are two large coefficients $\lambda_1$ and $\lambda_2$ with the corresponding $x_1$ and $x_2$ that are ``very different''.  Now we introduce a claim to show this observation. 

\begin{claim}\label{claim:convex}
$$ \sum_{u}\lambda_u g(x_u)- g(\sum_u \lambda_u x_u)\geq \frac{d_2(g)}{2} \frac{\lambda_1 \lambda_2}{\lambda_1 + \lambda_2}||x_1-x_2||^2
$$
where $d_2(g)$ is a lower bound of $g''(\cdot)$
\end{claim}

\begin{proof}
\begin{align*}
 g\left(\sum_u \lambda_u x_u\right)\leq (\lambda_1+\lambda_2)g\left(\frac{\lambda_1 x_1+\lambda_2 x_2}{\lambda_1+\lambda_2}\right) + \sum_{u>2}\lambda_u g(x_u)\leq \sum_{u}\lambda_u g(x_u)
\end{align*}

So 
\begin{align*}
& \sum_{u}\lambda_u g(x_u)- g(\sum_u \lambda_u x_u)\\
&\geq  \sum_{u}\lambda_u g(x_u)- (\lambda_1+\lambda_2)g\left(\frac{\lambda_1 x_1+\lambda_2 x_2}{\lambda_1+\lambda_2}\right) - \sum_{u>2}\lambda_u g(x_u)\\
&=  \lambda_1 g(x_1)+\lambda_2 g(x_2)-(\lambda_1+\lambda_2)g(\frac{\lambda_1 x_1+\lambda_2 x_2}{\lambda_1+\lambda_2})\\
&=  (\lambda_1+\lambda_2) (\frac{ \lambda_1 g(x_1)+\lambda_2 g(x_2)}{\lambda_1+\lambda_2}-g(\frac{\lambda_1 x_1+\lambda_2 x_2}{\lambda_1+\lambda_2}))\\
& \geq (\lambda_1+\lambda_2) \frac{d_2(g)}{2}\frac{\lambda_1 \lambda_2}{(\lambda_1 + \lambda_2)^2}||x_1-x_2||^2\\
&= \frac{d_2(g)}{2} \frac{\lambda_1 \lambda_2}{\lambda_1 + \lambda_2}||x_1-x_2||^2\\
\end{align*}
where $d_2(g)$ is the lower bound of $g''(\cdot)$

The first inequality follows if we rewrite $\sum_u\lambda_u x_u$ as $(\lambda_1+\lambda_2)\frac{\lambda_1 x_1+\lambda_2 x_2}{\lambda_1+\lambda_2}+\sum_{u>2}\lambda_u x_u$ and apply convexity.

Then we do several manipulations including taking $\lambda_1+\lambda_2$ outside.  For continuous convex function $g(\cdot)$, we have $tg(x)+(1-t)g(y)-g(t x+(1-t) y)\geq \frac{1}{2}d_2(g)t(1-t)||x-y||^2$ according to \cite{nesterov2013introductory}, then we replace $t$ by $\frac{\lambda_1}{\lambda_1+\lambda_2}$ and set $x=x_1,y=x_2$ and obtain the final result. 

\end{proof}

We can think of $\theta(u',v')$ and $\theta(u',w')$ as the two large coefficients (actually they are part of the coefficients). Then we need to find two ``very different'' entries that corresponding to those large coefficients. We pick two specific signals $s',t'\in \Sigma$ such that $\mathbf{q}_{s'}$ and $\mathbf{q}_{t'}$ are ``very different'' in position $v'$ and $w'$. The reason we do this is that when we compute $\theta \mathbf{q}$, $\theta(u',v')$ and $\theta(u',w')$ are the two large entries which correspond to the positions $v'$ and $w'$ in $\mathbf{q}$. Formally, we pick $s',t'\in \Sigma$ such that 
$$\left\Vert \frac{q(v'|s')}{q(v'|t')}-   \frac{q(w'|s')}{q(w'|t')}\right\Vert=\max_{s,t} \left\Vert\frac{q(v'|s)}{q(v'|t)}-   \frac{q(w'|s)}{q(w'|t)}\right\Vert $$

Once we have chosen the two specific signals, since $Pr(s',t')(D^*( \mathbf{q}_{s'}, \mathbf{q}_{t'})-D^*(\theta \mathbf{q}_{s'},\theta \mathbf{q}_{t'}))$ is less than (\ref{540}) based on the fact $D^*( \mathbf{q}_s, \mathbf{q}_t)-D^*(\theta \mathbf{q}_s,\theta \mathbf{q}_t)\geq 0$ for $s,t\neq s',t'$, we will give a lower bound of $Pr(s',t')(D^*( \mathbf{q}_{s'}, \mathbf{q}_{t'})-D^*(\theta \mathbf{q}_{s'},\theta \mathbf{q}_{t'}))$ which is also a lower bound of (\ref{540}). 

Let $f(x)=(\sqrt{x}-1)^2$. For convenience, we will write the dot product of two vectors $\sum_v a(v)b(v)$ as $a(\cdot)b(\cdot)$. Now we give a explicit form of $D^*$:

\begin{align}
& Pr(s',t')(D^*( \mathbf{q}_{s'}, \mathbf{q}_{t'})-D^*(\theta \mathbf{q}_{s'},\theta \mathbf{q}_{t'}))\\ \label{542}
& =  Pr(s',t')\left(\sum_v q(v|s')f\left(\frac{q(v|t')}{q(v|s')}\right)-\sum_u\theta(u,\cdot)q(\cdot|s') f\left(\frac{1}{\theta(u,\cdot)q(\cdot|s')}\theta(u,\cdot)q(\cdot|t')\right)\right) 
\end{align}

We take $\sum_u\theta(u,\cdot)q(\cdot|s')$ out and note that $\sum_u\theta(u,v)=1$, so $\sum_u \theta(u,\cdot)q(\cdot|s') \frac{1}{\theta(u,\cdot)q(\cdot|s')}\theta(u,v)=1$, then we obtain (\ref{543}) from (\ref{542}). 
\begin{align}
(\ref{542})= & Pr(s',t') \sum_u \theta(u,\cdot)q(\cdot|s')* \nonumber \\
& \left(\frac{1}{\theta(u,\cdot)q(\cdot|s')}\sum_v\theta(u,v)q(v|s')f\left(\frac{q(v|t')}{q(v|s')}\right) -f\left(\frac{1}{\theta(u,\cdot)q(\cdot|s')}\sum_v\theta(u,v)q(v|s')\frac{q(v|t')}{q(v|s')}\right)\right) \label{543}
\end{align}

Then we pick the special $u'$ to obtain (\ref{544}). For the part $\sum_{u\neq u'}$, since $f(\cdot)$ is a convex function, we have 
$$\frac{1}{\theta(u,\cdot)q(\cdot|s')}\sum_v\theta(u,v)q(v|s')f(\frac{q(v|t')}{q(v|s')}) \geq f\left(\frac{1}{\theta(u,\cdot)q(\cdot|s')}\sum_v\theta(u,v)q(v|s')\frac{q(v|t')}{q(v|s')}\right)$$ 
so (\ref{543}) is greater than (\ref{544}). 

\begin{align}
(\ref{543})\geq & Pr(s',t') \theta(u',\cdot)q(\cdot|s')* \nonumber \\ 
&\left(\frac{1}{\theta(u',\cdot)q(\cdot|s')}\sum_v\theta(u',v)q(v|s')f\left(\frac{q(v|t')}{q(v|s')}\right) -f\left(\frac{1}{\theta(u',\cdot)q(\cdot|s')}\sum_v\theta(u',v)q(v|s')\frac{q(v|t')}{q(v|s')}\right)\right) \label{544}
\end{align}

Note that $\theta(u',v')$ and $\theta(u',w')$ are large, so in the convex function $f(\cdot)$, there are two large coefficients $\frac{1}{\theta(u',\cdot)q(\cdot|s')}\theta(u',v')q(v'|s')$ and $\frac{1}{\theta(u',\cdot)q(\cdot|s')}\theta(u',w')q(w'|s')$ which correspond to $\frac{q(v'|t')}{q(v'|s')}$ and $\frac{q(w'|t')}{q(w'|s')}$. Then based on our choice for $s',t'$ and Claim~\ref{claim:convex}, we have 
\begin{align}
(\ref{544}) \geq & Pr(s',t') \theta(u',\cdot)q(\cdot|s') \frac{c_4}{2} \left(
\frac{ (\theta(u',v')q(v'|s'))*(\theta(v',w')q(w'|s'))}{ \theta(u',v')q(v'|s')+\theta(v',w')q(w'|s')}
\left\Vert\frac{q(v'|t')}{q(v'|s')}-\frac{q(w'|t')}{q(w'|s')}\right\Vert^2 \right)\\  \label{545}
\geq & c_2 (\tau c_1)^3 c_4 c_3
\end{align}

The last inequality follows since $Pr(s',t')\geq c_2$, both $\theta(u',v')q(v'|s')$ and $\theta(v',w')q(w'|s')$ are greater than $\tau c_1$.  Also note that:

 $$\theta(u',\cdot)q(\cdot|s')\geq \theta(u',v')q(v'|s')+\theta(v',w')q(w'|s')\geq 2 \tau c_1$$ and
 
 $$\theta(u',v')q(v'|s')+\theta(v',w')q(w'|s')\leq 1.$$

\end{proof}

\paragraph{Theorem~\ref{main_thm} Part 3: $\mathcal{M+}(\alpha,\beta,PS(\cdot,\cdot))$ has truth-telling as a \emph{($\tau_1,\gamma_1$)-robust-symmetric-quasi-focal}:}

\paragraph{Any \emph{symmetric} equilibrium that has agent-welfare close to truth-telling must be close to a permutation equilibrium: }

We have already proved that no symmetric equilibrium pays more than truth-telling. For the symmetric equilibrium $s^*$ such that $\textit{ClassificationScore}(s^*)>\textit{ClassificationScore}(truthtelling)-\gamma_1$, we have

\begin{align*}
&\textit{TotalDivergence}(truthtelling)=\textit{ClassificationScore}(truthtelling)\leq\textit{ClassificationScore}(s^*)+\gamma_1\\
\leq&\textit{TotalDivergence}(s^*_{BP})+\gamma_1
\leq \textit{ClassificationScore}(truthtelling)+\gamma_1 
\end{align*}

Let $\gamma_1=(\tau_1 c_1)^3 c_2 c_3 c_4$, then $s^*$ is $\tau_1$ close to a permutation equilibrium or there will be a contradiction based on Lemma~\ref{lem:general case lem}. By manipulations, we will obtain our result. 

\paragraph{Theorem~\ref{main_thm} Part 4: $\mathcal{M+}(\alpha,\beta,PS(\cdot,\cdot))$ has truth-telling as a \emph{($\tau_2,\gamma_2$)-robust-approximate-quasi-focal}:}

\paragraph{If the number of agents is sufficient large, any equilibrium that has agent-welfare close to truth-telling must be close to permutation equilibrium:}

Let $\epsilon=(\tau_2 c_1)^3 c_2 c_3 c_4$, if $n> \frac{32*m^2}{(\epsilon/2)^2}$, we have already proved that $\textit{TotalDivergence}(truthtelling)-\textit{TotalDivergence}(symmetrized\ s^*_{BP})<\epsilon$ based on Corollary~\ref{coro:ineq_cor}. If $s^*$ is not $\tau_2$ close to a permutation equilibrium, we will have 
$$ \textit{TotalDivergence}(truthtelling)-\textit{TotalDivergence}(symmetrized\ s^*_{BP}) > (\tau_2 c_1)^3 c_2 c_3 c_4 =\epsilon $$
which is a contradiction based on Lemma~\ref{lem:general case lem}. By manipulations, we obtain our result.

\subsection{Proof for Main Lemma}\label{section:proof_main_lemma}
In this section, we will pove the main lemma: the classification score of non-permutation equilibrium $s$ is less than the total divergence of the report profiles when agents report their best predictions given they still use the signal strategy of $s$. We first show the inequality and then show that if the equality holds, then $s$ is consistent and $s=s_{BP}$.

In order to show the inequality, we first show 
$$ \textit{TotalDivergence}(s)-\textit{TotalDivergence}(s_{BP})\leq \textit{Inconsistency}(s) $$
once we show this, since we have $\textit{ClassificationScore}= \textit{Diversity}-\textit{Inconsistency}$ and $\textit{Diversity}\leq \textit{TotalDivergence}$, our main lemma $ \textit{ClassificationScore}(s) \leq \textit{TotalDivergence}(s_{BP}) $  will follow since
\begin{align}
\textit{ClassificationScore}(s)=&\textit{Diversity}(s)-\textit{Inconsistency}(s)\nonumber \\ \label{eq:main:lemma}
\leq &\textit{TotalDivergence}(s)-\textit{Inconsistency}(s)\leq \textit{TotalDivergence}(s_{BP})
\end{align}

%This equation helps use prove the below claim: 

To prove $ \textit{TotalDivergence}(s)-\textit{TotalDivergence}(s_{BP})\leq \textit{Inconsistency}(s) $, we will write it in a explicit form: 

%I may need to put this paragraph in the previous section. 

\begin{align} \label{eq:td_td}
& \textit{TotalDivergence}(s)-\textit{TotalDivergence}(s_{BP})\\
= & \sum_{j,k,\sigma_j,\sigma_k} Pr(j,k)Pr(\sigma_j,\sigma_k) \int_{\hat{j},\hat{k}} Pr(\hat{j},\hat{k}) (D^*(\mathbf{\hat{p}}_j,\mathbf{\hat{p}}_k)-D^*(\theta_{-j}\mathbf{q}_{\sigma_j},\theta_{-k}\mathbf{q}_{\sigma_k}))
\end{align}

It is difficult to compare $D^*(\mathbf{\hat{p}}_j,\mathbf{\hat{p}}_k)$ and $D^*(\theta_{-j}\mathbf{q}_{\sigma_j},\theta_{-k}\mathbf{q}_{\sigma_k})$ directly. To deal with this problem, we introduce a new value $D^*(\mathbf{\hat{p}}_j,\theta_{-k}\mathbf{q}_{\sigma_k})$ and write (\ref{eq:td_td}) as
\begin{align}
& \sum_{j,k,\sigma_j,\sigma_k} Pr(j,k)Pr(\sigma_j,\sigma_k)* \nonumber \\  & \int_{\hat{j},\hat{k}} Pr(\hat{j},\hat{k}) \left(D^*(\mathbf{\hat{p}}_j,\mathbf{\hat{p}}_k)-D^*(\mathbf{\hat{p}}_j,\theta_{-k}\mathbf{q}_{\sigma_k})+D^*(\mathbf{\hat{p}}_j,\theta_{-k}\mathbf{q}_{\sigma_k})-D^*(\theta_{-j}\mathbf{q}_{\sigma_j},\theta_{-k}\mathbf{q}_{\sigma_k})\right)
\end{align}

We will first give the analysis for $D^*(\mathbf{\hat{p}}_j,\mathbf{\hat{p}}_k)-D^*(\mathbf{\hat{p}}_j,\theta_{-k}\mathbf{q}_{\sigma_k})$, then we will see $D^*(\mathbf{\hat{p}}_j,\theta_{-k}\mathbf{q}_{\sigma_k})-D^*(\theta_{-j}\mathbf{q}_{\sigma_j},\theta_{-k}\mathbf{q}_{\sigma_k})$ is similar. 

Remember that both $D^*(a,\cdot)$ and $D^*(\cdot,b)$ are convex functions. So $D^*(\mathbf{\hat{p}}_j,\mathbf{\hat{p}}_k)-D^*(\mathbf{\hat{p}}_j,\theta_{-k}\mathbf{q}_{\sigma_k})$ can be seen as $g(\mathbf{\hat{p}}_k)-g(\theta_{-k}\mathbf{q}_{\sigma_k})$ where $g(\cdot)$ is convex function $D^*(\mathbf{\hat{p}}_j,\cdot)$. 

Recall that 

\begin{align*}
\textit{Inconsistency}= \sum_{\substack{j\\k\neq j}}\sum_{\substack{\sigma_j,\sigma_k}} Pr(j,k)Pr(\sigma_j,\sigma_k) \int_{\hat{j},\hat{k}} Pr(\hat{j},\hat{k})\delta(\hat{\sigma}_j = \hat{\sigma}_k) \sqrt{D^*(\mathbf{\hat{p}}_j,\mathbf{\hat{p}}_k)}
\end{align*} 

We hope we can obtain a upper bound for $g(\mathbf{\hat{p}}_k)-g(\theta_{-k}\mathbf{q}_{\sigma_k})$ that relates to agent $k$'s neighbors' best response predictions. Here agent $k$'s \emph{neighbors} mean the agents who report the same signal with agent $k$ and \emph{best response prediction} means the reported prediction at equilibrium.

%\yk{I will rewrite this paragraph}
Now we begin to analyze the relationship between $\mathbf{\hat{p}}_k$ and $\theta_{-k}\mathbf{q}_{\sigma_k}$. Recall that each agent's payment depends on his prediction score and information score. $\theta_{-k}\mathbf{q}_{\sigma_k}$ maximizes the prediction score while $\mathbf{\hat{p}}_k$ maximizes the payment. The information score depends on agent $k$'s neighbors' reported predictions $\{\hat{\mathbf{p}}_l|l\neq k\}$. So we can see $\mathbf{\hat{p}}_k$ is related to both his best prediction $\theta_{-k}\mathbf{q}_{\sigma_k}$ and his neighbors' reported predictions $\{\hat{\mathbf{p}}_l|l\neq k\}$. Actually we will show that $\hat{\mathbf{p}}_k$ can be computed as a linear combination of $\theta_{-k}\mathbf{q}_{\sigma_k}$ and $\{\hat{\mathbf{p}}_l|l\neq k\}$, which is based on the fact that every proper scoring rule is linear for the first entry (we will discuss the detail in the below proof). Once we have this result, we can construct a linear system about agents' reported predictions $\{\hat{\mathbf{p}}_i|i\}$ and their best predictions. This linear system helps us obtain a upper bound for $g(\mathbf{\hat{p}}_k)-g(\theta_{-k}\mathbf{q}_{\sigma_k})$ which upper-bounds the distance between agent $k$'s best response prediction and his neighbors' best response predictions. 

\paragraph{Equilibrium Analysis}
%\subsection{Equilibrium Analysis}\label{subsection:equilibrium analysis}
We will analyze the equilibrium in our \emph{Truthful Mechanism} which is also the equilibrium in our \emph{Disagreement Mechanism}. We first show, in Claim~\ref{claim:actual prediction}, that at equilibrium, an agent's reported prediction only depends on his private signal and reported signal.  Then we use this property to construct a linear system and via this linear system, we obtain a upper bound for $g(\mathbf{\hat{p}}_k)-g(\theta_{-k}\mathbf{q}_{\sigma_k})$ in Claim~\ref{claim:equilibrium}. 

\begin{claim}\label{claim:actual prediction}
At any equilibrium $s=(s_1,...,s_n)$, for each agent $i$, fix $s_{-i}$, agent $i$'s private signal $\sigma_i\in \Sigma$ and reported signal $\hat{\sigma}_i\in \Sigma$, then there exists a unique prediction which is agent $i$'s best response. 

\end{claim}
We define this unique prediction as $\hat{\mathbf{p}}(i,\sigma_i,\hat{\sigma}_i)$

In other words, $s_i(\sigma_i)$ is a distribution over at most $m$ vectors: $\{(\hat{\sigma}_i,\hat{\mathbf{p}}(i,\sigma_i,\hat{\sigma}_i))|\hat{\sigma}_i\in \Sigma\}$ and $$Pr_{(\hat{\sigma}_i,\hat{\mathbf{p}}_i)\leftarrow s_i(\sigma)}(\hat{\sigma}_i,\hat{\mathbf{p}}_i)=\left\{
\begin{aligned}
\theta_i(\hat{\sigma}_i,\sigma_i) &  & \hat{\mathbf{p}}_i = \hat{\mathbf{p}}(i,\sigma_i,\hat{\sigma}_i) \\
0 &  &  \hat{\mathbf{p}}_i \neq  \hat{\mathbf{p}}(i,\sigma_i,\hat{\sigma}_i)
\end{aligned}\right.$$ 

\begin{proof}
For any agent $i$, assume his private signal is $\sigma_i$ and he reports $\hat{\sigma}_i$ at equilibrium $(s_1,s_2,...,s_n)$. Now we will prove there is a unique prediction that maximize agent $i$'s payment. 

\begin{align}\label{eq:actual prediction}
&\arg \max_{\hat{\mathbf{p}}} \E[payment(i,\mathcal{M}+)|\sigma_i]\\\label{eq:middle}
&=  \arg\max_{\hat{p}}\alpha PS(\theta_{-i}\mathbf{q}_{\sigma_i},\hat{\mathbf{p}}) + \beta \sum_{j\neq i}Pr(j)\sum_{\sigma_j} Pr(\sigma_j|\sigma_i)\int_{\hat{\sigma}_j,\hat{\mathbf{p}}_j}Pr_{(\hat{\sigma}_j,\hat{\mathbf{p}}_j)\leftarrow s_j(\sigma_j)}(\hat{\sigma}_j,\hat{\mathbf{p}}_j)\delta(\hat{\sigma}_i=\hat{\sigma}_j)PS(\mathbf{\hat{p}}_j,\hat{\mathbf{p}})\\ \label{eq:hatp}
&=  \arg\max_{\hat{p}} \left( \alpha+ \beta \sum_{j\neq i}Pr(j)\sum_{\sigma_j} Pr(\sigma_j|\sigma_i)\int_{\hat{\sigma}_j,\hat{\mathbf{p}}_j}Pr_{(\hat{\sigma}_j,\hat{\mathbf{p}}_j)\leftarrow s_j(\sigma_j)}(\hat{\sigma}_j,\hat{\mathbf{p}}_j)\delta(\hat{\sigma}_i=\hat{\sigma}_j) \right)\nonumber \\ &PS(\frac{\alpha \theta_{-i}\mathbf{q}_{\sigma_i} + \beta \sum_{j\neq i}Pr(j)\sum_{\sigma_j} Pr(\sigma_j|\sigma_i)\int_{\hat{\sigma}_j,\hat{\mathbf{p}}_j}Pr_{(\hat{\sigma}_j,\hat{\mathbf{p}}_j)\leftarrow s_j(\sigma_j)}(\hat{\sigma}_j,\hat{\mathbf{p}}_j)\delta(\hat{\sigma}_i=\hat{\sigma}_j)\mathbf{\hat{p}}_j}{\alpha+ \beta \sum_{j\neq i}Pr(j)\sum_{\sigma_j} Pr(\sigma_j|\sigma_i)\int_{\hat{\sigma}_j,\hat{\mathbf{p}}_j}Pr_{(\hat{\sigma}_j,\hat{\mathbf{p}}_j)\leftarrow s_j(\sigma_j)}(\hat{\sigma}_j,\hat{\mathbf{p}}_j)\delta(\hat{\sigma}_i=\hat{\sigma}_j)},\hat{\mathbf{p}})\\
&=  \frac{\alpha \theta_{-i}\mathbf{q}_{\sigma_i} + \beta \sum_{j\neq i}Pr(j)\sum_{\sigma_j} Pr(\sigma_j|\sigma_i)\int_{\hat{\sigma}_j,\hat{\mathbf{p}}_j}Pr_{(\hat{\sigma}_j,\hat{\mathbf{p}}_j)\leftarrow s_j(\sigma_j)}(\hat{\sigma}_j,\hat{\mathbf{p}}_j)\delta(\hat{\sigma}_i=\hat{\sigma}_j)\mathbf{\hat{p}}_j}{\alpha+ \beta \sum_{j\neq i}Pr(j)\sum_{\sigma_j} Pr(\sigma_j|\sigma_i)\int_{\hat{\sigma}_j,\hat{\mathbf{p}}_j}Pr_{(\hat{\sigma}_j,\hat{\mathbf{p}}_j)\leftarrow s_j(\sigma_j)}(\hat{\sigma}_j,\hat{\mathbf{p}}_j)\delta(\hat{\sigma}_i=\hat{\sigma}_j)}
\end{align}

In equation (\ref{eq:middle}), the first part is the prediction score of agent $i$, the second part is part of the information score of agent $i$. Note that for the information score $PS(\hat{\mathbf{p}}_j,\hat{\mathbf{p}})-PS(\hat{\mathbf{p}}_j,\hat{\mathbf{p}}_j)$ of agent $i$, only $PS(\hat{\mathbf{p}}_j,\hat{\mathbf{p}})$ is related to agent $i$'s reported prediction $\hat{\mathbf{p}}$ so we only consider this part to analyze the equilibrium. $Pr(j)$ is the probability that agent $j$ is matched with agent $i$, $Pr(\sigma_j|\sigma_i)$ is the probability that agent $j$ receives $\sigma_j$ given agent $i$ receives $\sigma_i$. Then given agent $j$'s strategy $s_j$ and private signal, we integrate over agent $j$ possible report profiles and only consider the case $\hat{\sigma}_i=\hat{\sigma}_j$. 

The second equality follows since proper scoring rule is linear for the first entry. 

The last equality follows since we obtain the highest value only if $\hat{\mathbf{p}}$ equals the first entry based on the property of strict proper scoring rule.

\end{proof}

The below claim tells us we can bound the distance between each agent's best response prediction (the prediction which maximizes his total reward) and his best prediction (the prediction which maximizes his prediciton score) by the distance between his best response prediction and his neighbors' best response predictions. 

\begin{claim}\label{claim:equilibrium}
For any convex function $g(\cdot)$, for any $\sigma_i$ and $\hat{\sigma}_i$, we have $$ \alpha Pr(\sigma_i)(g(\mathbf{\hat{p}}(i,\sigma_i,\hat{\sigma}_i))-g(\theta_{-i}\mathbf{q}_{\sigma_i}))\leq \beta \sum_{j\neq i}Pr(j)\sum_{\sigma_j}Pr(\sigma_j,\sigma_i)\theta_j(\hat{\sigma_i},\sigma_j) (g(\mathbf{\hat{p}}(j,\sigma_j,\hat{\sigma}_i))-g(\mathbf{\hat{p}}(i,\sigma_i,\hat{\sigma}_i))) $$
\end{claim}

\begin{proof}
Based on Claim~\ref{claim:actual prediction}, we can rewrite (\ref{eq:actual prediction})=(\ref{eq:hatp}) as a $n\times m\times m$ linear system about\\
$\{\mathbf{\hat{p}}(k,\sigma_k,\hat{\sigma}_k)|k\in[1,n],\sigma_k\in \Sigma,\hat{\sigma}_k\in\Sigma\}$:

\begin{align}\label{eq:linear system}
\mathbf{\hat{p}}(i,\sigma_i,\hat{\sigma}_i)&=\arg \max_{\hat{\mathbf{p}}} \E[payment(i,\mathcal{M}+)|\sigma_i]\\
= & \frac{\alpha \theta_{-i}\mathbf{q}_{\sigma_i} + \beta \sum_{j\neq i}Pr(j)\sum_{\sigma_j}Pr(\sigma_j|\sigma_i)\theta_j(\hat{\sigma}_i,\sigma_j)\mathbf{\hat{p}}(j,\sigma_j,\hat{\sigma}_i)}{\alpha+ \beta \sum_{j\neq i}Pr(j)\sum_{\sigma_j}Pr(\sigma_j|\sigma_i)\theta_j(\hat{\sigma}_i,\sigma_j)}
\end{align}

Fix $i$, let $\lambda_i = \frac{\alpha}{\alpha+ \beta \sum_{j\neq i}Pr(j)\sum_{\sigma_j}Pr(\sigma_j|\sigma_i)\theta_j(\hat{\sigma}_i,\sigma_j)}$, $\lambda_{j,\sigma_j} = \frac{\beta \sum_{j\neq i}Pr(j)\sum_{\sigma_j}Pr(\sigma_j|\sigma_i)\theta_j(\hat{\sigma}_i,\sigma_j)}{\alpha+ \beta \sum_{j\neq i}Pr(j)\sum_{\sigma_j}Pr(\sigma_j|\sigma_i)\theta_j(\hat{\sigma}_i,\sigma_j)}$ for $j\neq i$ and $\sigma_j\in\Sigma$, we have $\lambda_i+\sum_{j\neq i,\sigma_j}\lambda_{j,\sigma_j}=1$

Based on the convexity of $g(\cdot)$, we have 
\begin{align*}
g(\mathbf{\hat{p}}(i,\sigma_i,\hat{\sigma}_i))&= g(\lambda_i \theta_{-i}\mathbf{q}_{\sigma_i} + \sum_{j\neq i,\sigma_j}\lambda_{j,\sigma_j} \mathbf{\hat{p}}(j,\sigma_j,\hat{\sigma}_i) )\\
&\leq \lambda_i g(\theta_{-i}\mathbf{q}_{\sigma_i})+\sum_{j\neq i,\sigma_j}\lambda_{j,\sigma_j} g(\mathbf{\hat{p}}(j,\sigma_j,\hat{\sigma}_i))
\end{align*}

After substitutions, we multiply $\left(\alpha+ \beta Pr(j)\sum_{j\neq i}\sum_{\sigma_j}Pr(\sigma_j|\sigma_i)\theta_j(\hat{\sigma}_i,\sigma_j)\right) Pr(\sigma_i)$ in both sides. Note that $Pr(\sigma_i)Pr(\sigma_j|\sigma_i)=Pr(\sigma_j,\sigma_i)$, then by manipulation, the claim follows. 

%\yk{more explanation}

\end{proof}

Claim~\ref{claim:equilibrium} gives an upper bound to $g(\mathbf{\hat{p}}_k)-g(\theta_{-k}\mathbf{q}_{\sigma_k})$ that is the distance between agent $k$'s best response prediction and his neighbors' best response predictions. Now we continue the proof for our main lemma. 

To bound \begin{align}
& \sum_{j,k,\sigma_j,\sigma_k} Pr(j,k)Pr(\sigma_j,\sigma_k)* \nonumber \\  & \int_{\hat{j},\hat{k}} Pr(\hat{j},\hat{k}) (D^*(\mathbf{\hat{p}}_j,\mathbf{\hat{p}}_k)-D^*(\mathbf{\hat{p}}_j,\theta_{-k}\mathbf{q}_{\sigma_k})+D^*(\mathbf{\hat{p}}_j,\theta_{-k}\mathbf{q}_{\sigma_k})-D^*(\theta_{-j}\mathbf{q}_{\sigma_j},\theta_{-k}\mathbf{q}_{\sigma_k}))
\end{align}

We rewrite $\int_{\hat{j},\hat{k}} Pr(\hat{j},\hat{k})$ as $\theta_j(\hat{\sigma}_j,\sigma_j)\theta_k(\hat{\sigma}_k,\sigma_k)$ and $\hat{\mathbf{p}}_j$ as $\hat{\mathbf{p}}(j,\sigma_j,\hat{\sigma}_j)$, $\hat{\mathbf{p}}_k$ as $\hat{\mathbf{p}}(k,\sigma_k,\hat{\sigma}_k)$ which we can do because of Claim~\ref{claim:actual prediction}.

We first give an upper bound to 
$$\sum_{j,k} \sum_{\sigma_j,\sigma_k,\hat{\sigma}_j,\hat{\sigma}_k}Pr(j,k)Pr(\sigma_j,\sigma_k)\theta_j(\hat{\sigma}_j,\sigma_j)\theta_k(\hat{\sigma}_k,\sigma_k)(D^*(\mathbf{\hat{p}}(j,\sigma_j,\hat{\sigma}_j),\mathbf{\hat{p}}(k,\sigma_k,\hat{\sigma}_k))-D^*(\mathbf{\hat{p}}(j,\sigma_j,\hat{\sigma}_j),\theta_{-k} \mathbf{q}_{\sigma_k}))$$

The analysis for the second part is similar. 

Based on Claim~\ref{claim:equilibrium}, we have

\begin{align}
&\sum_{j,k} \sum_{\sigma_j,\sigma_k,\hat{\sigma}_j,\hat{\sigma}_k}Pr(j,k)Pr(\sigma_j,\sigma_k)\theta_j(\hat{\sigma}_j,\sigma_j)\theta_k(\hat{\sigma}_k,\sigma_k)(D^*(\mathbf{\hat{p}}(j,\sigma_j,\hat{\sigma}_j),\mathbf{\hat{p}}(k,\sigma_k,\hat{\sigma}_k))-D^*(\mathbf{\hat{p}}(j,\sigma_j,\hat{\sigma}_j),\theta_{-k} \mathbf{q}_{\sigma_k}))\label{e41} \\
&\leq \sum_{j,k} \sum_{\sigma_j,\sigma_k,\hat{\sigma}_j,\hat{\sigma}_k}Pr(j,k)Pr(\sigma_j,\sigma_k)\theta_j(\hat{\sigma}_j,\sigma_j)\theta_k(\hat{\sigma}_k,\sigma_k)\\
& \frac{\beta}{\alpha Pr(\sigma_k)}\sum_{l \neq k}\sum_{\sigma_l}Pr(l)Pr(\sigma_l,\sigma_k)\theta_l(\hat{\sigma}_k,\sigma_l)(D^*(\mathbf{\hat{p}}(j,\sigma_j,\hat{\sigma}_j),\mathbf{\hat{p}}(l,\sigma_l,\hat{\sigma}_k))-D^*(\mathbf{\hat{p}}(j,\sigma_j,\hat{\sigma}_j),\mathbf{\hat{p}}(k,\sigma_k,\hat{\sigma}_k)))\label{e42} 
\end{align}

Since $\frac{Pr(\sigma_j,\sigma_k)}{Pr(\sigma_k)}\leq 1$, we obtain (\ref{e44}) from (\ref{e42}). 

\begin{align}
(\ref{e42})\leq & \sum_{j,k} \sum_{\sigma_j,\sigma_k,\hat{\sigma}_j,\hat{\sigma}_k}Pr(j,k)\theta_j(\hat{\sigma}_j,\sigma_j)\theta_k(\hat{\sigma}_k,\sigma_k)\nonumber \\ & \frac{\beta}{\alpha}\sum_{l \neq k}\sum_{\sigma_l}Pr(l)Pr(\sigma_l,\sigma_k)\theta_l(\hat{\sigma}_k,\sigma_l)(D^*(\mathbf{\hat{p}}(j,\sigma_j,\hat{\sigma}_j),\mathbf{\hat{p}}(l,\sigma_l,\hat{\sigma}_k))-D^*(\mathbf{\hat{p}}(j,\sigma_j,\hat{\sigma}_j),\mathbf{\hat{p}}(k,\sigma_k,\hat{\sigma}_k)))\label{e44}  \\
\leq & \sum_{j,k} \sum_{\sigma_j,\sigma_k,\hat{\sigma}_j,\hat{\sigma}_k}Pr(j,k)\theta_j(\hat{\sigma}_j,\sigma_j)\theta_k(\hat{\sigma}_k,\sigma_k)\nonumber \\
& \frac{\beta}{\alpha}\sum_{l \neq k}\sum_{\sigma_l}Pr(l)Pr(\sigma_l,\sigma_k)\theta_l(\hat{\sigma}_k,\sigma_l)|(D^*(\mathbf{\hat{p}}(j,\sigma_j,\hat{\sigma}_j),\mathbf{\hat{p}}(l,\sigma_l,\hat{\sigma}_k))-D^*(\mathbf{\hat{p}}(j,\sigma_j,\hat{\sigma}_j),\mathbf{\hat{p}}(k,\sigma_k,\hat{\sigma}_k)))|\label{e45}
\end{align}

Note that (\ref{e45}) and (\ref{e44}) are identical except for the  value sign.

Then we obtain (\ref{e46}) from (\ref{e45}) since
\begin{align*}
|D^*(x,y)-D^*(x,z)|&\leq (\sqrt{D^*(x,y)}+\sqrt{D^*(x,z)})|\sqrt{D^*(x,y)}-\sqrt{D^*(x,z)})|\\
&\leq 2|\sqrt{D^*(x,y)}-\sqrt{D^*(x,z)})|
\end{align*}
The second inequality follows since $0\leq D^* \leq 1$

\begin{align}
(\ref{e45})\leq & 2 \sum_{j,k} \sum_{\sigma_j,\sigma_k,\hat{\sigma}_j,\hat{\sigma}_k}Pr(j,k)\theta_j(\hat{\sigma}_j,\sigma_j)\theta_k(\hat{\sigma}_k,\sigma_k)\nonumber \\
& \frac{\beta}{\alpha}\sum_{l \neq k}\sum_{\sigma_l}Pr(l)Pr(\sigma_l,\sigma_k)\theta_l(\hat{\sigma}_k,\sigma_l)\nonumber \\
& |(\sqrt{D^*(\mathbf{\hat{p}}(j,\sigma_j,\hat{\sigma}_j),\mathbf{\hat{p}}(l,\sigma_l,\hat{\sigma}_k))}-\sqrt{D^*(\mathbf{\hat{p}}(j,\sigma_j,\hat{\sigma}_j),\mathbf{\hat{p}}(k,\sigma_k,\hat{\sigma}_k))})|\label{e46}
\end{align}

Once we get (\ref{e46}), we can use the fact that $\sqrt{D^*}$ is metric which implies the triangle inequality; (\ref{e47}) follows.

\begin{align}
(\ref{e46})\leq & 2 \sum_{j,k} \sum_{\sigma_j,\sigma_k,\hat{\sigma}_j,\hat{\sigma}_k}Pr(j,k)\theta_j(\hat{\sigma}_j,\sigma_j)\theta_k(\hat{\sigma}_k,\sigma_k)\nonumber \\
& \frac{\beta}{\alpha}\sum_{l \neq k}\sum_{\sigma_l}Pr(l)Pr(\sigma_l,\sigma_k)\theta_l(\hat{\sigma}_k,\sigma_l)(\sqrt{D^*(\mathbf{\hat{p}}(k,\sigma_k,\hat{\sigma}_k),\mathbf{\hat{p}}(l,\sigma_l,\hat{\sigma}_k))})\label{e47}
\end{align}

Note that $\sum_{\sigma_j}\sum_{\hat{\sigma}_j}\theta_j(\hat{\sigma}_j,\sigma_j)=\sum_{\sigma_j}1=m$, also we have $\sum_j Pr(l)=\sum_j Pr(j)=1$, $Pr(k,l)=Pr(j,k)$ then (\ref{e48}) follows.

\begin{align}
(\ref{e47})=& 2 m \frac{\beta}{\alpha } \sum_{l}\sum_{k\neq l} \sum_{\sigma_k,\hat{\sigma}_k}Pr(k,l)\theta_k(\hat{\sigma}_k,\sigma_k)\sum_{\sigma_l}Pr(\sigma_l,\sigma_k)\theta_l(\hat{\sigma}_k,\sigma_l)(\sqrt{D^*(\mathbf{\hat{p}}(k,\sigma_k,\hat{\sigma}_k),\mathbf{\hat{p}}(l,\sigma_l,\hat{\sigma}_k))})\label{e48}\\
=& 2 m \frac{\beta}{\alpha } \sum_{k,l\neq k} \sum_{\sigma_k,\hat{\sigma}_k,\sigma_l}Pr(k,l)\theta_k(\hat{\sigma}_k,\sigma_k)Pr(\sigma_l,\sigma_k)\theta_l(\hat{\sigma}_k,\sigma_l)(\sqrt{D^*(\mathbf{\hat{p}}(k,\sigma_k,\hat{\sigma}_k),\mathbf{\hat{p}}(l,\sigma_l,\hat{\sigma}_k))})\label{e49}\\
= &  2 m \frac{\beta}{\alpha }\times \textit{Inconsistency}
\end{align}

%From ~(\ref{e48}) to ~(\ref{e49}): When $k=l$, $t$ must be $\sigma_l$, so $D^*(\mathbf{\hat{p}}(k,t,\hat{t}),\mathbf{\hat{p}}(l,\sigma_l,\hat{t}))=0$\\

\smallskip

The analysis for the second part

$$\sum_{j,k} \sum_{\sigma_j,\sigma_k,\hat{\sigma}_j,\hat{\sigma}_k}Pr(j,k)Pr(\sigma_j,\sigma_k)\theta_j(\hat{\sigma}_j,\sigma_j)\theta_k(\hat{\sigma}_k,\sigma_k)(D^*(\mathbf{\hat{p}}(j,\sigma_j,\hat{\sigma}_j),\theta_{-k} \mathbf{q}_{\sigma_k})-D^*(\theta_{-j}\mathbf{q}_{\sigma_j},\theta_{-k}\mathbf{q}_{\sigma_k}))$$

is similar, note that $j$ and $k$ are symmetric and $D^*(\cdot,\theta_{-k}\mathbf{q}_{\sigma_k})$ is a convex function. We can use Claim~\ref{claim:equilibrium} and triangle inequality to bound the second part by $2 m \frac{\beta}{\alpha }\times \textit{Inconsistency}$. 

So if we set $2m \frac{\beta}{\alpha}<\frac{1}{2}$, then $\textit{TotalDivergence}(s)-\textit{TotalDivergence}(s_{BP})<\textit{Inconsistency}$, proving the inequality in our main lemma. 

To prove that if the equality in our main lemma holds then $s=s_{BP}$, we first show that 

\begin{claim}\label{claim:main:lemma:eq}
The equality in $\textit{ClassificationScore}(s) \leq \textit{TotalDivergence}(s_{BP})$ holds iff \\ $\textit{Inconsistency}(s)=0$. 
\end{claim}
\begin{proof}
Note that (\ref{eq:main:lemma}) tells us when $\textit{ClassificationScore}(s)=\textit{TotalDivergence}(s_{BP})$, we have $\textit{Diversity}(s)=\textit{TotalDivergence}(s)$ which implies $\textit{Inconsistency}(s)=0$ based on Claim~\ref{claim:inconsistency:zero}.
\end{proof}

Then we will prove 
\begin{claim}
If $\textit{Inconsistency}(s)=0$ then $s=s_{BP}$
\end{claim}
\begin{proof}
Recall in (\ref{eq:linear system}), we have for any $i$,
\begin{align} 
\mathbf{\hat{p}}(i,\sigma_i,\hat{\sigma}_i)&=\arg \max_{\hat{\mathbf{p}}} \E[payment(i,\mathcal{M}+)|\sigma_i]\nonumber \\ \label{eq:linear system 2}
& =  \frac{\alpha \theta_{-i}\mathbf{q}_{\sigma_i} + \beta Pr(j)\sum_{j\neq i}\sum_{\sigma_j}Pr(\sigma_j|\sigma_i)\theta_j(\hat{\sigma}_i,\sigma_j)\mathbf{\hat{p}}(j,\sigma_j,\hat{\sigma}_i)}{\alpha+ \beta Pr(j)\sum_{j\neq i}\sum_{\sigma_j}Pr(\sigma_j|\sigma_i)\theta_j(\hat{\sigma}_i,\sigma_j)}
\end{align}
If $\textit{Inconsistency}(s)=0$, we can see if $\theta_j(\hat{\sigma}_i,\sigma_j)>0$ we must have $\hat{\mathbf{p}}(j,\sigma_j,\hat{\sigma}_i)= \hat{\mathbf{p}}_i(i,\sigma_i,\hat{\sigma}_i)$. So we have $\theta_{-i}\mathbf{q}_{\sigma_i}=\hat{\mathbf{p}}_i(i,\sigma_i,\hat{\sigma}_i)$ for any $i$ since we have $\alpha \theta_{-i}\mathbf{q}_{\sigma_i}=\alpha \hat{\mathbf{p}}_i(i,\sigma_i,\hat{\sigma}_i) $ if we multiply\\ $\alpha+ \beta Pr(j)\sum_{j\neq i}\sum_{\sigma_j}Pr(\sigma_j|\sigma_i)\theta_j(\hat{\sigma}_i,\sigma_j)$ in both sides of equation (\ref{eq:linear system 2}) and combine the fact that $\hat{\mathbf{p}}(j,\sigma_j,\hat{\sigma}_i)= \hat{\mathbf{p}}_i(i,\sigma_i,\hat{\sigma}_i)$. Thus we have $s=s_{BP}$. 
\end{proof}

\subsection{Proof for Claims}\label{section:proof_claims}
{
\renewcommand{\thetheorem}{\ref{claim:average signal strategy}}
\begin{claim}
 Assume that the distribution over all agents' private signals is $\omega\in \Delta_{\Sigma}$, the distribution over all agents' reported signals will be $\bar{\theta}_n\omega$.
\end{claim}
\addtocounter{theorem}{-1}
}
\begin{proof}[Proof for Claim~\ref{claim:average signal strategy}]
The probability of signal $\sigma$ will be $$\sum_i Pr(i)\sum_{\sigma'}\theta_i(\sigma,\sigma')\omega(\sigma')=\frac{1}{n} \sum_i \sum_{\sigma'}\theta_i(\sigma,\sigma')\omega({\sigma'})=\sum_{\sigma'} \bar{\theta}_n(\sigma,\sigma')\omega({\sigma'}) $$
where $Pr(i)$ is the probability agent $i$ is picked. For each agent $i$, we sum the probability agent $i$ receives private signal $\sigma'$ which is $\omega(\sigma')$ times the probability that he reports $\sigma$ given he receives $\sigma'$ which is $\theta_i(\sigma,\sigma')$ over all possible private signal $\sigma'$. 

So the distribution of reported signals is $\bar{\theta}_n\omega$. 
\end{proof}

{ \renewcommand{\thetheorem}{\ref{claim:best prediction}} \begin{claim}
For each agent $i$, if he receives private signal $\sigma_i$, agent $i$ will believe that the expected likelihood of other agents' reported signals is $\theta_{-i}\mathbf{q}_{\sigma_i}$ where $\theta_{-i}=\frac{\sum_{j\neq i}\theta_j}{n-1}$.
\end{claim} \addtocounter{theorem}{-1} }
\begin{proof}[Proof for Claim~\ref{claim:best prediction}]
For each agent $i$, given he receives private signal $\sigma_i$, he will believe the expected likelihood for other agents' private signals is $\mathbf{q}_{\sigma_i}$. Based on Claim~\ref{claim:average signal strategy}, he will believe the expected likelihood for other agents' reported signals is the average signal strategy of other agents' signal strategies times $\mathbf{q}_{\sigma_i}$ which is $\theta_{-i}\mathbf{q}_{\sigma_i}$ where $\theta_{-i}=\frac{\sum_{j\neq i}\theta_j}{n-1}$.
\end{proof}

{ \renewcommand{\thetheorem}{\ref{claim:permutation_matrix}} \begin{claim}
For any transition matrix $\theta_{m\times m}$ where the sum of every column is 1, $\theta$ is a permutation matrix iff for any row of $\theta$, there  at most one non-zero entry. 
\end{claim} \addtocounter{theorem}{-1} }

\begin{proof}[Proof for Claim~\ref{claim:permutation_matrix}]
It is clear that any permutation matrix has exactly one non-zero entry, which is 1, in each row and each column. Thus we only need to prove the direction that if for any row of $\theta$, there is at most one non-zero entry, $\theta$ must be a permutation matrix.

We first prove that there are exactly $m$ non-zero entries in $\theta$: if for any row of $\theta$, there is at most one non-zero entry, we can see $\theta$ has at most $m$ non-zero entries. $\theta$ is a transition matrix where the sum of every column is 1,which implies that $\theta$ has at least $m$ non-zero entries. Thus we proved there are exactly $m$ non-zero entries in $\theta$. 

We have just shown that $\theta$ has exactly $m$ non-zero entries. Since $\theta$ has at most one non-zero entry in each row, $\theta$ must have exactly one non-zero entry in each row. $\theta$ also has at least one non-zero entry in each column since it is a transition matrix, so $\theta$ must have exactly one non-zero entry 1 in each column. Thus $\theta$ has exactly one non-zero entry 1 in each row and each column which implies that $\theta$ is a permutation matrix.   
\end{proof}

{ \renewcommand{\thetheorem}{\ref{claim:sameeq}} \begin{claim}
The \emph{Disagreement Mechanism} has the same equilibria as the \emph{Truthful Mechanism}. 
\end{claim} \addtocounter{theorem}{-1} }
\begin{proof}[Proof for Claim~\ref{claim:sameeq}]
The value of $score_C(r_j,r_k)$ does not depend on agent $i$'s strategy.  The term related to agent $i$'s strategy contained in $score_{\mathcal{M}}$ is $payment_{\mathcal{M}(\alpha,\beta,PS(\cdot,\cdot))}(i,\mathbf{r})$. This implies that agent $i$'s marginal benefit from deviation in $\mathcal{M+}(\alpha,\beta,PS(\cdot,\cdot))$ is the same with its marginal benefit from the same deviation in $\mathcal{M}(\alpha,\beta,PS(\cdot,\cdot))$. 
\end{proof}

{ \renewcommand{\thetheorem}{\ref{claim:cdi}} \begin{claim}
$$\textit{ClassificationScore}=\textit{Diversity}-\textit{Inconsistency}$$
\end{claim} \addtocounter{theorem}{-1} }
\begin{proof}[Proof for Claim~\ref{claim:cdi}]
Based on the definition of \textit{ClassificationScore}, we have
\begin{align}\label{eq:classification score}
 \sum_{\substack{i\\j\neq i}}\sum_{\substack{k\neq i,j\\\sigma_i,\sigma_j,\sigma_k}}  & Pr(i) Pr(\sigma_i)Pr(j,k)Pr(\sigma_j,\sigma_k|\sigma_i) * \\
& \int_{\hat{\sigma}_j,\hat{\mathbf{p}}_j,\hat{\sigma}_k,\hat{\mathbf{p}}_k}Pr_{(\hat{\sigma}_j,\hat{\mathbf{p}}_j)\leftarrow s_j(\sigma_j)}(\hat{\sigma}_j,\hat{\mathbf{p}}_j) Pr_{(\hat{\sigma}_k,\hat{\mathbf{p}}_k)\leftarrow s_k(\sigma_k)}(\hat{\sigma}_k,\hat{\mathbf{p}}_k)  score_C(r_j,r_k)
\end{align}

%where $Pr(i)$ means the probability agent $i$ is picked, $Pr(\sigma_i)$ is the probability agent $i$ receives private signal $\sigma$, $Pr(j,k)$ is the probability that agent $j,k$ are randomly matched with agent $i$, $Pr(\sigma_j,\sigma_k|\sigma_i)$ is the probability that agent $j$ receives private signal $\sigma_j$ and agent $k$ receives private signal $\sigma_k$ condition on the fact that agent $i$ receives private signal $\sigma_i$. 

Now we begin our proof:

\begin{align*}
&  \sum_{\substack{i\\j\neq i}}\sum_{\substack{k\neq i,j\\\sigma_i,\sigma_j,\sigma_k}} Pr(i) Pr(\sigma_i)Pr(j,k)Pr(\sigma_j,\sigma_k|\sigma_i) \int_{\hat{j},\hat{k}}Pr(\hat{j},\hat{k}) score_C(r_j,r_k)\\
=&  \sum_{\substack{i\\j\neq i}}\sum_{\substack{k\neq i,j\\\sigma_j,\sigma_k}} \frac{1}{n(n-1)(n-2)} Pr(\sigma_j,\sigma_k) \int_{\hat{j},\hat{k}} Pr(\hat{j},\hat{k}) score_C(r_j,r_k)  \\
=& \frac{1}{n(n-1)}  \sum_{\substack{j\\k\neq j}}\sum_{\substack{\sigma_j,\sigma_k}} Pr(\sigma_j,\sigma_k) \int_{\hat{j},\hat{k}} Pr(\hat{j},\hat{k}) score_C(r_j,r_k)\\
=& \sum_{\substack{j\\k\neq j}}\sum_{\substack{\sigma_j,\sigma_k}} Pr(j,k) Pr(\sigma_j,\sigma_k) \int_{\hat{j},\hat{k}} Pr(\hat{j},\hat{k}) score_C(r_j,r_k)
\end{align*}

The first equality follows since fix $j,k$, $score_C(r_j,r_k)$ does not depend on $i$ and \\
we also have $\sum_{\sigma_i}Pr(\sigma_i)Pr(\sigma_j,\sigma_k|\sigma_i)=Pr(\sigma_j,\sigma_k)$. 

The second equality follows since for any $(j,k),j\neq k$ pair, there are $n-2$ numbers that are neither $j$ nor $k$ which means $(j,k)$ will repeat $n-2$ times since there are $n-2$ possible $i$.

By definition we can see $\textit{ClassificationScore}=\textit{Diversity}-\textit{Inconsistency}$. 
\end{proof}

%We define $\textit{Diversity}-\textit{Inconsistency}$ as $\textit{ClassificationScore}$, so the expected average $\textit{ClassificationScore}$ is 

%For other equilibria, intuitively, they will be inconsistent or concentrated. So we want to show that the equilibrium that obtains the strictly higher average classification score than any other strategy that is not its permutation is "close" to permutation strategy profile in some sense. Once we proved this, the mechanism $\mathcal{M+}(\alpha,\beta,PS(\cdot,\cdot))$ that satisfies the above two conditions we discussed in the beginning of Section~\ref{section:AQFM} can make truth-telling approximate-quasi-focal. Now we describe this mechanism:
%\yk{Need to define a new quasi-focal} 

{ \renewcommand{\thetheorem}{\ref{claim:per}} \begin{claim}
Any permutation strategy profile has the same \textit{ClassificationScore}, \textit{Diversity}, and \textit{Inconsistency} with truth-telling. 
\end{claim} \addtocounter{theorem}{-1} }
\begin{proof}[Proof for Claim~\ref{claim:per}]
Any permutation strategy profile's report profiles can be seen as a relabeling to truth-telling's report profiles, which implies the claim. 
\end{proof}

{ \renewcommand{\thetheorem}{\ref{claim:welfare}} \begin{claim}
The average agent-welfare in our \emph{Disagreement Mechanism} is $\textit{ClassificationScore}$  
\end{claim} \addtocounter{theorem}{-1} }
\begin{proof}[Proof for Claim~\ref{claim:welfare}]
We only need to prove $\sum_i score_{\mathcal{M}}(i,\mathbf{r})=0$. 
\begin{align*}
\sum_i score_{\mathcal{M}}(i,\mathbf{r})=&\sum_{i\in A} score_{\mathcal{M}}(i,\mathbf{r})+\sum_{i\in B} score_{\mathcal{M}}(i,\mathbf{r})\\
=&\sum_{i\in A}\left( payment_{\mathcal{M}(\alpha,\beta,PS(\cdot,\cdot))}(i,\mathbf{r})-\frac{1}{|A|}\sum_{i\in B} payment_{\mathcal{M}(\alpha,\beta,PS(\cdot,\cdot))}(i,\mathbf{r})\right)\\
&+\sum_{i\in B}\left(
payment_{\mathcal{M}(\alpha,\beta,PS(\cdot,\cdot))}(i,\mathbf{r})-\frac{1}{|B|}\sum_{i\in A} payment_{\mathcal{M}(\alpha,\beta,PS(\cdot,\cdot))}(i,\mathbf{r})\right)=0\\
\end{align*} 
\end{proof}

{ \renewcommand{\thetheorem}{\ref{claim:inconsistency:zero}} \begin{claim}
For any strategy profile $s$, $\textit{Diversity}(s)=\textit{TotalDivergence}(s)$ $\Leftrightarrow$ $\textit{Inconsistency}(s)=0$
\end{claim} \addtocounter{theorem}{-1} }
\begin{proof}[Proof for Claim~\ref{claim:inconsistency:zero}]
Note that 
\begin{align*}
&\textit{TotalDivergence}(s)-\textit{Diversity}(s)\\
=&\sum_{\substack{j\\k\neq j}}\sum_{\substack{\sigma_j,\sigma_k}} Pr(j,k)Pr(\sigma_j,\sigma_k) \int_{\hat{j},\hat{k}} Pr(\hat{j},\hat{k})\delta(\hat{\sigma}_j= \hat{\sigma}_k)D^*(\mathbf{\hat{p}}_j,\mathbf{\hat{p}}_k)
\end{align*}

while $$\textit{Inconsistency}(s)=\sum_{\substack{j\\k\neq j}}\sum_{\substack{\sigma_j,\sigma_k}} Pr(j,k)Pr(\sigma_j,\sigma_k) \int_{\hat{j},\hat{k}} Pr(\hat{j},\hat{k})\delta(\hat{\sigma}_j= \hat{\sigma}_k)\sqrt{D^*(\mathbf{\hat{p}}_j,\mathbf{\hat{p}}_k)}$$ 

Because each part in $\textit{TotalDivergence}(s)-\textit{Diversity}(s)$ is non-negative,  $\textit{TotalDivergence}(s)-\textit{Diversity}(s)=0$ will imply $Pr(\sigma_j,\sigma_k) \int_{\hat{j},\hat{k}} Pr(\hat{j},\hat{k})\delta(\hat{\sigma}_j= \hat{\sigma}_k)D^*(\mathbf{\hat{p}}_j,\mathbf{\hat{p}}_k)=0$.

So we have $Pr(\hat{j},\hat{k})\delta(\hat{\sigma}_j= \hat{\sigma}_k)=0$ or $D^*(\mathbf{\hat{p}}_j,\mathbf{\hat{p}}_k)=0$ which implies $Pr(\sigma_j,\sigma_k) \int_{\hat{j},\hat{k}} Pr(\hat{j},\hat{k})\delta(\hat{\sigma}_j= \hat{\sigma}_k)\sqrt{D^*(\mathbf{\hat{p}}_j,\mathbf{\hat{p}}_k)}=0$. The proof for another direction is similar. 
\end{proof}

\else 
%\linepenalty=10000 

The proof is given in the full version; here we just give a few highlights.  We must show that permutation strategies will have strictly higher agent welfare than any other symmetric equilibrium and if the number of agents are sufficient large, the equilibria with a higher agent welfare or even an agent welfare ``close'' to truth-telling must be ``close'' to a permutation strategy profile. We first show that the agent welfare of our \textit{Disagreement Mechanism} is $\textit{Diversity}-\textit{Inconsistency}$, which follows by a straightforward computation. It remains to show that $\textit{Diversity}-\textit{Inconsistency}$ has the aforementioned properties.

 We call a strategy profile a \emph{best prediction strategy profile} if for any $i$, agent $i$ reports a prediction that maximizes his prediction score. Based on Claim~\ref{claim:best prediction}, agent $i$ will report $\theta_{-i} \mathbf{q}_{\sigma_i}$ given $\sigma_i$ is his private signal and recall that $\theta_{-i}=\frac{\sum_{j\neq i}\theta_i}{n-1}$ where $(\theta_1,\theta_2,....,\theta_n)$ is the signal strategy. We call this strategy profile a \emph{symmetric} best prediction strategy profile if there exists a signal strategy $\theta$ such that $\theta_i=\theta$ for any $i$. Based on the definition of permutation strategy profile, it is clear that any permutation strategy profile is a symmetric \emph{best prediction} strategy profile.

Consider two agents who report different signals. If they use a permutation strategy profile $\pi$ then their predictions will be $\pi\mathbf{q}_{\sigma},\pi\mathbf{q}_{\sigma'}$ given their private signals are $\sigma\neq \sigma'$. If they use a symmetric best prediction strategy, then their reported predictions will be $\theta \mathbf{q}_{\sigma},\theta \mathbf{q}_{\sigma'}$. In the first case, the Hellinger divergence between the two agents' reported predictions is $D^*(\pi\mathbf{q}_{\sigma},\pi\mathbf{q}_{\sigma'})=D^*(\mathbf{q}_{\sigma},\mathbf{q}_{\sigma'})$ while in the second case, the Hellinger divergence between the two agents' reported predictions is $D^*(\theta \mathbf{q}_{\sigma},\theta \mathbf{q}_{\sigma'})\leq D^*(\mathbf{q}_{\sigma},\mathbf{q}_{\sigma'})=D^*(\pi\mathbf{q}_{\sigma},\pi\mathbf{q}_{\sigma'})$. The inequality follows from the information monotonicity of Hellinger divergence. Thus, the two agents' predictions in the second case is ``closer'' than those in the first case. So a permutation strategy profile is more diverse than any other symmetric best prediction strategy, and additionally have no inconsistency.

However, the biggest challenge is that there exists equilibria that are not \textit{best prediction strategy profiles}. To deal with this challenge, we map each equilibrium $s^*$ to a strategy profile $s^*_{BP}$ that belongs to \textit{best prediction strategy profiles}. The technical heart of the proof bounds the $Diversity - Inconsistency$ score of an equilibrium strategy profile $s^*$ by the $Diversity - Inconsistency$ score of its corresponding best prediction strategy profile $s^*_{BP}$. Once we finish this, we can bound $Diversity - Inconsistency$ score of any equilibrium strategy profile by $Diversity - Inconsistency$ score of permutation strategy profiles and finish the proof. 

%%%REMOVED FOR SPACE??
%Arriving at this bound requires a non-trivial understanding of the structure of the equilibria, and especially the relation between the different agents' prediction reports in any equilibria.  Given a strategy profile for the reports, we obtain a system of linear equations relating the prediction reports. Achieving this bound also requires the delicate use of the triangle inequality applied to $\sqrt{D^*}$.

In the more complicated assymetric case, the difficulty is that even if agents play best prediction strategy profiles, we cannot use information monotonicity to prove permutation strategy profiles gain the strictly highest $Diversity - Inconsistency$ score. However, if the number of agents is large enough, we will see any strategy profile that belongs to \textit{best prediction strategy profiles} family is ``almost symmetric''. This ``almost symmetric'' result fit the our proof in asymmetric case into above framework. 

Finally, we show that equilibrium that have the $Diversity - Inconsistency$ score close to that of truth-telling, must be close to a permutation equilibrium.

\fi

\section{Conclusion}\label{section:conclusion}
We have shown that our  \emph{Disagreement Mechanism}  promotes truth-telling by 1) having truth-telling as a Bayesian Nash equilibrium; 2) having no other symmetric equilibrium with agent welfare more than the truth-telling equilibrium; 3) having the agent welfare of any equilibrium approach that of truth-telling as the number of agents increases; 4) requiring that  any equilibrium with agent welfare close to that of truth-telling must be close to a permutation strategy.  

We have argued, that our mechanism is near optimal in the sense that no truthful mechanism without knowledge of the common prior can avoid having permutation equilibrium with high agent welfare, and, in our mechanism, any equilibrium with agent welfare even close to that of truth-telling must be close to a permutation equilibrium.  

Permutation equilibria are intuitively unnatural and risky as they require extreme coordination amongst the agents.  We believe they are very unlikely to occur in practice.   
Additionally, any asymmetric equilibrium also seem unlikely, especially as the number of agents increases because a) such deviations help less as the number of agents' increases; b) implementing such deviations will become increasingly difficult as the number of agents increases.  
Thus our results about symmetric equilibria and equilibrium in general are quite strong, despite the impossibility result.

In addition to the above results, our work has several contributions in the techniques employed:

(1) In a common prior setting, agents with the same private information cannot agree to disagree. Thus, agents with the same signal should report similarly. Our \emph{Disagreement Mechanism} encourages not only agents with the same private information to agree, but also agents with different private information to disagree. We do this by employing tools from information theory, namely \emph{Information Monotonicity}. Despite their natural and powerful application, to our knowledge, this is the first time such tools have been explicitly employed in the peer prediction literature. 

\ifnum\fullversion=0
(2) We created a framework for understanding the space of equilibrium which was  integral to our results.  

We hope that both the information theory tools and the new understanding of equilibrium are useful in future work.
%We hope these insights prove useful in future work.
\else
(2)  Additionally, we created a framework for understanding the space of equilibrium which was  integral to our results.

We hope that both the information theory tools and the new understanding of equilibrium introduced in this work will continue to provide useful insights for designing and analyzing future mechanisms in peer-prediction and related settings.
\fi

\newpage 

\ifnum\fullversion=1

%\section{Modified Decomposable Payment Scheme}
%\input{additional_explanation}

\else

\fi

\bibliographystyle{plainnat}
\bibliography{peer,refs}

\newpage

\section{Appendix}
To understand the strictness condition more in Lemma~\ref{lem:im}, we give an example where the strictness condition is not satisfied:
\begin{example}
$\mathbf{p}=(0.1\ 0.2\ 0.7)$, $\mathbf{q}=(0.2\ 0.4\   0.4)$, $\theta=\left(\begin{array}{ccc}
0.3 & 0.6 & 0 \\
0.7 & 0.4 & 0\\
0 & 0 & 1 \\
\end{array} \right) $.

We show by case analysis that we cannot find $\sigma, \sigma',\sigma''$ such that $\theta(\sigma,\sigma') \mathbf{p}(\sigma')>0$, $\theta(\sigma,\sigma'') \mathbf{p}(\sigma'')>0$ and  $\frac{\mathbf{p}(\sigma'')}{\mathbf{p}(\sigma')}\neq \frac{\mathbf{q}(\sigma'')}{\mathbf{q}(\sigma')}$.

First note that because  $\frac{\mathbf{p}(\sigma'')}{\mathbf{p}(\sigma')}\neq \frac{\mathbf{q}(\sigma'')}{\mathbf{q}(\sigma')}$, it cannot be that $\sigma' = \sigma''$, nor can it be the case that $\sigma',\sigma'' \in \{1, 2\}$ because  $\frac{\mathbf{p}(1)}{\mathbf{p}(2)}=\frac{\mathbf{q}(1)}{\mathbf{q}(2)}$ and $\frac{\mathbf{p}(2)}{\mathbf{p}(1)}=\frac{\mathbf{q}(2)}{\mathbf{q}(1)}$.  Thus it must be that either $\sigma' \in \{1, 2\}$ and $\sigma'' = 3$ or  $\sigma' = 3 $ and $\sigma'' \in \{1, 2\}$.  Because these are symmetric, we consider the first case.  

Because  $\theta(\sigma,\sigma') \mathbf{p}(\sigma')>0$ it must be that $\sigma \in \{1, 2\}$, but because $\theta(\sigma,\sigma'') \mathbf{p}(\sigma'')>0$, it must be that $\sigma = 3$.  So no assignment of  $\sigma, \sigma',\sigma''$  is possible. 
 
Thus, the strictness condition is not satisfied. By simple calculations, we have $\theta\mathbf{p}=(0.15\ 0.15\ 0.7)$, $\theta\mathbf{q}=(0.3\ 0.3\ 0.4)$. By some algebraic calculations, we have $D_f(\mathbf{p},\mathbf{q})= D_f(\theta \mathbf{p},\theta \mathbf{q})$ for any function $f$.
\end{example}

\end{document}